\documentclass[11pt]{article}
\usepackage[letterpaper,margin=1in]{geometry}
\usepackage{xcolor,graphicx,amssymb,amsthm,amsmath,hyperref,url,cite,mathrsfs,stmaryrd,thm-restate,authblk}
\usepackage[dvipsnames]{xcolor}
\counterwithin{figure}{section}
\usepackage{thm-restate}
\usepackage[utf8]{inputenc}
\usepackage[T1]{fontenc}
\usepackage{datetime}
\graphicspath{{./figures/}}

\newtheorem{theorem}{Theorem}[section]
\newtheorem{lemma}[theorem]{Lemma}
\newtheorem{corollary}[theorem]{Corollary}
\newtheorem{proposition}[theorem]{Proposition}
\newtheorem{observation}{Observation}

\newcommand{\emphdef}[1]{\textbf{#1}}
\newcommand{\algo}[2]{\vspace{1em}\noindent\textbf{#1.}\hspace{1em}#2\vspace{1em}}
\newcommand{\algoref}[1]{{\normalfont\textbf{#1}}}

\begin{document}


\title{Computing the Intrinsic Delaunay Triangulation of a Closed Polyhedral Surface} 

\author{Lo{\"i}c Dubois}
\affil{Notre Dame, USA\footnote{This work was done while the author was working at LIGM, CNRS, Université Gustave Eiffel, F-77454 Marne-la-Vallée, France.}}
\date{2026}

\maketitle

\begin{abstract}
Every surface that is intrinsically polyhedral can be represented by a portalgon: a collection of polygons in the Euclidean plane with some pairs of equally long edges abstractly identified. While this representation is arguably simpler than meshes (flat polygons in $\mathbb R^3$ forming a surface), it has unbounded \emph{happiness}: a shortest path in the surface may visit the same polygon arbitrarily many times. This pathological behavior is an obstacle towards efficient algorithms. On the other hand, Löffler, Ophelders, Staals, and Silveira~[SoCG~2023] recently proved that the (intrinsic) Delaunay triangulations have bounded happiness.

In this paper, given a closed polyhedral surface $S$, represented by a triangular portalgon $T$, we provide an algorithm to compute the Delaunay triangulation of $S$ whose vertices are the singularities of $S$ (the points whose surrounding angle is distinct from $2\pi$). The time complexity of our algorithm is polynomial in the number of triangles and in the logarithm of the aspect ratio $r$ of $T$. Within our model of computation, we show that the dependency in $\log r$ is unavoidable. Our algorithm can be used to pre-process a triangular portalgon before computing shortest paths on its surface, and to determine whether the surfaces of two triangular portalgons are isometric. 
\end{abstract}


\section{Introduction}\label{sec:intro}

In one of its simplest forms a \emph{triangulation} is a finite collection of disjoint triangles in the Euclidean plane, together with a partial matching of the sides of the triangles such that any two matched sides have the same length (Figure~\ref{fig:portalgon}). This simple representation appears under different names in the literature (intrinsic triangulation~\cite{sharp2021geometry,sharp2019navigating}, portalgon~\cite{portalgons}). Cutting out the triangles from the plane and identifying the matched sides isometrically, respecting the orientations of the triangles, provides a (compact, orientable) \emph{polyhedral} surface. This surface is \emph{closed} if, in addition, it is connected and without boundary. 

In this paper we consider the \emph{Delaunay} triangulation of a closed polyhedral surface whose vertex set consists of the singularities (the points surrounded by an angle distinct from $2 \pi$) of the surface (except for flat tori, see below). It is generically unique. Our main contribution is an algorithm to compute it from any triangulation of the surface, whose time complexity is polynomial in the number of triangles and in the logarithm of the aspect ratio of the input triangulation. Our second contribution is a lower bound showing that the dependency in the logarithm of the aspect ratio is unavoidable in our model of computation.

Before describing our contributions in more detail, we discuss related works.

\subsection{Related works}

Polyhedral surfaces can also be obtained from \emph{meshes}, flat triangles in $\mathbb R^3$ glued along their edges. Moreover, every mesh defines a triangulation of its surface. Yet triangulations are more general than meshes: most triangulations cannot be obtained from a mesh. Some recent algorithms advantageously operate on triangulations of polyhedral surfaces without reference to a mesh~\cite{sharp2020you,takayama2022compatible,liu2023surface}. In this context the adjective “intrinsic” is sometimes placed before the name “triangulation” to make the distinction with the particular triangulations arising from a mesh. In the mathematical community, a prominent example is that of a translation surface~\cite{masur2006ergodic,zorich2006flat,hubert2006introduction,gutkin2000affine}, which arises naturally in the study of billiards in rational polygons.

Triangulations are so general that not all of them are suitable for computation, compared to meshes. Prominently, a fundamental problem on polyhedral surfaces is to compute the distance or a report a shortest path between two points. On a mesh, shortest paths can be computed in time polynomial in the number of triangles. For example, an algorithm of Mitchell, Mount, and Papadimitriou~\cite{mitchell1987discrete} propagates waves along the surface, starting from the source, in a discrete manner. See also Chen and Han~\cite{chen1996shortest}. On a generic triangulation however (not arising from a mesh), the number of times a shortest path visits a triangle is not bounded by any function of the number of triangles, as noted for example almost 20 years ago by Erickson~\cite{bworld}. Recently, Löffler, Ophelders, Staals, and Silveira~\cite{portalgons} coined the term \emph{happiness} of a triangulation, for the maximum number of times a shortest path visits a triangle. They adapted the single-source shortest paths algorithm of Mitchell, Mount, and Papadimitriou~\cite{mitchell1987discrete} from meshes to triangulations, whose time complexity now depends on the happiness of the triangulation (it is more efficient on triangulations of \emph{low} happiness). 

This raises the problem of replacing any given triangulation by another triangulation of the same surface whose happiness is “low”. Among the many remeshing algorithms~\cite{heckbert1997survey,shewchuk1997delaunay,ruppert1995delaunay,shewchuk2002delaunay}, only few have been ported to the general context of intrinsic triangulations~\cite{sharp2021geometry}, and the only solution we are aware of that compares to our main result (Theorem~\ref{thm:main result} below), by Löffler, Ophelders, Staals, and Silveira~\cite[Section~5]{portalgons}, is restricted to particular inputs whose surfaces are all homeomorphic to an annulus; we will use it as a black box. Importantly, the same authors also showed that Delaunay triangulations have bounded happiness~\cite[Section~4.2]{portalgons}. 

Delaunay triangulations are classical objects of computational geometry~\cite{de2000computational,fortune2017voronoi,aurenhammer2013voronoi}, closely related to shortest paths. While mostly known in the plane, they generalize to closed polyhedral surfaces, see for example the depiction of Bobenko and Springborn~\cite{bobenko2007discrete}. To compute a Delaunay triangulation from an arbitrary intrinsic triangulation there are, to our knowledge, only two approaches, and neither compares to our main result (Theorem~\ref{thm:main result}). One approach computes a Voronoi diagram with a suitably adapted multiple-source shortest path algorithm, and then derives from it a Delaunay tessellation, see for example Mount~\cite{mount1985voronoi}, Liu, Chen, and Tang~\cite{liu2010construction}, and Liu, Xu, Fan, and He~\cite{liu2015efficient}. Another approach starts from an initial triangulation and flips its edges until it reaches a Delaunay triangulation; it was proved to terminate by Indermitte, Liebling, Troyanov, and Cl{\'e}men{\c{c}}on~\cite{bobenko2007discrete,indermitte2001voronoi}. 

\subsection{Our results}

In order to state our results precisely, it now matters to make the distinction between a triangulation and the data structure representing it, and to allow for more general polygons than triangles. Following Löffler, Ophelders, Staals, and Silveira~\cite{portalgons}, we call \emph{portalgon} the collection $T$ of polygons in the Euclidean plane and the partial matching of their sides. We denote by $\mathcal S(T)$ the associated polyhedral surface. We say that the portalgon $T$ is \emph{triangular} if all polygons are triangles. The sides of the polygons, once identified, constitute a graph $T^1$ embedded on $\mathcal S(T)$ (the polygons themselves correspond to the faces of $T^1$): it is this graph $T^1$ that we call \emph{triangulation} if $T$ is triangular, and we call $T^1$ a \emph{tessellation} in general.

We consider, on a closed polyhedral surface $S$, the unique Delaunay tessellation $\mathcal D$ of $S$ whose vertices are exactly the singularities of $S$, with a single very special exception: if $S$ has no singularity, then $S$ is a flat torus and we let $\mathcal D$ be any of the Delaunay tessellations of $S$ that have exactly one vertex, for one can be mapped to the other via an orientation-preserving isometry of $S$ anyway. In any case, we say that $\mathcal D$ is \emph{the} Delaunay tessellation of $S$, in a slight abuse.\footnote{Given a set $V$ of points of the surface, finite, non-empty, and containing all the singularities, our results easily extend to Delaunay triangulations whose vertex set is $V$, but this is incidental to us.} It is “generically” a triangulation, but not always. If not, then triangulating the faces of $\mathcal D$ along any vertex-to-vertex arcs provides a Delaunay triangulation. The \emph{aspect ratio} of a triangular portalgon $T$ is the maximum side length of a triangle of $T$ divided by the smallest height of a triangle of $T$ (possibly another triangle). Our main contribution is:

\begin{restatable}{theorem}{maintheoremportalgon}\label{thm:main result}
Let $T$ be a portalgon of $n$ triangles, of aspect ratio $r$, whose surface $\mathcal S(T)$ is closed. One can compute the portalgon of the Delaunay tessellation of $\mathcal S(T)$ in $O(n^3 \log^2(n) \cdot \log^4(r))$ time.
\end{restatable}

As already mentioned, the only two methods we are aware of for computing a Delaunay tessellation from an arbitrary triangulation are the flip algorithm and the computation of the dual Voronoi diagram. The time complexities of these algorithms are not bounded by any polynomial in $n$ and $\log(r)$. 

Applications of Theorem~\ref{thm:main result} are provided in Appendix~\ref{app:applications}. Briefly, on the portalgon returned by Theorem~\ref{thm:main result}, shortest paths can be computed in $O(n^2 \log^{O(1)} n)$ time. And Theorem~\ref{thm:main result} enables to test whether the surfaces of two given portalgons are isometric, simply by computing and comparing the portalgons of the associated Delaunay tessellations.

We analyze our algorithms within the real RAM model of computation described by Erickson, van der Hoog, and Miltzow~\cite{erickson2022smoothing}. It is an extension of the standard integer word RAM, with an additional memory array storing reals, and with additional instructions. On such a machine, we represent each polygon of a portalgon $T$ by the list of its vertices, and each vertex is by its two coordinates, stored in the memory array dedicated to reals. So displacing (translating or rotating) the polygons in the plane provides different representations of $T$. When modifying a portalgon $T$, we actually modify our representation of $T$, using elementary operations that are easily seen to be achievable by a real RAM.

Within this model of computation, our second contribution is a lower bound that backs our main result, Theorem~\ref{thm:main result}, by showing that the polynomial dependency in the logarithm of the aspect ratio is unavoidable:

\begin{restatable}{theorem}{maintheoremlowerbound}\label{thm:main lower bound}
Let $c \in (0,1)$. There are a flat torus $S$, and for every $x \in (1,\infty)$, a representation of a portalgon $T_x$, with two triangles, whose aspect ratio is $O(x^2)$, whose surface is $S$, that satisfy the following. There is no real RAM algorithm computing a representation of the portalgon of the Delaunay tessellation of $S$ from $T_x$ in $O((\log x)^c)$ time.
\end{restatable}

Altogether Theorem~\ref{thm:main result} and Theorem~\ref{thm:main lower bound} show that, within our model of computation, the complexity of computing the Delaunay tessellation from an arbitrary triangulation of a (closed, orientable) polyhedral surface is polynomial in the number of triangles and in the logarithm of the aspect ratio of the input triangulation.


\subsection{Overview of the proof of Theorem~\ref{thm:main result}}

The proof of Theorem~\ref{thm:main lower bound} is deferred to Appendix~\ref{app:lower bound}. The rest of the paper is dedicated to the proof of Theorem~\ref{thm:main result}, of which we now provide an overview.

We introduce a slight variation of happiness, more suitable to our needs, which we call \emph{segment-happiness}. To prove Theorem~\ref{thm:main result}, the crux of the matter is to replace the input triangular portalgon by another triangular portalgon of the same surface, whose segment-happiness is “low”. For this purpose, our approach is to first focus on portalgons $T$ whose surface $\mathcal S(T)$ is \emph{flat}: the interior of $\mathcal S(T)$ has no singularity. Note that here we allow $\mathcal S(T)$ to have boundary, and this boundary may have singularities. The \emph{systole} of $\mathcal S(T)$ is the smallest length of a non-contractible geodesic closed curve in $\mathcal S(T)$. Our key technical result is:

\begin{restatable}{proposition}{corealgo}\label{T:core algorithm}
Let $T$ be a portalgon of $n$ triangles, whose sides are all smaller than $L > 0$.  Assume that $\mathcal S(T)$ is flat. Let $s > 0$ be smaller than the systole of $\mathcal S(T)$. One can compute in $O(n \log^2(n) \cdot \log^2 (2+L/s))$ time a portalgon of $O(n \cdot \log(2+L/s))$ triangles, whose surface is isometric to that of $T$, and whose segment-happiness is $O(\log(n) \cdot \log^2 (2+L/s))$.
\end{restatable}


Sections~\ref{sec:bifaces}--\ref{sec:analysis} are devoted to the proof of Proposition~\ref{T:core algorithm}. In Section~\ref{sec:bifaces} we focus on particular portalgons, whose surfaces are all homeomorphic to an annulus; the definitions and results of this section are used by the algorithm of Proposition~\ref{T:core algorithm}. In Section~\ref{sec:algo} we describe the algorithm for Proposition~\ref{T:core algorithm}. It is a finely tuned combination of elementary operations such as inserting and deleting edges and vertices in graphs. While the algorithm itself is relatively simple, its analysis is more involved, and is sketched in Section~\ref{sec:analysis}. In this section, we first provide a combinatorial analysis, and then we prepare for the geometric analysis by introducing a new parameter on the simple geodesic paths $e$ of a flat surface, \emph{enclosure}, possibly of independent interest. Informally, $e$ is enclosed when a short non-contractible loop can be attached to a point of $e$ not too close to the endpoints of $e$. We then use enclosure to analyze the algorithm from a geometric point of view, proving Proposition~\ref{T:core algorithm}.

In Appendix~\ref{app:extension} we extend Proposition~\ref{T:core algorithm} from flat surfaces to surfaces having singularities in their interior, essentially by cutting out caps around these singularities. To get a cleaner result, we also replace $2+L/s$ by the aspect ratio of $T$, and we replace segment-happiness by happiness, obtaining:

\begin{restatable}{proposition}{thmimproving}\label{thm:improving}
Let $T$ be a portalgon of $n$ triangles, of aspect ratio~$r$. One can compute in $O(n \log^2(n) \cdot \log^2(r))$ time a portalgon of $O(n \cdot \log(r))$ triangles, whose surface is $\mathcal S(T)$, and whose happiness is $O(n \log(n) \cdot \log^2(r))$.
\end{restatable}

We have not discussed Delaunay tessellations yet. Still, we are almost ready to prove Theorem~\ref{thm:main result}. Indeed, once we have a portalgon of low happiness, we can compute shortest paths on the surface. And, as already mentioned, shortest path algorithms classically extend to construct Voronoi diagrams and then Delaunay tessellations. Formally:

\begin{restatable}{proposition}{thmisometry}\label{thm:isometry}
Let $T$ be a portalgon of $n$ triangles, of happiness $h$, such that $\mathcal S(T)$ is closed. One can compute the portalgon of the Delaunay tessellation of $\mathcal S(T)$ in $O(n^2h \log(nh))$ time.
\end{restatable}

We could not find a statement equivalent to Proposition~\ref{thm:isometry} in the literature, so we provide a proof in Appendix~\ref{A:Delaunay} for completeness. We insist that the proof of Proposition~\ref{thm:isometry} is incidental to us, and Proposition~\ref{thm:isometry} is not surprising at all. Our contribution is really the proof of Proposition~\ref{thm:improving}. Theorem~\ref{thm:main result} is immediate from Proposition~\ref{thm:improving} and Proposition~\ref{thm:isometry}:

\begin{proof}[Proof of Theorem~\ref{thm:main result}]
Proposition~\ref{thm:improving} computes in $O(n \log^2(n) \cdot \log^2(r))$ time a portalgon $T'$ of $O(n \cdot \log(r))$ triangles, whose happiness is $O(n \log(n) \cdot \log^2(r))$. Proposition~\ref{thm:isometry} then computes the portalgon of the Delaunay tessellation from $T'$ in $O(n^3 \log^2(n)\cdot \log^4(r))$ time.
\end{proof}


\section{Preliminaries}\label{sec:preliminaries}

We use without review standard notions of graph theory and low dimensional topology and geometry, referring to textbooks for details~\cite{diestel,armstrong2013basic,s-ctcgt-93,do1992riemannian}. We only mention that on a surface $S$, a path  $p: [0,1] \to S$ is \emphdef{simple} if its restriction to the interval $(0,1)$ is injective, in which case the image of $(0,1)$ by $p$ is the \emphdef{relative interior} of $p$. We denote by $\ell(p)$ the length of a geodesic path $p$. Throughout the paper, logarithms are in base two.

The definition of Delaunay tessellation given by Bobenko and Springborn~\cite[Section~2]{bobenko2007discrete} is not used in the core of the paper, but only in in Appendix~\ref{A:Delaunay} for proving Proposition~\ref{thm:isometry}. We collect details on this definition in Appendix~\ref{app:delaunay and voronoi defs} for completeness.

\subsection{Portalgons, tessellations, and polyhedral surfaces}

A \emphdef{portalgon} $T$ is a disjoint collection of oriented polygons in the Euclidean plane, together with a partial matching of the sides of the polygons such that any two matched sides have the same length. It is \emphdef{triangular} if all polygons are triangles. See Figure~\ref{fig:portalgon}. Any subset of the polygons defines a \emphdef{sub-portalgon} $T'$ of $T$: two sides of polygons are matched in $T'$ if and only if they are matched in $T$. In a portalgon $T$, identifying the matched sides, isometrically, and respecting the orientations of the polygons, provides \emphdef{the surface of} $T$, denoted $\mathcal S(T)$; it is a 2-dimensional Riemannian manifold whose metric may have singularities. The sides of the polygons of $T$ correspond to a graph $T^1$ embedded on $\mathcal S(T)$, the \emphdef{1-skeleton} of $T$.

A \emphdef{polyhedral surface} is any Riemannian manifold $S$ (possibly with singularities) isometric to the surface of a portalgon. And when we say that a portalgon $T$ is \emphdef{a portalgon of} $S$, we implicitly fix an isometry between $\mathcal S(T)$ and $S$. A \emphdef{tessellation} of $S$ is any 1-skeleton of a portalgon of $S$, it is a \emphdef{triangulation} if the portalgon is triangular.

\begin{figure}[ht]
    \centering
    \includegraphics[width=0.6\linewidth]{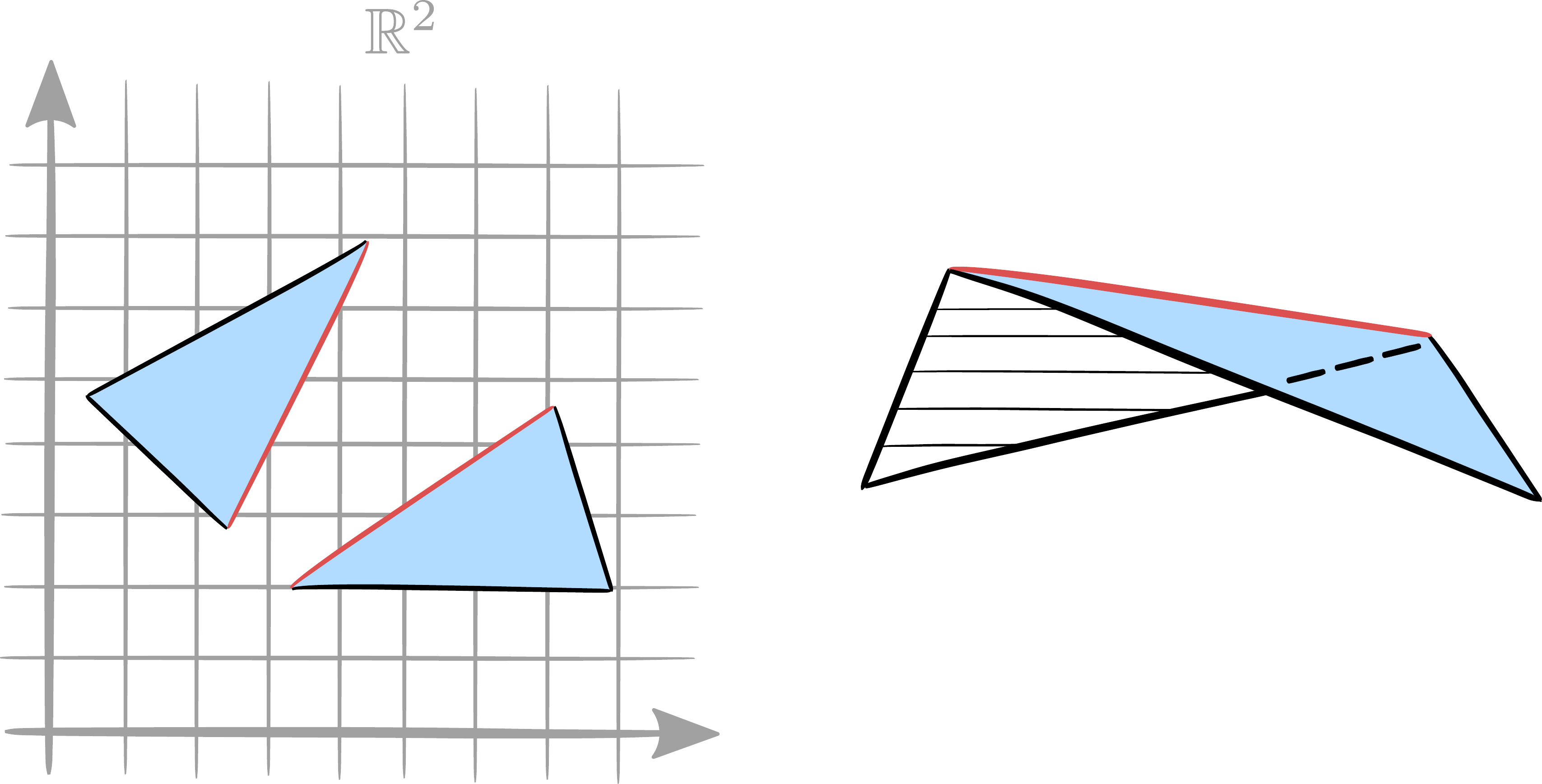}
    \caption{(Left) A triangular portalgon $T$: two triangles in the Euclidean plane, with two sides matched in red. (Right) The surface $\mathcal S(T)$, and the 1-skeleton $T^1$.}
    \label{fig:portalgon}
\end{figure}


Consider a polyhedral surface $S$, a triangulation $T^1$ of $S$, a vertex $x$ of $T^1$, and the sum $a$ of the angles of faces of $T^1$ around $x$. The point $x$ is a \emphdef{singularity} if $x$ lies in the boundary of $S$ and $a \neq \pi$, or if $x$ lies in the interior of $S$ and $a \neq 2 \pi$. Every other point of $S$ is \emphdef{flat}. This does not depend on any particular triangulation of $S$. A surface $S$ is \emphdef{flat} if its interior has no singularity (although its boundary may have singularities). The closed flat surfaces are called \emphdef{flat tori}.

\subsection{Aspect ratio, systole, happiness, and segment-happiness}

The \emphdef{aspect ratio} of a triangular portalgon $T$ is the maximum side length of a triangle of $T$ divided by the smallest height of a triangle of $T$ (possibly another triangle). Note that the aspect ratio is always greater than or equal to $2/\sqrt 3 > 1$, because the maximum side length of a triangle is always greater than or equal to $2/\sqrt 3$ times its smallest height.

The \emphdef{systole} of a polyhedral surface $S$ is the smallest length of a non-contractible geodesic closed curve in $S$, except in the particular case where every closed curve in $S$ is contractible, in which case the systole is $\infty$. The important thing is that for every positive real $s$ smaller than the systole of $S$, any non-contractible closed curve in $S$ is longer than $s$.

The \emphdef{happiness} of a portalgon $T$ is the maximum number of times a shortest path in $\mathcal S(T)$ visits the image of a polygon of $T$, maximized over all the shortest paths of $\mathcal S(T)$ and all the polygons of $T$ (see~\cite[Section~3]{portalgons}). We introduce a variation, more suitable to our needs. In a polyhedral surface $S$, a \emphdef{segment} is a simple geodesic path $e$ whose relative interior is disjoint from any singularity of $S$. The \emphdef{segment-happiness} of $e$ in $S$, denoted $h_{S}(e)$, is the maximum number of intersections between $e$ and a shortest path of $S$, maximized over all the shortest paths of $S$. The \emphdef{segment-happiness} of a portalgon $T$ is then the maximum segment-happiness $h_{\mathcal S(T)}(e)$, maximized over the edges $e$ of its 1-skeleton $T^1$. A priori, the segment-happiness of a portalgon $T$ differs from the happiness of $T$. Indeed a path in $\mathcal S(T)$ may visit many times a face of $T^1$ without intersecting any edge of $T^1$ more than once, if the face has high degree. However, if $T$ is triangular, then the happiness and the segment-happiness of $T$ do not differ by more than a constant factor.


\section{Tubes and bifaces}\label{sec:bifaces}

In this section we focus on particular triangular portalgons. See Figure~\ref{fig:bifaces}. A \emphdef{tube} is a triangular portalgon $X$ whose surface $\mathcal S(X)$ is homeomorphic to an annulus and has no singularity in its interior, and whose 1-skeleton $X^1$ has exactly one vertex on each boundary component of $\mathcal S(X)$. Among tubes, a \emphdef{biface} is a portalgon $B$ of two triangles whose respective sides $s_0,s_1,s_2$ and $s_0',s_1',s_2'$, in order (clockwise say), are such that $s_0$ is matched with $s_0'$ and $s_1$ is matched with $s_1'$. Its 1-skeleton $B^1$ has four edges: two loop edges forming the two boundary components of $\mathcal S(B)$, which we call \emphdef{boundary edges}, and two edges whose relative interiors are included in the interior of $\mathcal S(B)$, which we call \emphdef{interior edges}. 

\begin{figure}[ht]
    \centering
    \includegraphics[width=0.7\linewidth]{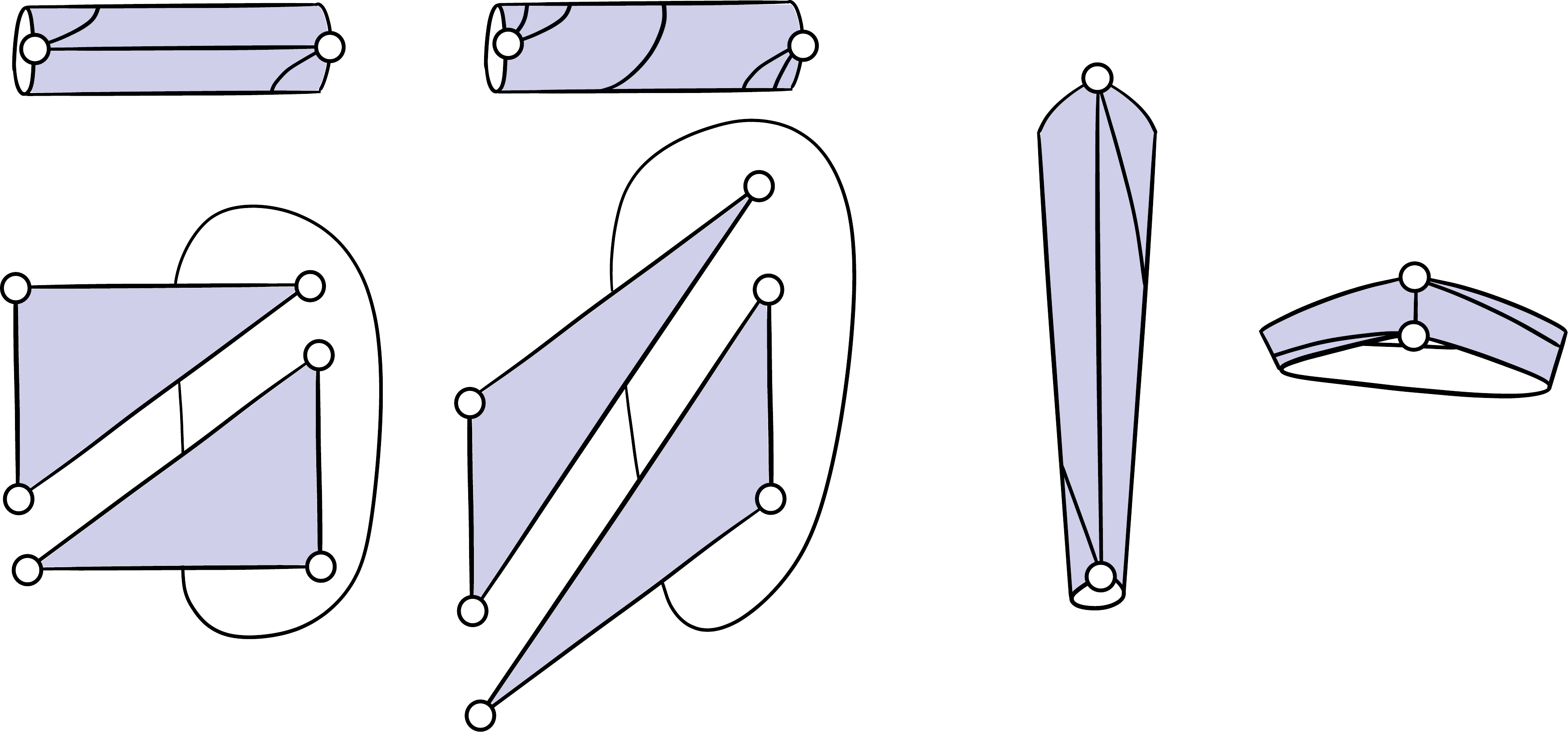}
    \caption{(From left to right) A good biface, a biface not good, a thin biface, a thick biface.}
    \label{fig:bifaces}
\end{figure}

We say that a biface $B$ is \emphdef{good} if the two interior edges $e$ and $f$ of $B^1$ satisfy both of the following up to possibly exchanging $e$ and $f$. First, $e$ is a shortest path in $\mathcal S(B)$. Second, cut $\mathcal S(B)$ along $e$, and consider the resulting quadrilateral. If this quadrilateral has two diagonals then $f$ is shortest among the two diagonals. We will distinguish good bifaces. A good biface $B$ is \emphdef{thin} if every interior edge of $B^1$ is longer than every boundary edge of $B^1$. Otherwise $B$ is \emphdef{thick}. While tubes and bifaces have unbounded happiness, good bifaces on the other hand satisfy the following (Appendix~\ref{app:bifaces}):

\begin{lemma}\label{L:biface interior hapiness}
Given a good biface $B$, let $e$ be an interior edge of $B^1$. Then $h_{\mathcal S(B)}(e) \leq 6$.
\end{lemma}

We will use the elementary operation of replacing a tube by a good biface (Appendix~\ref{app:bifaces}):

\begin{restatable}{proposition}{lemcomputebiface}\label{P:compute good biface}
Let $X$ be a tube with $n$ triangles, whose sides are smaller than $L > 0$. Let $s > 0$ be smaller than the systole of $\mathcal S(X)$. One can compute a good biface whose surface is $\mathcal S(X)$ in $O(n \log n \cdot \log(2+L/s))$ time.
\end{restatable}

Proposition~\ref{P:compute good biface} is similar to a result described by L\"{o}ffler, Ophelders, Silveira, and Staals~\cite[Theorem~45]{portalgons}, building upon a ray shooting algorithm of Erickson and Nayyeri~\cite{erickson2012tracing}.


\section{Description of the algorithm}\label{sec:algo}

In this section we describe our algorithm for Proposition~\ref{T:core algorithm}. We first describe the elementary operations and the data structure, before giving the algorithm itself. Along the way, we provide informal explanations of our choices. We do not prove anything, as the analysis of the algorithm is deferred to Section~\ref{sec:analysis}.

\subsection{Inserting vertices and edges}

Informally, our goal is to “improve the geometry” of a triangular portalgon $T$. We will make this precise in Section~\ref{sec:analysis}. Roughly, the issue is that, without any condition on $T$, the edges of $T^1$ that lie in the interior of $\mathcal S(T)$ can be arbitrarily long, so one of them may intersect some shortest path arbitrarily many times by wrapping around the surface, and so the segment-happiness of $T$ can be arbitrarily large. A naive way of shortening an edge is to cut the edge in two at its middle point.

\algo{InsertVertices}{Given a triangular portalgon $T$, consider every edge $e$ of $T^1$ that lies in the interior of $\mathcal S(T)$, and insert the middle point of $e$ as a vertex in $T^1$.}

Appendix~\ref{app:detail routines} details how to modify the portalgon $T$ to perform \algoref{InsertVertices}. Applying \algoref{InsertVertices} to a \emph{triangular} portalgon $T$ produces a portalgon $T'$ whose polygons are usually not triangles. We now consider transforming $T'$ into a triangular portalgon. To do that we repeatedly cut the polygons of $T'$. We need a definition. In the plane consider a polygon $P$, two vertices $u \neq v$ of $P$, and the rectilinear segment $a$ between $u$ and $v$. If the relative interior of $a$ is included in the interior of $P$ then $a$ is called a vertex-to-vertex arc of $P$. It is easily seen that if $P$ is not a triangle then $P$ has at least one vertex-to-vertex arc. Among the vertex-to-vertex arcs of $P$, the shortest ones are the \emphdef{shortcuts} of $P$. We emphasize that we consider the shortest ones among all the vertex-to-vertex arcs, without fixing the endpoints, but the endpoints are chosen among the vertices of $P$. In a portalgon $T$ every polygon $P$ corresponds to a face $F$ of $T^1$, and every shortcut of $P$ corresponds to a path whose relative interior is included in $F$: we say of this path that it is a shortcut of $F$.

\algo{InsertEdges}{Given a portalgon $T$, as long as there is a face of $T^1$ that is not a triangle, insert a shortcut of this face as an edge in $T^1$.}

Appendix~\ref{app:detail routines} details how to modify the portalgon $T$ to perform \algoref{InsertEdges}. We shall apply \algoref{InsertVertices} followed by \algoref{InsertEdges} to a triangular portalgon $T$ in order to produce another triangular portalgon $T'$, hopefully with a “nicer geometry”. The  problem is now that $T'^1$ has more vertices than $T^1$. All the other operations of the algorithm are devoted to keeping the number of vertices low.

\subsection{Deleting vertices}

From now on it is important that every surface considered is flat, there is no singularity in its interior. Given a triangular portalgon $T$, assuming that the surface $\mathcal S(T)$ is flat, we consider decreasing the number of vertices of $T^1$. To do that we naturally consider deleting some vertices. Not all vertices can be deleted. For example a vertex incident to a loop edge cannot be deleted. Also we will not delete vertices that lie on the boundary of the surface $\mathcal S(T)$. A vertex of $T^1$ is \emphdef{weak} if it lies in the interior of $\mathcal S(T)$ and is not incident to any loop edge in $T^1$. It is \emphdef{strong} otherwise.

\algo{DeleteVertices}{Given a triangular portalgon $T$ whose surface $\mathcal S(T)$ is flat, construct a maximal independent set $V$ of weak vertices of $T^1$ that have degree smaller than or equal to six. For every vertex $v \in V$ delete $v$ and its incident edges from $T^1$.}

Appendix~\ref{app:detail routines} details how to modify the portalgon $T$ to perform \algoref{DeleteVertices}. Afterward the polygons of $T$ are usually not triangles anymore, but this will be solved by applying \algoref{InsertEdges} after each application of \algoref{DeleteVertices}. Observe that in \algoref{DeleteVertices} we delete only vertices of degree smaller than or equal to six. Informally, the reason is that deleting a weak vertex of degree $d \geq 3$ creates a face of degree $d$ around it. We then insert $d-3$ edges in this face when applying \algoref{InsertEdges}. The problem is that only a constant number of edges can be inserted in each face without risking to destroy our improvements on the geometry of the tessellation. This is why we make sure that $d = O(1)$ beforehand. The exact bound on $d$ is not really important (although changing it would change some constants of the algorithm), but it must be at least six so that we can still remove a fraction of the excess vertices this way, at least when most of them are strong. Similar ideas can be found in the literature, see for example Kirkpatrick~\cite[Lemma~3.2]{kirkpatrick1983optimal}.

\subsection{Simplifying tubes}

The operation \algoref{DeleteVertices} cannot delete strong vertices,  and among them the vertices that lie the interior of the surface and are incident to a loop edge. In this section we describe an operation for deleting such vertices. 

In order to grasp the intuition, observe, informally, that it is possible that almost all the vertices of $T^1$ lie in the interior of $\mathcal S(T)$ and are incident to a loop edge. Fortunately, it turns out that in this case there must be a sub-portalgon $X$ of $T$ such that $X$ is a tube and the interior of $\mathcal S(X)$ contains loop edges of $X^1$. We delete such loop edges by replacing $X$ by a good biface with Proposition~\ref{P:compute good biface}. There is one subtlety: we must choose $X$ carefully so that we replace any concatenation of tubes by a single biface when possible, in order to delete the loops in-between the tubes, instead of replacing each tube individually. That leads to:

\algo{SimplifyTubes}{In a triangular portalgon $T$ whose surface $\mathcal S(T)$ is flat, do the following:
\begin{enumerate}
    \item In $T^1$ build a set $J$ of loop edges that lie in the interior of $\mathcal S(T)$ and are pairwise~disjoint, as follows. There are two cases:
    \begin{enumerate}
        \item If $\mathcal S(T)$ is homeomorphic to a torus, do the following. Let $J$ contain two disjoint loop edges of $T^1$ if there exist two such edges, otherwise let $J = \emptyset$. 
        
        \item Otherwise do the following. Construct a set $J'$ of loop edges by considering every vertex $v$ of $T^1$ that lies in the interior of $\mathcal S(T)$ and is incident to a loop edge, and by putting one of the loop edges incident to $v$ in $J'$. Then build a subset $J \subseteq J'$ by removing from $J'$ every $e \in J'$ satisfying both of the following. First, cutting $\mathcal S(T)$ along the loops in $J'$, and considering the resulting connected components, two such components are adjacent to $e$ (instead of one), say $S_0$ and $S_1$. Second, each one of the two sub-portalgons of $T$ whose surfaces are $S_0$ and $S_1$ is a tube.
        
    \end{enumerate}
    \item  Cut the surface $\mathcal S(T)$ along the loops in $J$. Each resulting component is the surface of a sub-portalgon $X$ of $T$. If $X$ is a tube replace $X$ by a good biface $B$.
\end{enumerate}}

The idea behind step 1b is to remove loops from $J'$ so that step 2 replaces a concatenation of tubes by a single good biface when possible, instead of replacing the tubes individually.

\subsection{Data structure for marking bifaces as inactive}

We are almost ready to give the algorithm, but there is still one important thing to describe. In step 2 of \algoref{SimplifyTubes}, if the good biface $B$ is \emph{thin} we will not just replace $X$ by $B$, but we will also make sure to not modify $B$ ever again. In this sense $B$ becomes inactive. Doing so requires a data structure remembering which parts of the portalgon are inactive.

\begin{figure}
    \centering
    \includegraphics[width=0.7\linewidth]{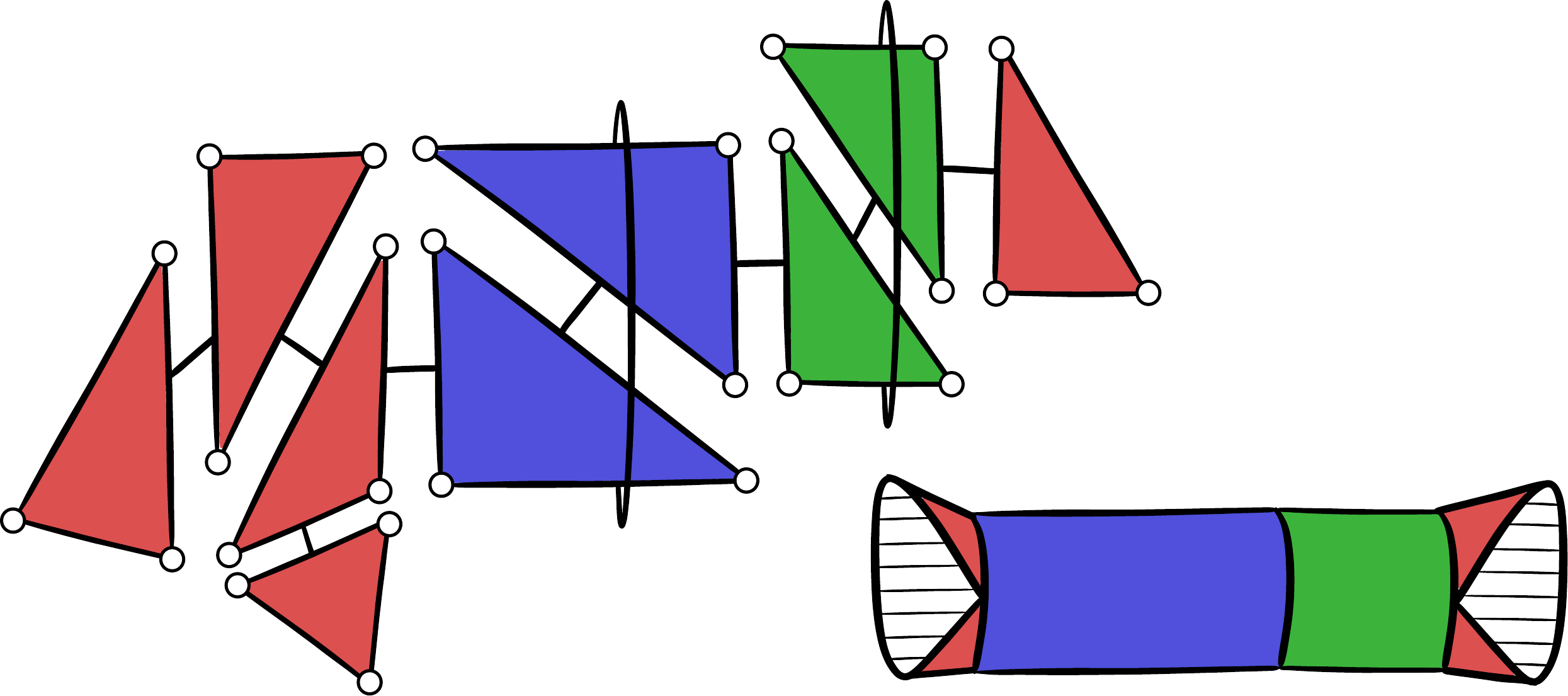}
    \caption{Data structure for \algoref{Algorithm}: a portalgon whose polygons are partitioned, here by color, inducing sub-portalgons called regions, and a region singularized as active, here in red.}
    \label{fig:DSPortalgon}
\end{figure}

See Figure~\ref{fig:DSPortalgon}. The data structure maintains a portalgon $R$ together with a partition of the polygons of $R$. Each set $X$ of polygons in the partition defines a sub-portalgon of $R$ which we call \emphdef{region}. One region is singularized as the \emphdef{active} region $R_A$. The other regions are \emphdef{inactive}. Note that the surface of the active region may be disconnected, and that the surfaces of distinct inactive regions may be adjacent.


The data structure will be initialized by setting $R_A = R$, without inactive region. Then the algorithm will apply the routines \algoref{InsertVertices}, \algoref{InsertEdges}, \algoref{DeleteVertices}, and \algoref{SimplifyTubes} to the active region $R_A$, and mark as inactive every thin biface encountered in step 2 of \algoref{SimplifyTubes}. The surface of $R_A$ will diminish over time as more and more regions are marked inactive. This may increase the numbers of connected components and boundary components of $\mathcal S(R_A)$, ruining our efforts to keep the combinatorial complexity of $R_A$ bounded. To counteract this, we introduce:

\algo{Gardening}{Every connected component of $\mathcal S(R_A)$ is the surface of a sub-portalgon $X$ of $R_A$. If $X$ is a tube replace $X$ by a good biface $B$, and mark $B$ as inactive.}


We described everything that the algorithm can do to the data structure. This immediately implies three invariants maintained by the algorithm. First, 1) Every polygon of the active region has degree at most six, and 2) Every inactive region is a good biface. For the last invariant we need a definition. Recall that in $R$ if two sides $s$ and $s'$ of polygons are matched then $s$ and $s'$ correspond to an edge $e$ of $R^1$. If moreover $s$ and $s'$ belong to different polygons, and if their respective polygons belong to different regions, we say that $e$ is \emphdef{separating}. Then $e$ is a loop, for it is a boundary edge of a biface by 2), and $e$ belongs to the interior of $\mathcal S(R)$. The third invariant is that 3) The separating loops are pairwise disjoint (no two of them are based at the same vertex of $R^1$).

\subsection{Algorithm}

The algorithm repeatedly applies two parts. The first part “improves the geometry” by applying \algoref{InsertVertices} and then \algoref{InsertEdges}. However this increases the number of vertices. So the second part applies \algoref{SimplifyTubes}, \algoref{DeleteVertices}, and \algoref{InsertEdges}, together with \algoref{Gardening}. The second part can only remove a fraction of the vertices at once, so it is repeated several times. It turns out that 350 repetitions suffice.

\algo{Algorithm}{Given a triangular portalgon $T$ whose surface $\mathcal S(T)$ is flat, and $N \geq 1$, do the following. Initialize the data structure by letting $R$ be the input portalgon $T$, and by letting the active region $R_A$ be $R$ itself, without inactive region. Repeat $N$ times the following:
\begin{enumerate}
    \item Apply \algoref{InsertVertices} to $R_A$. Then apply \algoref{InsertEdges} to $R_A$.
    \item Repeat 350 times the following:
    \begin{enumerate}
        \item Apply \algoref{Gardening}. Then apply \algoref{SimplifyTubes} to $R_A$ but in step 2 of \algoref{SimplifyTubes}, whenever $B$ is thin, mark $B$ as inactive. Apply \algoref{Gardening} again.
        \item Apply \algoref{DeleteVertices} to $R_A$. Then apply \algoref{InsertEdges} to $R_A$.
    \end{enumerate}
\end{enumerate}
In the end return $R$.}


When proving Proposition~\ref{T:core algorithm}, we will apply \algoref{Algorithm} with $N = \lceil \log(2 + L/s) \rceil$.


\section{Analysis of the algorithm}\label{sec:analysis}

In this section, we sketch the analysis of \algoref{Algorithm} (Section~\ref{sec:algo}) to prove Proposition~\ref{T:core algorithm}.

\subsection{Combinatorial analysis}\label{sec:combi analysis}


\begin{restatable}{proposition}{propcombianalysis}\label{prop:bound number vertices}
Apply \algoref{Algorithm} to a portalgon $T$ of $n$ triangles, whose surface $\mathcal S(T)$ is flat. During the execution the number of polygons of the active region $R_A$ is $O(n)$.
\end{restatable}

We only sketch the proof of Proposition~\ref{prop:bound number vertices}, the complete proof is deferred to Appendix~\ref{app:combi analysis}.

\begin{proof}[Sketch of proof]
We consider $R_A^1$, the 1-skeleton of the active region $R_A$, and we show that the number $m_A$ of vertices of $R_A^1$ remains $O(n)$ throughout the execution. There are two loops in the algorithm: the main loop, which repeats $N$ times, and the interior loop, which repeats 350 times within each iteration of the main loop. To prove the lemma, we consider a single iteration of the main loop, we assume that $m_A$ exceeds $n$ by at least a constant factor at the beginning of the iteration, and we prove that $m_A$ has decreased after the iteration. 

The iteration starts with \algoref{InsertVertices}. This is the only moment where $m_A$ may increase, and we prove that $m_A$ is multiplied by at most a constant factor. Then the iteration applies the interior loop, and we claim that, as long as $m_A$ exceeds $n$ by a constant factor, $m_A$ is divided by at least a constant factor by each iteration of the interior loop. We show that this claim implies the lemma as the interior loop is applied sufficiently many times to counteract the initial increase of $m_A$. To prove the claim, we show that for \algoref{DeleteVertices} to remove a fraction of the vertices of $R_A^1$, it suffices that $m_A$ vastly exceeds the genus and the number of boundary components of $\mathcal S(R_A)$, and that almost all of the vertices of $R_A^1$ are weak. We show that this is ensured by first applying \algoref{Gardening} and \algoref{SimplifyTubes}.
\end{proof}


\subsection{Enclosure}\label{sec:enclosure}

To analyze \algoref{Algorithm} from a geometric point of view, we introduce, on the segments of a flat surface $S$, a parameter that we call \emph{enclosure}. So consider a segment $e$ of $S$. See Figure~\ref{fig:enclosure}. 

Informally, $e$ is “enclosed” in $S$ when a short non-contractible loop can be attached to a point of $e$ not too close to the endpoints of $e$. Formally, consider a point $x$ in the relative interior of $e$. We denote by $\langle x \rangle_e$ the minimum length of the two sub-segments of $e$ separated by $x$. Assume that there exists a loop $\gamma$ based at $x$ in $S$, such that $\gamma$ is geodesic except possibly at its basepoint. Further assume that its length satisfies $\ell(e) < \langle x \rangle_e$. In this case $\gamma$ and $e$ are necessarily in \emph{general position}: informally, they do not overlap, more formally, every sufficiently short sub-path of $\gamma$ is either disjoint from $e$ or its intersection with $e$ is a single point. There are two cases: either $\gamma$ crosses $e$ at $x$, or $\gamma$ meets $x$ on only one side of $e$. If $\gamma$ crosses $e$ at $x$, then we say that $\gamma$ \emphdef{encloses} $e$ \emphdef{in} $S$. Also we say that $\gamma$ encloses $e$ \emphdef{by a factor of} $\langle x \rangle_e / \ell(\gamma)$ in $S$. The \emphdef{enclosure} $c_S(e) \geq 1$ is the supremum of the ratios $\langle x \rangle_e / \ell(\gamma)$ over all the basepoints $x$ in the relative interior of $e$, and over all the loops $\gamma$ based at $x$ that enclose $e$ in $S$. It is conventionally set to one if there is no loop enclosing $e$ in $S$.

\begin{figure}[ht]
    \centering
    \includegraphics[width=0.5\linewidth]{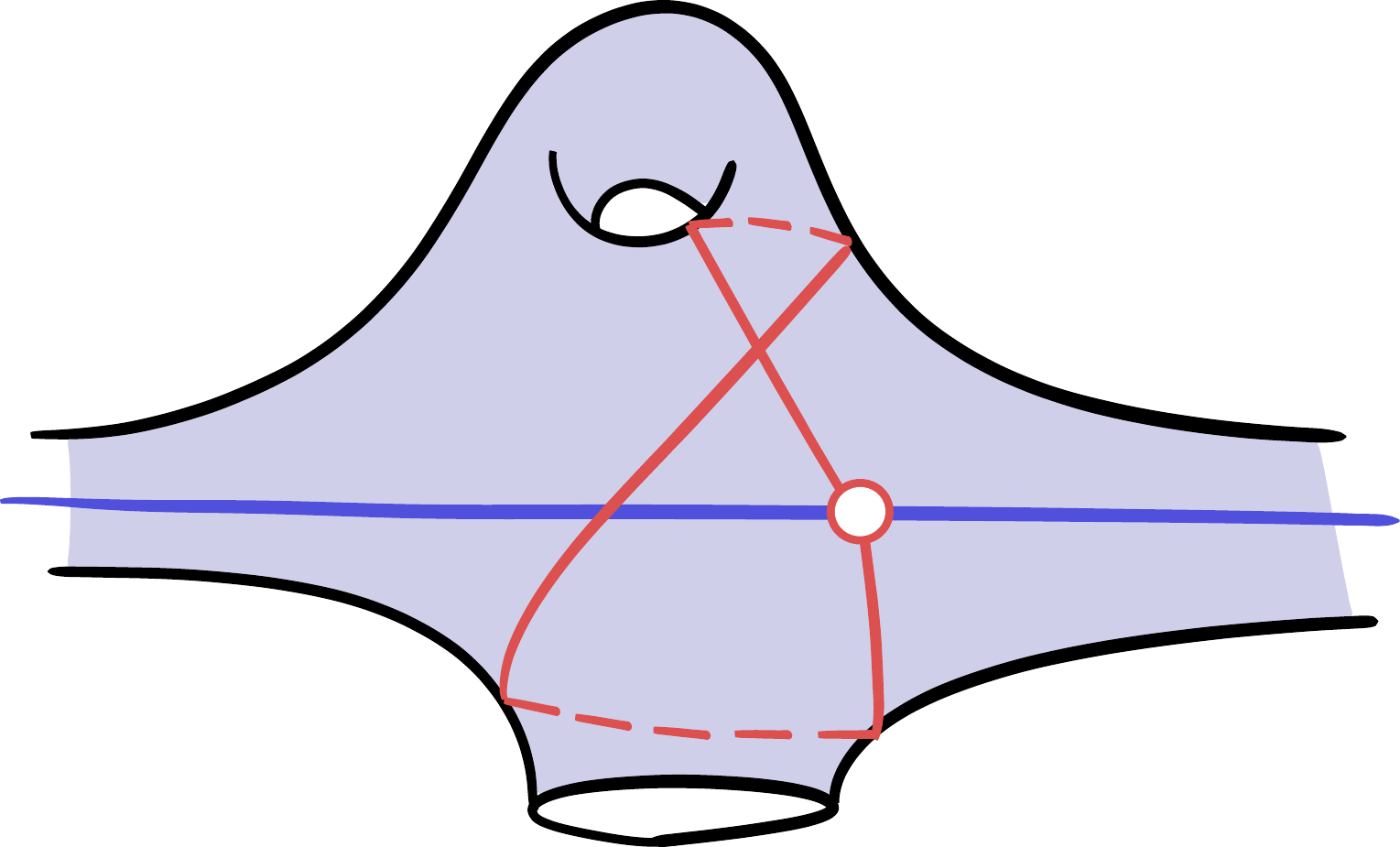}
    \caption{The red loop encloses the blue segment in the surface.}
    \label{fig:enclosure}
\end{figure}

The segment-happiness $h_S(e)$ and the length $\ell(e)$ can be bounded from above using the enclosure $c_S(e)$. Our bound depends on the surface $S$. More precisely, on the systole of $S$ and the diameter of $S$. But instead of the diameter of $S$, we consider a triangulation of $S$, and we use its number $n$ of triangles together with the maximum length $L$ of its edges. This will be more convenient to us when analyzing \algoref{Algorithm}. We prove (Appendix~\ref{app:enclosure}):

\begin{restatable}{proposition}{proplambdaandhapinessglobal}\label{P:lambda and hapiness global}
Let $e$ be a segment of $S$. Let $s > 0$ be smaller than the systole of $S$. Assume that there is a triangulation of $S$ with $n \geq 1$ triangles, whose edges are all smaller than $L > 0$. Then $h_S(e) = O(c_S(e) \cdot (1+ \log c_S(e) + \log n + \log \lceil L/s\rceil ))$ and $\ell(e) /s = O(c_S(e) \cdot n \cdot \lceil L/s\rceil^2)$.
\end{restatable}

In Proposition~\ref{P:lambda and hapiness global} the $O()$ notation does not depend on $S$, it involves a universal constant. In the second inequality of Proposition~\ref{P:lambda and hapiness global} the exact powers above $\lceil L/s\rceil$ and $n$, here $2$ and $1$, do not matter to us. We need only a polynomial in $\lceil L/s\rceil$ and $n$. 

\subsection{Geometric analysis}\label{sec:geometric analysis}

The geometric analysis of\textbf{Algorithm} consists in two properties on the enclosure and the length of the edges involved in any execution: Lemma~\ref{P:active boundary enclosure} and Proposition~\ref{P:main loop geometry} below, whose proofs we sketch in Section~\ref{sec:Pactiveboudnaryenclosure} and Section~\ref{sec:Pmainloopgeometry}. Each proof relies on properties of enclosure that are independent of \algoref{Algorithm}, or of any portalgon, and can be seen as independent mathematical contributions of us. In this extended abstract we only explain how these properties of enclosure serve to analyse \algoref{Algorithm}, deferring their proof to Appendix~\ref{app:enclosure}.

In this section, we fix a portalgon $T$ of $n$ triangles, whose sides are smaller than some positive real $L$, and whose surface $\mathcal S(T)$ is flat. We abbreviate $S = \mathcal S(T)$. We apply the algorithm \algoref{Algorithm} to $T$, and we discuss the execution of the algorithm.


\subsubsection{The separating loops are not very enclosed}\label{sec:Pactiveboudnaryenclosure}



\begin{restatable}{lemma}{activeboundaryenclosure}\label{P:active boundary enclosure}
Any time during the execution every separating loop $e$ satisfies $c_S(e) \leq 2$.
\end{restatable}

Lemma~\ref{P:active boundary enclosure} follows from the following property of enclosure (Appendix~\ref{app:enclosure}):


\begin{restatable}{proposition}{bifaceboundaryhapiness}\label{L:biface boundary hapiness}
Assume that $S$ contains the surface of a thin biface $B$, and let $e$ be one of the two boundary edges of $B^1$. Then $c_S(e) \leq 2$. 
\end{restatable}

\begin{proof}[Proof of Lemma~\ref{P:active boundary enclosure}]
Only step 2 of \algoref{SimplifyTubes} may create a separating loop, by marking a \emph{thin} biface $B$ as inactive. Then $B$ is is never touched again by the algorithm. So the algorithm maintains the invariant that every separating loop $e$ is adjacent to the surface of at least one inactive region that is a \emph{thin} biface. So $c_S(e) \leq 2$ by Proposition~\ref{L:biface boundary hapiness}.
\end{proof}

\subsubsection{The very enclosed edges shorten exponentially fast}\label{sec:Pmainloopgeometry}



\begin{restatable}{proposition}{scalesdownexponentially}\label{P:main loop geometry}
After $i \geq 1$ iterations of the main loop, let $e$ be an edge of $R_A^1$. If $c_S(e) > 22000 \cdot i$ then $\ell(e) < 2^{1-i} L$.
\end{restatable}

To prove Proposition~\ref{P:main loop geometry}, we analyze each routine applied. Informally, each application of \algoref{InsertVertices} “improves the geometry” of the active region, and the rest of the algorithm does not deteriorate this improvement too much. Formally:

\begin{restatable}{lemma}{lemgeominsertvertices}\label{lem:geom insert vertices}
Consider the active regions $R_A$ and $R_A'$ respectively before and after some application of \algoref{InsertVertices}. Assume that there is an edge $e'$ of $R_A'^1$ such that $c_S(e') > 2$. Then there is an edge $e$ of $R_A^1$ such that $c_S(e) \geq c_S(e')$ and $\ell(e) \geq 2 \ell(e')$.
\end{restatable}


Lemma~\ref{lem:geom insert vertices} follows from the following (easy) property of enclosure (Appendix~\ref{app:enclosure}):

\begin{restatable}{lemma}{lemmapropedges}\label{prop:edges}
Let $f \subseteq e$ be segments in $S$. Then $c_S(e) \geq c_S(f)$.
\end{restatable}

\begin{proof}[Proof of Lemma~\ref{lem:geom insert vertices}]
First observe that $e'$ is not included in the boundary of $\mathcal S(R_A')$ because $e'$ is enclosed and thus not included in the boundary of $S$, and because $e'$ is not a separating loop by Lemma~\ref{P:active boundary enclosure}. So there is an edge $e$ of $R_A^1$ such that $e'$ is one of the two half-segments obtained after the insertion of the middle point of $e$ as a vertex. Then $\ell(e) = 2 \ell(e')$. And $c_S(e) \geq c_S(e')$ by Lemma~\ref{prop:edges}.
\end{proof}

\begin{restatable}{lemma}{lemgeominsertedges}\label{lem:geom insert edges}
Consider the active regions $R_A$ and $R_A'$ respectively before and after some application of \algoref{InsertEdges}. Assume that there is an edge $e'$ of $R_A'^1$ such that $c_S(e') > 14$. Then there is an edge $e$ of $R_A^1$ such that $c_S(e) \geq c_S(e') - 12$ and $\ell(e) \geq (1 - 12/c_S(e')) \cdot \ell(e')$.
\end{restatable}


Lemma~\ref{lem:geom insert edges} follows from the following (key) property of enclosure (Appendix~\ref{app:enclosure}):

\begin{restatable}{proposition}{proppolygons}\label{prop:polygons}
Let $F$ be a face of a tessellation of $S$. Assume that $F$ has a shortcut $e$ such that $c_S(e) > 6$. Then $F$ has a side $f$ such that $c_S(f) \geq c_S(e)-4$ and $\ell(f) \geq (1 - 4/c_S(e)) \cdot \ell(e)$.
\end{restatable}

\begin{proof}[Proof of Lemma~\ref{lem:geom insert edges}]
Here we crucially use the fact that every polygon of $R_A$ has degree at most six. so that at most three edges are inserted within the polygon. Indeed either $e'$ was already an edge of $R_A^1$ and there is nothing to do, or $e'$ has been inserted in some face $F$ of $R_A^1$. At most three edges were inserted in $F$, and Proposition~\ref{prop:polygons} applied at most three times gives a boundary edge $e$ of $F$ such that $c_S(e) \geq c_S(e') - 12$ and $\ell(e) \geq (1 - 12/c_S(e')) \ell(e')$. 
\end{proof}

\begin{restatable}{lemma}{lemgeomsimplifytubes}\label{lem:geom simplify tubes}
Consider the active regions $R_A$ and $R_A'$ respectively before and after some application of \algoref{SimplifyTubes}. Assume that there is an edge $e'$ of $R_A'^1$ such that $c_S(e') > 6$. Then there is an edge $e$ of $R_A^1$ such that $c_S(e) \geq c_S(e') - 5$ and $\ell(e) \geq (1 - 4/c_S(e')) \cdot \ell(e')$.
\end{restatable}

Lemma~\ref{lem:geom simplify tubes} is similar to Lemma~\ref{lem:geom insert edges}, its proof is deferred to Appendix~\ref{sec:app geom simplify tubes}. 


\begin{proof}[Proof of Proposition~\ref{P:main loop geometry}]
Consider the active regions $R_A$ and $R'_A$ respectively at the beginning of the algorithm, and after $i$ iterations of the main loop. Assume that there is an edge $e'$ in $R_A'^1$ such that $c_S(e') > 22000 \cdot i$. During those $i$ iterations there has been $i$ applications of \algoref{InsertVertices}, $351 i$ applications of \algoref{InsertEdges}, and $350i$ applications of \algoref{SimplifyTubes}. Also $12 \cdot 351i + 5 \cdot 350i < 11000 i$. So Lemma~\ref{lem:geom insert vertices}, Lemma~\ref{lem:geom insert edges}, and Lemma~\ref{lem:geom simplify tubes} imply that there is an edge $e$ in $R_A^1$ such that $\ell(e) \geq 2^i(1 - 11000 i/c_S(e')) \ell(e') > 2^{i-1} \ell(e')$. And $\ell(e) \leq L$ because $e$ belongs to the input triangulation $T^1$.
\end{proof}

\subsection{Proof of Proposition~\ref{T:core algorithm}}\label{sec:proofofTcorealgo}

We need a last (easy) lemma (Appendix~\ref{app:bound systole}):

\begin{lemma}\label{lem:bound systole}
Let $S$ be a flat surface. Assume that $S$ contains the surface of a tube $X$. Then the systole of $\mathcal S(X)$ is greater than or equal to the systole of $S$.
\end{lemma}

\begin{proof}[Proof of Proposition~\ref{T:core algorithm}]
Apply \algoref{Algorithm} to $T$ with $N = \lceil \log(2+L/s) \rceil$, resulting in a triangular portalgon $R$. By Proposition~\ref{prop:bound number vertices} the number of polygons of the active region is $O(n)$ throughout the execution. So in the end $R$ has $O(n \cdot \log(2+L/s))$ triangles; Indeed each iteration of the main loop marks $O(n)$ triangles as inactive, and there are $\lceil \log(2+L/s) \rceil$ iterations of the main loop. We have two claims that immediately imply the proposition. 

Our first claim is that the algorithm takes $O(n \log^2(n) \cdot \log^2(2+L/s))$ time. Let us prove this first claim. Each application of \algoref{InsertVertices} or \algoref{InsertEdges} takes $O(n)$ time. And each application of \algoref{SimplifyTubes} or \algoref{Gardening} takes $O(n \log(n) \cdot \log(2+\Lambda/s))$ time by Proposition~\ref{P:compute good biface} and Lemma~\ref{lem:bound systole}, where $\Lambda$ is the maximum length reached by an edge of the 1-skeleton of the active region during the execution. Now let us bound $\Lambda$. If at some point an edge $e$ of the 1-skeleton of the active region is longer than $L$ then $c_S(e) = O(\log(2+L/s))$ by Proposition~\ref{P:main loop geometry}. Moreover $\ell(e)/s = O(c_S(e) \cdot n \lceil L/s\rceil^2)$ by Proposition~\ref{P:lambda and hapiness global}. This proves $\log(2+\Lambda/s) = O(\log(n) + \log(2+L/s))$, which proves the claim.

Our second claim is that in the end every edge $e$ of $R^1$ satisfies $h_S(e) = O(\log (n) \cdot \log^2(2+L/s))$. Let us prove this second claim. First observe that if $e$ is in $R_A^1$ then $c_S(e) < 22000 \log(2+L/s)$, for otherwise Proposition~\ref{P:main loop geometry} would imply $\ell(e) < 2s$, implying that no loop encloses $e$ in $S$, a contradiction. In this case $h_S(e) = O(\log(2+L/s)\cdot (\log(n) + \log(2+L/s)))$ by Proposition~\ref{P:lambda and hapiness global}, and we are done. Every other edge of $R^1$ belongs to the 1-skeleton of an inactive good biface $B$. Every boundary edge $e$ of $B^1$ is either a boundary component of $S$ or a separating loop, so $c_S(e) \leq 2$ by Lemma~\ref{P:active boundary enclosure}, and so $h_S(e) = O(\log(n) + \log(2+L/s))$ by Proposition~\ref{P:lambda and hapiness global}. Every interior edge $f$ of $B^1$ then satisfies $h_S(f) = O(\log(n) + \log(2+L/s))$ by Lemma~\ref{L:biface interior hapiness}. This proves the second claim, and the proposition.
\end{proof}


\bibliographystyle{plainurl}
\bibliography{bib}

\appendix


\section{Appendix of Section~\ref{sec:bifaces}}\label{app:bifaces}


\begin{proof}[Proof of Lemma~\ref{L:biface interior hapiness}]
Among the two interior edges of $B^1$, let $f$ be a shortest one. Let $g \neq f$ be the other interior edge of $B^1$. Let $p$ be a shortest path in $\mathcal S(B)$. The relative interior $\mathring p$ of $p$ cannot intersect the relative interior of $f$ twice for those intersections would be crossings and $p$ and $f$ are both shortest paths because $B$ is good. So $\mathring p$ intersects $f$ less than four times. Then $\mathring p$ cannot intersect the relative interior of $g$ five times, for those intersections would be crossings, and $\mathring p$ would intersect $f$ in-between any two consecutive crossings with the relative interior of $g$. Altogether $p$ intersects $f$ and $g$ at most six times each.
\end{proof}


The rest of this section is dedicated to the proof of Proposition~\ref{P:compute good biface}, which we restate for convenience:

\lemcomputebiface*

Proposition~\ref{P:compute good biface} is similar to but different from a result of L\"{o}ffler, Ophelders, Silveira, and Staals~\cite[Theorem~45]{portalgons} (building upon a ray shooting algorithm of Erickson and Nayyeri~\cite{erickson2012tracing}), in which the authors provide an algorithm to transform a biface into a portalgon of bounded happiness, and of bounded combinatorial complexity. They extend their result from bifaces to portalgons $X$ such that the dual graph of $X^1$ in $\mathcal S(X)$ has at most one simple cycle, but unfortunately this does not include tubes. We extend their result to tubes to prove Proposition~\ref{P:compute good biface}, reusing some of ideas developed in the core of the paper. 

We need a few lemmas. The following is a corollary of~\cite[Theorem~45]{portalgons}:

\begin{lemma}\label{T:Make a biface good}
Let $B$ be a biface of happiness $h$. One can compute in $O(1+ \log h)$ time a good biface whose surface is that of $B$.
\end{lemma}

\begin{proof}
By the result of L\"{o}ffler, Ophelders, Silveira, and Staals~\cite[Theorem~45]{portalgons} we can compute in $O(1+ \log h)$ time a portalgon $T$, whose surface is $\mathcal S(B)$, whose happiness is $O(1)$, and whose 1-skeleton $T^1$ has $O(1)$ edges. Without loss of generality the two vertices $b_0$ and $b_1$ of $B^1$ are also vertices of $T^1$, and we know which vertices of the polygons of $T$ correspond to $b_0$ and $b_1$. 

We now describe how to compute, in constant time, from $T$, a good biface of $\mathcal S(T)$. The key thing is that we can exploit the fact that $T$ has $O(1)$ combinatorial complexity and happiness to compute by exhaustive search. First compute, in constant time, by exhaustive search, a shortest path $q$ between $b_0$ and $b_1$ in $\mathcal S(T)$: represent $q$ by its pre-image in the polygons of $T$. Then cut the polygons of $T$ along the pre-image of $q$: every time a polygon is cut in two along a segment $a$, the two edges issued of $a$ are not matched in the resulting portalgon (the goal is to cut the surface of $T$, not just changing $T$). Consider the resulting portalgon $D$. Then $\mathcal S(D)$ is homeomorphic to a closed disk. The two endpoints $b_0$ and $b_1$ of $q$ become a set $V$ of four vertices of $D^1$ that lie on the boundary of $\mathcal S(D)$. Every singularity of $\mathcal S(D)$ lies on the boundary of $\mathcal S(D)$ and belongs to $V$. Now replace $D$ by a triangular portalgon $D'$, of the same surface, and such that the vertex set of $D'^1$ is exactly $V$. This can be done for example by iteratively inserting vertex-to-vertex arcs in the faces of $D^1$ to make $D^1$ a triangulation, and by deleting a vertex $v$ of $D^1$ and its incident edges. When $v$ lies on the boundary of $\mathcal S(D)$, only the edges whose relative interior is included in the interior of $\mathcal S(D)$ are deleted. In the end, identify back the occurrences of $q$ on the boundary of $\mathcal S(D')$, by matching the two corresponding sides of polygons in $D'$, thereby obtaining a biface $B'$ of $\mathcal S(B)$ such that $q$ is an interior edge of $B'$. Change the other interior edge of $B'$ if necessary so that $B'$ is good.
\end{proof}

Consider $k \geq 1$ bifaces $B_1, \dots, B_k$. For every $1 \leq i \leq k$ let $e_i$ and $f_i$ be the two sides of triangles of $B_i$ that correspond to the boundary of $\mathcal S(B_i)$. If $i < k$, assume $\ell(e_i) = \ell(f_{i+1})$, and match $e_i$ with $f_{i+1}$. The resulting triangular portalgon $T$ is a \emphdef{concatenation} of the bifaces $B_1, \dots, B_k$. Note that $T$ is not necessarily a tube, for the vertices of $T^1$ in the interior of $\mathcal S(T)$ may be singularities.

\begin{lemma}\label{L:concat of bifaces is happy}
Let $T$ be the concatenation of two good bifaces. If $T$ is a tube, then one can compute in constant time a good biface whose surface is that of $T$.
\end{lemma}

\begin{proof}
Consider a shortest path $p$ in $\mathcal S(T)$, and the loop edge $e$ of $T^1$ that lies in the interior of $\mathcal S(T)$, in-between the surfaces of the two bifaces. We claim that the relative interior of $p$ does not cross the relative interior of $e$ more than twice. By contradiction assume that $p$ crosses the relative interior of $e$ three times. There is a connected component $S_0$ of $\mathcal S(T) \setminus e$ whose angle at the base vertex of $e$ is greater than or equal to $\pi$. Some portion $p'$ of $p$ enters $S_0$ and then leaves $S_0$ by two of the three crossings between $p$ and $e$. One of the two connected components of $S_0 \setminus p'$, say $S_1$, is homeomorphic to an open disk. By construction $S_1$ has at most three angles distinct from $\pi$: at the two points where $p$ crosses $e$, and possibly at the base vertex of $e$. By the Gauss-Bonnet theorem, there are exactly three such angles, not less, and they are all smaller than $\pi$. One of them is the incidence of $S_0$ and the base vertex of $e$. This is a contradiction. This proves the claim.

Using the claim immediately the intersection of $p$ and $e$ has $O(1)$ connected components, so $p$ writes as a concatenation of $k = O(1)$ paths $p_1, \dots, p_k$ such that for every $1 \leq i \leq k$ the path $p_i$ is either included in $e$ or its relative interior is disjoint from $e$. Every edge $f \neq e$ of $T^1$ intersects $p_i$ less than 7 times: if $f$ is included in the boundary of $\mathcal S(T)$ then $f$ intersects $p_i$ at most once, otherwise Lemma~\ref{L:biface interior hapiness} applies. So $f$ intersects $p$ less than $O(1)$ times. We proved that the segment-happiness of $T$ is $O(1)$. Then the happiness of $T$ is also $O(1)$ because the polygons of $T$ are all triangles. So we can compute a good biface whose surface is that of $T$ in constant time, exactly as in the proof of Lemma~\ref{T:Make a biface good}.
\end{proof}

We will use the following simple consequence of Euler's formula, similar to Lemma~\ref{lem:euler general}: 

\begin{lemma}\label{L:euler annulus}
Let $S$ be the topological annulus. Let $Y$ be a topological triangulation of $S$ that has only one vertex on each boundary component of $S$. Among the vertices of $Y$ that lie in the interior of $S$ and are not incident to any loop edge, at least half have degree smaller than or equal to ten.
\end{lemma}

\begin{proof}
We may assume without loss of generality that no vertex of $Y$ in the interior of $S$ is incident to a loop edge, by cutting $S$ open at an interior loop edge and recursing on the resulting two triangulations otherwise. Euler's formula gives $m - m_1 + m_2 = 0$, where $m$, $m_1$, and $m_2$ count respectively the vertices, edges, and faces of $Y^1$. Double counting gives $3m_2 = 2m_1 -2$ and $\sum_{v} \deg v = 2m_1$, where the sum is over the vertices $v$ of $Y$. Then $\sum_v (6 - \deg v)  = 4$. The two vertices of $Y$ on the boundary of $S$ have degree greater than or equal to four. So in the interior of $S$ every vertex of degree greater than ten must be compensated by a vertex of degree smaller than or equal to ten.
\end{proof}

Now we start proving Proposition~\ref{P:compute good biface}. In particular we fix a tube $X$ with $n$ triangles, whose sides are all smaller than some $L > 0$. 

\begin{lemma}\label{L:concatenation of bifaces}
One can compute in $O(n \log n)$ time a concatenation of less than $3n$ bifaces, whose surface is that of $X$, whose edges are all shorter than $(3n)^c L$ with $c = \log_{14/13}(3) < 15$.
\end{lemma}

\begin{proof}
Let us first describe the algorithm before analyzing it. As long as there are vertices of $X^1$ in the interior of $\mathcal S(X)$ that are not incident to any loop edge and have degree smaller than or equal to ten, we consider a maximal independent set $V$ of such vertices, and we do the following. First we delete all the vertices in $V$ along with their incident edges. Then we insert arbitrary vertex-to-vertex arcs in the faces of $X^1$ to make $X^1$ a triangulation again. 

The algorithm terminates because the number of vertices of $X^1$ decreases at each iteration. In the end every vertex in the interior of $\mathcal S(X)$ is incident to a loop edge by Lemma~\ref{L:euler annulus}, so $X$ is a concatenation of less than $m$ bifaces, where $m \leq 3n$ is the initial number of vertices of $X^1$. Each iteration can be performed in $O(n)$ time by maintaining a bucket with the vertices of degree smaller than or equal to ten. And we claim that there less than $\log_{14/13} m$ iterations. Before proving the claim, observe that it implies the lemma. Indeed the algorithm then terminates in $O(n \log n)$ time. Also no edge can get longer than $3^{\log_{14/13} m} L = m^c L$ because the maximum edge length of $X^1$ cannot be multiplied by more than 3 at each iteration.

Let us now prove the claim. Consider the number $m'$ of vertices of $X^1$ not incident to any loop edge that lie in the interior of $\mathcal S(X)$. By Lemma~\ref{L:euler annulus}, if $m' > 0$ before an iteration of the algorithm, then at least $m'/2$ such vertices have degree smaller than or equal to ten. So $V$ contains at least $m'/14$ vertices, which are deleted. Every non-deleted vertex that was incident to a loop edge before the iteration remains incident to a loop edge after the iteration. We proved that $m'$ is divided by at least $14/13$ during the iteration, which proves the claim.
\end{proof}

\begin{proof}[Proof of Proposition~\ref{P:compute good biface}]
Apply Lemma~\ref{L:concatenation of bifaces}, and replace $X$ in $O(n \log n)$ time by a concatenation of less than $3n$ bifaces whose edges are smaller than $(3n)^c L$ for some constant $c > 0$. Each biface $B$ has segment-happiness $O(1+(3n)^c L /s)$; indeed the systole of $\mathcal S(B)$ is greater than or equal to the systole of $X$, so every segment $e$ in $\mathcal S(B)$ satisfies $h_{\mathcal S(B)}(e) = O(1+\ell(e)/s)$. Replace $B$ by a good biface whose surface is that of $B$ in $O(\log(n) + \log(2+L/s))$ time with Lemma~\ref{T:Make a biface good}. Doing so for all bifaces takes $O(n \cdot (\log(n) + \log(2+L/s)))$ time in total. We crudely bound this running time from above by $O(n \log(n) \cdot \log(2+L/s))$. In the end apply Lemma~\ref{L:concat of bifaces is happy} repeatedly to merge those $O(n)$ good bifaces into a single good biface, in $O(n)$ total time.
\end{proof}


\section{Appendix of Section~\ref{sec:algo}}\label{app:detail routines}

In this section, we detail how to modify a portalgon $T$ to perform the routines \algoref{InsertVertices}, \algoref{InsertEdges}, and \algoref{DeleteVertices}.

To perform \algoref{InsertVertices}, recall that $T$ is given as a disjoint collection of triangles in the plane, together with a partial matching of their sides: we consider every triangle $P$ of $T$, and every side $s$ of $P$ that is matched in $T$, and we make the middle point of $s$ a new vertex of $P$.

We perform \algoref{InsertEdges} as follows: as long as there is a polygon $P$ of $T$ that is not a triangle, we cut $P$ into two polygons along a shortcut. This creates two new polygon sides, which we match in $T$.

We perform \algoref{DeleteVertices} as follows. To delete a vertex $v$ of $T^1$, we consider the triangle vertices of $T$ that correspond to $v$. No two of them belong to the same triangle for otherwise there would be a loop of $T^1$ based at $v$, contradicting the assumption that $v$ is weak. We move their triangles in the plane so that these vertices are now placed at the same point of the plane, and so that the triangles are placed in the correct cyclic order around this point, without overlapping. This is possible because $v$ lies in the interior of $\mathcal S(T)$, and because we assumed that every point in the interior of $\mathcal S(T)$ is flat: it is surrounded by an angle of $2\pi$. Now the union of the triangles is a polygon. In $T$, we replace all the triangles by this single polygon.


\section{Appendix of Section~\ref{sec:combi analysis}: proof of Proposition~\ref{prop:bound number vertices}}\label{app:combi analysis}

In this section we prove Proposition~\ref{prop:bound number vertices}, which we restate for convenience:

\propcombianalysis*

We analyze each operation independently before proving Proposition~\ref{prop:bound number vertices}. Our analysis is on the number vertices of $R_A^1$, not the number of polygons of $R_A$, but bounding one immediately bounds the other, as we shall see, and we find it more convenient to reason about the vertices of $R_A^1$.

\subsection{Analysis of \algoref{InsertVertices}}

We start by bounding the increase in vertices of \algoref{InsertVertices}:

\begin{lemma}\label{lem:insert how many vertices}
Let $T$ be a triangular portalgon. Let $g$ be the genus of $\mathcal S(T)$. Let $m$ be the number of vertices of $T^1$. Apply \algoref{InsertVertices} to $T$ and consider the resulting portalgon $T'$. Then $T'^1$ has less than $7(g+m)$ vertices.
\end{lemma}

Lemma~\ref{lem:insert how many vertices} relies on the following classical consequence of Euler's formula:

\begin{lemma}\label{lem:euler edges}
There are less than $6(g+m)$ edges in $T^1$.
\end{lemma}

\begin{proof}
Let $m_1$ and $m_2$ count respectively the edges and the faces of $T^1$, and let $b$ count the boundary components of $\mathcal S(T)$. Double counting gives $3m_2 \leq 2m_1$. Euler's formula gives $m_1 - m_2 = m + 2g+ b - 2$. And we have $b \leq m$. Therefore $m_1 \leq 3m_1 - 3m_2 < 6(m+g)$. 
\end{proof}

\begin{proof}[Proof of Lemma~\ref{lem:insert how many vertices}]
There are no more vertices inserted than there are edges in $T^1$, and there are less than $6(g+m)$ edges in $T^1$ by Lemma~\ref{lem:euler edges}.
\end{proof}

\subsection{Analysis of \algoref{DeleteVertices}}

For \algoref{DeleteVertices} to remove a fraction of the vertices, it suffices that the number of vertices vastly exceeds the topology of the surface, and that almost all of the vertices are weak:

\begin{lemma}\label{lem:deletes fraction of vertices}
Let $T$ be triangular portalgon whose surface $\mathcal S(T)$ is flat. Let $m$ be the number of vertices of $T^1$. Let $g$ be the genus of $\mathcal S(T)$, and let $\bar m$ be the number of strong vertices of $T^1$. Apply \algoref{DeleteVertices} to $T$ and consider the resulting portalgon $T'$. If $m > 24(g+\bar m)$ then $T'^1$ has less than $167m/168$ vertices.
\end{lemma}

Lemma~\ref{lem:deletes fraction of vertices} relies on the following classical consequence of Euler's formula:

\begin{lemma}\label{lem:euler general}
Let $S$ a topological surface of genus $g$ with $b$ boundary components. Let $Y$ be a topological triangulation of $S$ with $m$ vertices. If $m > 24(g+b)$ then at least $m/12$ vertices of $Y$ have degree smaller than or equal to 6.
\end{lemma}

\begin{proof}
Let $m_1$ and $m_2$ count respectively the edges and the faces of $Y$. Euler's formula gives $6m - 6m_1 + 6m_2 = 12 - 12g - 6b$. Double counting gives $3m_2 \leq 2m_1 - b$ and $2m_1 = \sum_v \deg v$, where the sum is over the vertices, and where $\deg v$ denotes the degree of a vertex $v$. Then $\sum_v 6 - \deg v = 6m - 2m_1 \geq 6m - 6m_1 + 6m_2 + 2b \geq 12 - 12g - 4b > -m/2$. Let $a$ and $b$ count the number of vertices whose degree is respectively smaller than or equal to six, and greater than six. Then $b < 5a + m/2$. Assuming $a < m/12$, we get $b < 11m/12$, and so $a + b < m$. This is a contradiction. This proves the lemma.
\end{proof}

\begin{proof}[Proof of Lemma~\ref{lem:deletes fraction of vertices}]
Let $b$ be the number of boundary components of $\mathcal S(T)$. We have $m > 24 (g+b)$. Indeed we assumed $m > 24 (g+\bar m)$, and we have $\bar m \geq b$ as every boundary component of $\mathcal S(T)$ contains a strong vertex of $T^1$. So by Lemma~\ref{lem:euler general} at least $m/12$ vertices of $T^1$ have degree smaller than or equal to six. Moreover less than $m/24$ vertices of $T^1$ are strong by assumption. So more than $m/24$ vertices of $T^1$ are weak and have degree smaller than or equal to six. Any maximal independent set of such vertices contains more than $m/ (24 \times 7) = m/168$ vertices, so \algoref{DeleteVertices} deletes more than $m/168$ vertices.
\end{proof}

\subsection{Analysis of \algoref{SimplifyTubes}}

Right after applying \algoref{SimplifyTubes} the number of vertices that lie in the interior of the surface and are incident to a loop is bounded by the topology of the surface:

\begin{lemma}\label{lem:tubing23}
Let $T$ be a triangular portalgon whose surface $\mathcal S(T)$ is flat. Let $g$ and $b$ be the genus and the number of boundary components of $\mathcal S(T)$. Apply \algoref{SimplifyTubes} to $T$, and consider the resulting portalgon $T'$. At most $9(g + b)$ vertices of $T'^1$ lie in the interior of $\mathcal S(T')$ and are incident to a loop in $T'^1$.
\end{lemma}

Lemma~\ref{lem:tubing23} relies on the following:

\begin{lemma}\label{L:packing}
Let $I$ be a set of loop edges of $T^1$ that lie in the interior of $\mathcal S(T)$ and are pairwise disjoint. In $I$ all but at most $9(g+b)$ loops $e$ satisfy the following: there are two connected components of $\mathcal S(T) \setminus I$ incident to $e$, and each of them is the surface of a sub-portalgon of $T$ that is a tube.
\end{lemma}

\begin{proof}
Cut $\mathcal S(T)$ along $I$, and consider the resulting connected components. Those components are the surfaces of sub-portalgons of $T$. Let $Z$ contain those sub-portalgons of $T$. Let $Z' \subseteq Z$ contain the sub-portalgons that are not tubes. Without loss of generality $I \neq \emptyset$. Then every $T_0 \in Z$ is such that $\partial \mathcal S(T_0) \neq \emptyset$ because $\mathcal S(T)$ is connected. Let $\chi(T_0)$ and $d(T_0)$ be respectively the Euler characteristic of $\mathcal S(T_0)$ and the number of boundary components of $\mathcal S(T)$ that belong to $\mathcal S(T_0)$. Let $\lambda(T_0) = 2d(T_0) - \chi(T_0)$. 

We claim that every $T_0 \in Z$ satisfies $\lambda(T_0) \geq 0$, and that if $T_0 \in Z'$ then $\lambda(T_0) > 0$. Indeed we have $\chi(T_0) \leq 1$ because $\mathcal S(T_0)$ is not homeomorphic to a sphere. So assuming $\lambda(T_0) \leq 0$, we get $d(T_0) = 0$. Then $\chi(T_0) \neq 1$ for otherwise $\mathcal S(T_0)$ would be homeomorphic to a disk, would have no curved point in its interior, and would be bounded by a single geodesic loop issued of $I$, contradicting the formula of Gauss--Bonnet. So $\chi(T_0) = 0$. Then $T_0$ is a tube because $\mathcal S(T_0)$ is not homeomorphic to a torus. This proves the claim. 

Now for every $T_0 \in Z'$ let $b(T_0)$ be the number of boundary components of $\mathcal S(T_0)$. The claim implies $b(T_0) \leq 2 - \chi(T_0) \leq 2 + \lambda(T_0) \leq 3 \lambda(T_0)$. So $\sum_{T_0 \in Z'} b(T_0) \leq 3 \sum_{T_0 \in Z'} \lambda(T_0) \leq 3 \sum_{T_0 \in Z} \lambda(T_0) \leq 9(g+b)$. Therefore at most $9(g+b)$ loops in $I$ are incident to the surface of some $T_0 \in Z'$. If every other loop in $I$ is incident to the surfaces of two distinct $T_0, T_1 \in Z$ then we are done. Otherwise there is a loop $e \in I$ incident to the surface of only one $T_0 \in Z$. Because $T_0$ is a tube, $\mathcal S(T)$ is a homeomorphic to a torus, and $e$ is the only loop in $I$, so we are done. This proves the lemma.
\end{proof}

\begin{proof}[Proof of Lemma~\ref{lem:tubing23}]
We claim that in the application of \algoref{SimplifyTubes} the set $J$ contains at most $9(g+b)$ loops. This is true if step 1a is applied, for in this case $g = 1$ and $J$ contains either zero or two loops. And if step 1b is applied all but $9(g+b)$ loops in $J'$ are incident to two distinct connected component of $\mathcal S(T) \setminus J'$ whose corresponding sub-portalgons of $T$ are tubes, by Lemma~\ref{L:packing}. Those loops are not retained in $J$. This proves the claim.

In the particular case where the surface $\mathcal S(T)$ is homeomorphic to a torus, and where $T^1$ contains exactly one vertex incident to loop a edge, the application of \algoref{SimplifyTubes} does nothing and $T = T'$. In this case the lemma is proved. In all other cases if a vertex $v$ of $T'^1$ lies in the interior of $\mathcal S(T')$ and is incident to a loop edge in $T'^1$, then $v$ is the base vertex of some loop in $J$. Indeed $v$ would otherwise have been deleted by \algoref{SimplifyTubes} when replacing a tube by a biface. There are at most $9(g+b)$ such vertices by our claim. This proves the lemma.
\end{proof}

\subsection{Analysis of \algoref{Gardening}}

Right after applying \algoref{Gardening} the topology of $\mathcal S(R_A)$, the surface of the active region, is bounded by the topology of $\mathcal S(R)$, the whole surface:

\begin{lemma}\label{lem:gardening}
Let $g$ and $b$ be the genus and the number of boundary components of $\mathcal S(R)$. The genus of $\mathcal S(R_A)$ is smaller than or equal to $g$. And right after applying \algoref{Gardening} $\mathcal S(R_A)$ has at most $10(g+b)$ boundary components.
\end{lemma}

\begin{proof}
The genus of $\mathcal S(R_A)$ is smaller than or equal to $g$. It is the number of boundary components of $\mathcal S(R_A)$ that we must handle. Each boundary component of $\mathcal S(R_A)$ is either a boundary component of $\mathcal S(R)$, and there are $b$ of them, or it is a separating loop. We bound the the number of separating loops adjacent to $\mathcal S(R_A)$, so let $I$ contain those loops. Each $e \in I$ is incident to two connected components of $\mathcal S(R) \setminus I$: one of them is in $\mathcal S(R_A)$, the other is not. The component in $\mathcal S(R_A)$ is not the surface of a tube because \algoref{Gardening} was just applied. So $I$ contains at most $9(g+b)$ loops by Lemma~\ref{L:packing}. We proved that $\mathcal S(R_A)$ has at most $10(g+b)$ boundary components.
\end{proof}

\subsection{Proof of Proposition~\ref{prop:bound number vertices}}

\begin{proof}[Proof of Proposition~\ref{prop:bound number vertices}]
We will prove that the number of vertices of $R_A^1$ is $O(n)$ throughout the execution. This will prove the lemma for then the number of edges of $R_A^1$ is also $O(n)$ by Lemma~\ref{lem:euler edges}, because the genus of $\mathcal S(R_A)$ is $O(n)$, and so the number of polygons of $R_A$ is also $O(n)$.

Consider the input triangular portalgon $T$. Let $m$ be the number of vertices of $T^1$. Let $g$ and $b$ be the genus and the number of boundary components of $\mathcal S(T)$. Observe that $m \leq 3n$, $g \leq n$, and $b \leq n$. We will argue using $m$, $g$, and $b$ instead of $n$. There are two loops in the algorithm: the main loop, which repeats $N$ times, and the interior loop, which repeats 350 times within each iteration of the main loop.

First we consider a single iteration of the interior loop. Let $m_A$ be the number of vertices of $R_A^1$ at the begining of this iteration. Observe that the iteration does not insert any new vertex in $R_A^1$. We claim that if $m_A > 3000(g+b+m)$ then less than $167m_A/168$ vertices are in $R_A^1$ at the end of the loop. To prove the claim first observe that after each application of \algoref{Gardening} $\mathcal S(R_A)$ has at most $10(g+b)$ boundary components by Lemma~\ref{lem:gardening}. And the genus of $\mathcal S(R_A)$ is smaller than or equal to $g$. Now after the application of \algoref{SimplifyTubes} at most $9(g + 10(g+b)) \leq 99(g+b)$ vertices of $R_A^1$ lie in the interior of $\mathcal S(R_A)$ and are incident to a loop by Lemma~\ref{lem:tubing23}. This is still the case just before the application of \algoref{DeleteVertices}. Moreover, at this point, at most $m + 10(g+b)$ vertices of $R_A^1$ lie on the boundary of $\mathcal S(R_A)$; indeed every such vertex is either a vertex of $T^1$, and there are at most $m$, or it is the base vertex of a separating loop, in which case it is the unique vertex in its boundary component of $\mathcal S(R_A)$, and there are at most $10(g+b)$. Altogether, just before the application of DeleteVertices, the number $\bar m_A$ of strong vertices of $R_A^1$ satisfies $\bar m_A \leq m + 109(g+b)$. If at this point $R_A^1$ has at most $24(g + \bar m_A)$ vertices then it already has less than $167m_A/168$ vertices because we assumed $m_A > 3000(g+b+m)$. Otherwise less than $167m_A/168$ vertices remain after \algoref{DeleteVertices} by Lemma~\ref{lem:deletes fraction of vertices}. In any case the claim is proved.

Now we prove the lemma by considering a single iteration of the main loop. Assuming that $R_A^1$ has more than $3000(g+b+m)$ vertices at the beginning of the iteration, we shall prove that in the end of the iteration the number of vertices of $R_A^1$ has decreased. To do so first observe that the iteration starts with \algoref{InsertVertices}, and this is the only moment where vertices are inserted. At this point the number of vertices of $R_A^1$ is multiplied by less than $8$ by Lemma~\ref{lem:insert how many vertices}. And by our claim, as long as the number of vertices exceeds $3000 (g+b+m)$ it is divided by more than $168/167$ by each iteration of the interior loop. There are 350 iterations of the interior loop, and $8 < (168/167)^{350}$. This proves the lemma.
\end{proof}


\section{Appendix of Sections~\ref{sec:enclosure}~and~\ref{sec:geometric analysis}: enclosure}\label{app:enclosure}


\subsection{Proof Lemma~\ref{prop:edges}}

In this section we prove Lemma~\ref{prop:edges}, which we restate for convenience:

\lemmapropedges*

\begin{proof}[Proof of Lemma~\ref{prop:edges}]
Let $t > 1$. Assume that there is a loop $\gamma$, based at a point $x$, that encloses $f$ by a factor of $t$. Then $\gamma$ encloses $e$ by a factor of $t$ because $\langle x \rangle_f \leq \langle x \rangle_e$.
\end{proof}

\subsection{Proof of Proposition~\ref{prop:polygons}}

In this section we prove Proposition~\ref{prop:polygons}, which we restate for convenience:

\proppolygons*

First we need a lemma:

\begin{lemma}\label{rebasing}
In $S$, let $e$ and $f$ be two segments whose relative interiors are disjoint, and let $\gamma$ be a geodesic loop. Assume that $\gamma$ encloses $e$ by a factor of $t > 2$, and that $\gamma$ intersects $f$ at a point $y$ such that $\langle y \rangle_f > \ell(\gamma)$. Rebase $\gamma$ at $y$, and let $\gamma'$ be the geodesic loop homotopic to it. Then $\gamma'$ meets $y$ on both sides of $f$.
\end{lemma}

We emphasize that, in Lemma~\ref{rebasing}, the loops $\gamma$ and $\gamma'$ have distinct basepoints, and that they may not be geodesic at their basepoints.

\begin{proof}
We have $\ell(\gamma') \leq \ell(\gamma)$ so $\ell(\gamma') < \langle y \rangle_f$, and so $\gamma'$ is in general position with $f$. We prove the lemma by contradiction, so assume that $\gamma'$ meets $y$ only on the right side of $f$, for some direction of $f$. In the universal covering space $\widetilde S$ of $S$, consider a lift $\widetilde f$ of $f$. Let $\widetilde y$ be the lift of $y$ that belongs to $\widetilde f$. Because the interior of $\widetilde S$ is flat, there is a geodesic $\widetilde L$, containing $\widetilde f$, such that on both ends $\widetilde L$ is either infinite or reaches the boundary of $\widetilde S$. Then $\widetilde L$ separates $\widetilde S$ in two connected components. The two lifts of $\gamma'$ incident to $\widetilde y$ meet $\widetilde y$ on the right side of $\widetilde f$ by assumption, and they are otherwise disjoint from $\widetilde L$. In particular, their other endpoints lie on the right side of $\widetilde L$.

We have $\ell(\gamma) < \langle y \rangle_f$ so $\gamma$ is in general position with $f$. Direct $\gamma$ so that $\gamma$ crosses $f$ from right to left at $y$, and write $\gamma$ as the concatenation of two paths $\gamma_0$ and $\gamma_1$ respectively before and after its crossing at $y$. There is a lift $\widetilde \gamma_1$ of $\gamma_1$ that leaves $\widetilde y$ on the left of $\widetilde f$. And $\widetilde \gamma_1$ is otherwise disjoint from $\widetilde L$, because the interior of $\widetilde S$ is flat. Thus the endpoint $\widetilde x$ of $\widetilde \gamma_1$ lies on the left of $\widetilde L$. There is a lift $\widetilde \gamma_0$ of $\gamma_0$ that starts at $\widetilde x$. And $\widetilde \gamma_0$ is otherwise disjoint from $\widetilde \gamma_1$ because $\gamma$ meets $x$ on both sides of $e$, and because the interior of $\widetilde S$ is flat. By the previous paragraph, the endpoint of $\widetilde \gamma_0$ lies on the right side of $\widetilde L$, so $\widetilde \gamma_0$ intersects $\widetilde L$. Cut $\widetilde \gamma_0$ at its first intersection point $\widetilde z$ with $\widetilde L$. Let $\widetilde I$ be the sub-segment of $\widetilde L$ between $\widetilde y$ and $\widetilde z$. The concatenation of the prefix of $\widetilde \gamma_0$ ending at $\widetilde z$, of $\widetilde I$, and of $\widetilde \gamma_1$ is a simple closed curve $\widetilde C$. At $\widetilde x$, there is a portion of $\widetilde e$ that enters the bounded side of $\widetilde C$, because $\gamma$ meets $x$ on both sides of $e$. This portion of $\widetilde e$ can be extended into a geodesic $\widetilde p$ that meets $\widetilde C$ at some point $\widetilde v$, because the interior of $\widetilde S$ is flat. Then $\widetilde v$ belongs to the relative interior of $\widetilde I$. We claim that $\widetilde v$ belongs to the relative interiors of both $\widetilde e$ and $\widetilde f$, which is a contradiction because the relative interiors of $e$ and $f$ are disjoint. To prove the claim, first observe that the distance between $\widetilde y$ and $\widetilde z$ in $\widetilde S$ is at most $\ell(\gamma)$, and this distance is equal to the length of $\widetilde I$, because the interior of $\widetilde S$ is flat. So the sub-segment of $\widetilde I$ between $\widetilde y$ and $\widetilde v$ is no longer than $\ell(\gamma) < \langle y \rangle_f$, and is thus included in the relative interior of $\widetilde f$. Also, the distance between $\widetilde v$ and $\widetilde x$ is smaller than or equal to $2 \ell(\gamma) \leq 2 \langle x \rangle_e / t < \langle x \rangle_e$, so $\widetilde p$ is included in the relative interior of $\widetilde e$.
\end{proof}

\begin{figure}[ht]
    \centering
    \includegraphics[width=0.6\linewidth]{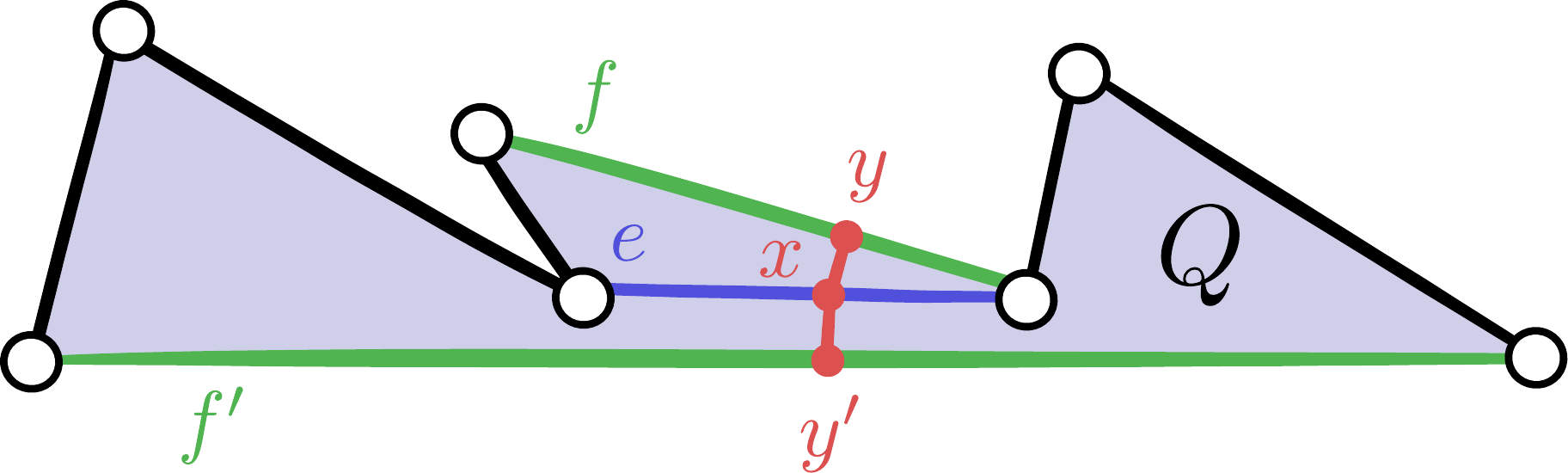}
    \caption{The setting of Lemma~\ref{inapolygon}.}
    \label{fig:sarc proof}
\end{figure}

The proof of Proposition~\ref{prop:polygons} also relies on the following construction. See Figure~\ref{fig:sarc proof}. In the Euclidean plane $\mathbb R^2$ let $Q$ be a polygon with more than three vertices. Let $e$ be a shortcut of $Q$. Let $f$ and $f'$ be sides of $Q$ separated by $e$ along the boundary of $Q$. Let $x$ be a point in the relative interior of $e$. Let $y$ and $y'$ be points that lie on respectively $f$ and $f'$ (possibly vertices of $Q$), and do not lie on $e$. Consider the segments $p$ and $p'$ between $x$ and respectively $y$ and $y'$. Assume that the relative interiors $p$ and $p'$ are included in the interior of $Q$. Then:

\begin{lemma}\label{inapolygon}
Let $t > 6$. If $\ell(p) \leq \langle x \rangle_e / t$ and $\ell(p') \leq \langle x \rangle_e / t$, then at least one of $f$ and $f'$, say $f$, is such that $\langle y \rangle_f \geq (1 -4/t) \cdot \langle x \rangle_e$ and $\ell(f) \geq (1-4/t)\cdot \ell(e)$.
\end{lemma}

\begin{proof}
Assume without loss of generality that $e$ is horizontal, that $f$ lies above $e$, and that $x$ is the origin $(0,0) \in \mathbb R^2$. Then $x$ cuts $e$ into two segments $e_0$ and $e_1$, respectively the right and left one. Let $v_0$ and $v_1$ be respectively the right and left endpoints of $e$. Consider the following algorithm in three phases. In the first phase consider the point $z = x$ and move $z$ along $p$. Doing so, consider the segments from $z$ to $v_0$ and $v_1$. If moving $z$ makes the relative interior of one of those two segments intersect $\partial Q$, then stop: this is a break condition. Also break if $z$ reached $y$ and $y$ is a vertex of $Q$. Otherwise the algorithm enters its second phase. Then $y$ cuts $f$ in two segments $f_0$ and $f_1$, where $f_0$ is on the right of $y$ as seen from the path $p$ directed from $x$ to $y$. In phase two move $z$ along $f_0$ or $f_1$, choosing carefully which segment to move along so that the second coordinate of $z$ does not increase. We assume without loss of generality that $z$ moves along $f_0$, by flipping the figure horizontally otherwise. Move along $f_0$ by a distance of $(1-4/t) \ell(e_0)$, but break if $z$ reaches the right end-vertex of $f$, or if the relative interior of the segment between $z$ and $v_0$ intersects $\partial Q$. If the algorithm did not break, it enters its third and final phase. In this phase put $z$ back on $y$, and move it along the other sub-segment of $f$, here $f_1$, by a distance of $(1-4/t) \ell(e_1)$, breaking if $z$ reaches the left end-vertex of $f$, or if the relative interior of the segment between $z$ and $v_1$ intersects $\partial Q$. 

If the algorithm did not break then $\ell(f) \geq (1-4/t) \ell(e)$ and $\langle y \rangle_f \geq (1-4/t) \langle x \rangle_e$ and we are done. Otherwise, if the algorithm broke, consider the triangle $\Delta$ between $v_0$, $v_1$, and $z$. The break conditions ensure that the interior of $\Delta$ is included in the interior of $Q$, and that there is a vertex $w$ of $Q$ that lies on $\partial \Delta$ and not on $e$. We claim that the inner-angles of $\Delta$ at $v_0$ and $v_1$ are both strictly smaller than $\pi/4$. We prove this claim by considering the coordinates $(\alpha,\beta) \in \mathbb R \times [0,+\infty[$ of $z$, and the coordinates $(\ell(e_0), 0)$ and $(-\ell(e_1), 0)$ of $v_0$ and $v_1$ respectively, and by proving that the invariants $\ell(e_0) - \alpha > \beta$ and $\alpha + \ell(e_1) > \beta$ hold at any time during the algorithm. Let $m = \min(\ell(e_0), \ell(e_1)) = \langle x \rangle_e$. In the first phase $|\alpha| \leq m/t$ and $0 \leq \beta \leq m/t$, so the invariants hold because $t > 2$. In the second phase $\beta$ does not increase and $\alpha$ does not decrease. Moreover $\alpha$ does not increase by more than $\ell(e_0)(1-4/t)$ so the invariants hold. If the second phase ends without breaking then the absolute slope $\lambda$ of the line supporting $f$ is smaller than or equal to $1/(t-5)$. Indeed during the second phase $\beta$ decreased by at most $m/t$ while $z$ moved a distance $\ell(e_0)(1-4/t)$, so $\alpha$ increased by at least $\ell(e_0)(1-4/t) - m/t$, and so $1/\lambda \geq \ell(e_0)(1-4/t)t/m - 1 \geq t - 5$. In the third phase $\alpha \geq -m/t -\ell(e_1)(1-4/t)$ and $\beta \leq m/t + \lambda \ell(e_1)(1-4/t)$ so $\alpha + \ell(e_1) \geq 3\ell(e_1)/t > \beta$ because $t >6$. Also $\beta$ increases less than $\alpha$ decreases because $\lambda < 1/2$, so $\ell(e_0) - \alpha > \beta$ remains true. This proves the claim.

Applying the algorithm to $p'$ and $f'$ on the other side of $e$, either the algorithm does not break in which case $\ell(f') \geq (1-4/t) \ell(e)$, $\langle y' \rangle_{f'} \geq (1-4/t) \langle x \rangle_e$, and we are done. Or the algorithm breaks and we get similarly a triangle $\Delta'$ and a vertex $w'$ of $P$. The inner angles of $\Delta'$ at $v_0$ and $v_1$ are also both strictly smaller than $\pi/4$, so the relative interior of the segment between $w$ and $w'$ is included in the interior of the quadrilateral formed by $\Delta$ and $\Delta'$, and is strictly shorter than $e$. This segment is a vertex-to-vertex arc of $Q$ shorter than $e$, a contradiction.
\end{proof}

\begin{proof}[Proof of Proposition~\ref{prop:polygons}]
Let $t > 6$. Assume that there is a geodesic loop $\gamma$ that encloses $e$ by a factor of $t$. Let $x$ be the basepoint of $\gamma$. In the Euclidean plane, consider the polygon $Q$ corresponding to $F$. Let $\widehat e$ and $\widehat x$ be the pre-images of $e$ and $x$ in $Q$. Consider the prefix and the suffix of $\gamma$ that leave $x$ on both sides of $e$ to meet $\partial F$, and their pre-image paths in $Q$ that meet two boundary edges $\widehat f$ and $\widehat f'$ of $Q$, at respective points $\widehat y$ and $\widehat y'$. By Lemma~\ref{inapolygon}, one of those two points, say $\widehat y$ without loss of generality, is such that $\langle \widehat y \rangle_{\widehat f} \geq (1-4/t)\langle \widehat x \rangle_{\widehat e}$ and $\ell(\widehat f) \geq (1-4/t) \ell(\widehat e)$. Also $\widehat f$ projects to a boundary edge $f$ of $F$, and $\widehat y$ projects to a point $y$ in the relative interior of $f$. Now rebase $\gamma$ at $y$, and consider the geodesic loop $\gamma'$ homotopic to it (where the basepoint at $y$ is fixed by the homotopy). Then $\ell(\gamma') \leq \ell(\gamma) = \langle x \rangle_e / t < \langle y \rangle_f / (t-4)$. In particular $\ell(\gamma') < \langle y \rangle_f$ because $t > 5$. And $\gamma'$ meets $y$ on both sides of $f$ by Lemma~\ref{rebasing}, because $t > 2$.
\end{proof}

\subsection{Interior of a thick biface}


In this section we prove the following:

\begin{restatable}{proposition}{propthicktube}\label{prop:thick tube}
Assume that $S$ contains the surface of a thick biface $B$, and let $e$ be one of the two interior edges of $B^1$. Assume that $c_S(e) > 6$. Then there is a boundary edge $f$ of $B^1$ such that $c_S(f) \geq c_S(e)-5$ and $\ell(f) \geq (1-4/c_S(e)) \cdot \ell(e)$.
\end{restatable}

First we need a lemma:

\begin{lemma}\label{L:angle in good biface}
Let $B$ be a good biface. Among the two interior edges of $B^1$ let $f$ be a longest one. Let $F$ be the face of $B^1$ adjacent to $f$. Each corner of $F$ incident to $f$ has angle smaller than or equal to $\pi/2$.
\end{lemma}

\begin{proof}
Among the two interior edges of $B^1$ let $e$ be a shortest one, and let $g \neq e$ be the other one. Then $e$, $g$, and $f$ are the sides of $F$. The angle at the corner of $F$ between $f$ and $g$ is smaller than $\pi/2$ because $\ell(e) \leq \ell(g)$. Now consider the corner $c$ between $f$ and $e$. Cut $\mathcal S(B)$ open along $e$ and consider the resulting quadrilateral $Q$ in the plane. The edge $f$ of $B^1$ corresponds to a side $\widehat f$ of $Q$, the edge $e$ corresponds to two opposite sides $\widehat e$ and $\widehat e'$, and the edge $g$ corresponds to a vertex-to-vertex arc $\widehat g$ of $Q$. Also the other boundary edge $f' \neq f$ of $B^1$ corresponds to the side $\widehat f'$ of $Q$ opposite to $\widehat f$. And the corner $c$ corresponds to the corner $\widehat c$ of $Q$ between $\widehat e$ and $\widehat f$. Let $\widehat d$ be the corner of $Q$ opposite to $\widehat c$, between $\widehat e'$ and $\widehat f'$. Assume by contradiction that the angle at $\widehat c$ is greater than $\pi/2$. We have $\ell(\widehat e) = \ell(\widehat e')$ and $\ell(\widehat f) \geq \ell(\widehat f')$ so the angle at $\widehat d$ is greater than or equal to the angle at $\widehat c$, and in particular is also greater than $\pi/2$. The two other angles of $Q$ are smaller than $\pi$, so $Q$ is convex and admits a diagonal $p \neq \widehat g$. Consider the unique circle $C$ that admits $\widehat g$ as a diameter. Then the two endpoints of $p$ lie in the interior of $C$. So $p$ is shorter than $\widehat g$. This contradicts the assumption that $B$ is good.
\end{proof}

\begin{proof}[Proof of Proposition~\ref{prop:thick tube}]
Among the two interior edges of $B^1$ let $g$ be a shortest one. Among the two boundary edges of $B^1$ let $g'$ be a longest one. Then $\ell(g) \leq \ell(g')$ because $B$ is thick. We claim that if $c_S(g) > 2$, then $c_S(g') \geq c_S(g)-1$. To prove the claim let $t > 2$ and assume that there is a loop $\gamma$ that encloses $g$ by a factor of $t$ in $S$. Let $x$ be the basepoint of $\gamma$. Let $F$ be the face of $B^1$ adjacent to $g'$. Around $x$ there is a portion of $\gamma$ that enters $F$.  This portion of $\gamma$ must leave $F$ by a point $y$ of $g'$ because the angle of $F$ between $g$ and $g'$ is smaller than or equal to $\pi/2$ by Lemma~\ref{L:angle in good biface}, because $\ell(g) \leq \ell(g')$, and because $\ell(\gamma) = \langle x \rangle_g / t < \langle x \rangle_g / \sqrt 2$. Then $\langle y \rangle_{g'} \geq \langle x \rangle_g - \ell(\gamma) = (1 - 1/t) \langle x \rangle_g$ by triangular inequality and because $\ell(g) \leq \ell(g')$. Rebase $\gamma$ at $y$, and consider the geodesic loop $\gamma'$ homotopic to it (where the homotopy fixes the basepoint at $y$). Then $\ell(\gamma') \leq \ell(\gamma) = \langle x \rangle_g / t \leq \langle y \rangle_{g'}/(t-1)$. And $\gamma'$ encloses $g'$ by Lemma~\ref{rebasing}, because $t > 2$. This proves the claim.

If $e = g$ we are done by our claim, so assume that $e$ is a longest interior edge of $B^1$. Deleting $e$ merges the two faces of $B^1$ into a single face $F'$ of which $e$ is a shortcut, because $B$ is good. So Proposition~\ref{prop:polygons} applies because $c_S(e) > 6$: there is a boundary edge $f$ of $F'$ such that $c_S(f) \geq c_S(e)-4$ and $\ell(f) \geq (1-4/c_S(e))\ell(e)$. If $f$ is a boundary edge of $B^1$ we are done. Otherwise $f = g$ so $\ell(g') \geq \ell(f) \geq (1-4/c_S(e))\ell(e)$ and $c_S(g') \geq c_S(f)-1 \geq c_S(e)-5$ by our claim because $c_S(e) > 6$. This proves the proposition.
\end{proof}

\subsection{Proof of Proposition~\ref{L:biface boundary hapiness}}

In this section we prove Proposition~\ref{L:biface boundary hapiness}, which we restate for convenience:

\bifaceboundaryhapiness*

First we need two lemmas:

\begin{lemma}\label{L:angle in thin biface}
Let $B$ be a thin biface. Among the two interior edges of $B^1$ let $e$ be a shortest one. Each of the four corners between $e$ and the boundary of $\mathcal S(B)$ has angle greater than $\pi/4$.
\end{lemma}

\begin{proof}
Assume by contradiction that there is a corner $c$ between $e$ and a boundary edge $f$ of $B^1$ whose angle is smaller than or equal to $\pi/4$. Cut $\mathcal S(B)$ open along $e$ and embed the resulting quadrilateral $Q$ in the plane, isometrically. The edge $e$ corresponds to two opposite sides $\widehat e$ and $\widehat e'$ of $Q$. The edge $f$ corresponds to one of the other two sides of $Q$, that we call $\widehat f$. The vertex $v$ of the corner $c$ corresponds to the two end-vertices of $\widehat f$: let $\widehat v$ be the one incident to $\widehat e$, and let $\widehat v'$ be the one incident to $\widehat e'$. Without loss of generality the corner $c$ corresponds to the corner of $Q$ at $\widehat v$, whose angle is thus smaller than or equal to $\pi/4$. Consider the orthogonal projection $x$ of $\widehat v'$ on the line containing $\widehat e$. Then $x$ belongs to $\widehat e$ because $\widehat e$ is longer than $\widehat f$, as $B$ is thin. The segment $p$ between $x$ and $\widehat v'$ is shorter than the portion of $\widehat e$ between $x$ and $\widehat v$. Also $p$ is included in $Q$ because $\widehat e$ and $\widehat e'$ are longer than $\widehat f$. Thus $p$ projects to a path that shortcuts $e$, contradicting the fact that $B$ is a good biface.
\end{proof}

\begin{lemma}\label{L:path length in thin biface} 
In $\mathcal S(B)$ every path $p$ between the two boundary components of $\mathcal S(B)$ is such that $\ell(p) \geq \ell(e)/2$.
\end{lemma}

\begin{proof}
Without loss of generality one of the two endpoints of $p$ (at least) is a vertex $v$ of $B^1$. Consider the other endpoint $x$ of $p$, and the vertex $w \neq v$ of $B^1$. There is a path $q$ from $x$ to $w$ in the boundary of $\mathcal S(B)$. Without loss of generality $\ell(q) \leq \ell(e)/2$ because $B$ is thin. Also $e$ is a shortest path because $B$ is good. So $\ell(p) + \ell(q) \geq \ell(e)$. We proved $\ell(p) \geq \ell(e)/2$.
\end{proof}

\begin{proof}[Proof of Proposition~\ref{L:biface boundary hapiness}]
Among the two interior edges of $B^1$ let $e$ be a shortest one. Let $f$ be any one of the two boundary edges of $B^1$. We have $\ell(e) \geq \ell(f)$ because $B$ is thin. Assume by contradiction that there is in $S$ a loop $\gamma$ that encloses $f$ by a factor of $t > 2$. Let $x$ be the basepoint of $\gamma$. There is a portion of $\gamma$ that leaves $x$ and enters the interior of $\mathcal S(B)$. This portion of $\gamma$ cannot leave $\mathcal S(B)$ via the other boundary edge of $\mathcal S(B)$, for otherwise $\ell(\gamma) \geq \ell(e)/2$ by Lemma~\ref{L:path length in thin biface}, so $\ell(\gamma) > \langle x \rangle_f / t$, a contradiction. Then $\gamma$ intersects $e$. And $f$ and $e$ have a corner whose angle is smaller than $\pi/4$ because $\ell(\gamma) < \langle x \rangle_f /2$. This contradicts Lemma~\ref{L:angle in thin biface}.
\end{proof}

\subsection{Proof of Proposition~\ref{P:lambda and hapiness global}}

In this section we prove Proposition~\ref{P:lambda and hapiness global}, which we restate for convenience:

\proplambdaandhapinessglobal*

First we need a few lemmas.

\begin{lemma}\label{shortest path encloses}
Let $t > 1$. Assume that there is a shortest path whose relative interior crosses the relative interior of $e$ twice in the same direction, at points $x$ and $y$. If the sub-segment of $e$ between $x$ and $y$ is shorter than $\langle x \rangle_e / 2 t$ then $c_S(e) > t$.
\end{lemma}

\begin{proof}
Consider the portion $p$ of the shortest path that starts just before its crossing at $x$, and ends just before its crossing at $y$. Consider also the geodesic path $q$ that runs parallel to the sub-segment of $e$ from $y$ to $x$, such that the concatenation of $p$ and $q$ forms a loop $\gamma$. Base $\gamma$ at $x$. There is a unique geodesic loop $\gamma'$ homotopic to $\gamma$ (where the base-point at $x$ is fixed in the homotopy) because the interior of $S$ is flat. We have that $\gamma'$ is not the constant loop based at $x$; for otherwise $\gamma$ would be contractible, so $p$ would be homotopic to the reversal of $q$, and so $p$ would actually be equal to the reversal of $q$ because the interior of $S$ is flat, a contradiction. Moreover $\gamma'$ is shorter than $\langle x \rangle_e / t$; indeed $\gamma'$ is not longer than $\gamma$, $q$ is shorter that $\langle x\rangle_e / 2t$ by assumption, and $p$ is not longer than $q$ because $p$ is a shortest path. Then $\gamma'$ is in general position with $e$. We shall prove that $\gamma'$ meets $x$ on both sides of $e$. This will prove the lemma for then $\gamma'$ will enclose $e$ by a factor of $\langle x \rangle_e / \ell(\gamma') > t$.

Let us prove that. Orient $e$ so that $\gamma$ crosses $e$ from right to left. Consider the universal covering space $\widetilde S$ of $S$, and a lift $\widetilde e$ of $e$ in $\widetilde S$. The interior of $\widetilde S$ being flat, there is a geodesic $\widetilde L$, containing $\widetilde e$, such that on both ends $\widetilde L$ is either infinite or reaches the boundary of $\widetilde S$. And $\widetilde L$ separates $\widetilde S$ in two connected components. Now let $\widetilde x$ be the lift of $x$ in $\widetilde e$. There are two lifts of $\gamma'$ incident to $\widetilde x$: one lift $\widetilde \gamma'_0$ starts at $\widetilde x$, the other lift $\widetilde \gamma'_1$ ends at $\widetilde x$. Let $\widetilde a_0$ be the endpoint of $\widetilde \gamma'_0$, and let $\widetilde a_1$ be the startpoint of $\widetilde \gamma'_1$. We claim that $\widetilde a_0$ lies strictly to the left of $\widetilde L$, and that $\widetilde a_1$ lies strictly to the right of $\widetilde L$. This claim implies that $\widetilde \gamma'_0$ meets $\widetilde x$ on the left of $\widetilde e$, and that $\widetilde \gamma'_1$ meets $\widetilde x$ on the right of $\widetilde e$, which implies that $\gamma'$ meets $x$ on both sides of $e$.

Let us prove the claim. First we prove that $\widetilde a_0$ lies strictly to the left of $\widetilde L$. To do so consider also the lift $\widetilde p$ of $p$ that starts at $\widetilde x$, and the lift $\widetilde q$ of $q$ that starts at the endpoint of $\widetilde p$. The endpoint of $\widetilde q$ is $\widetilde a_0$ because the concatenation of $\widetilde p$ and $\widetilde q$ is a lift of $\gamma$, and because $\gamma$ is homotopic to $\gamma'$. By definition $\widetilde p$ leaves $\widetilde x$ on the left of $\widetilde L$. Also $\widetilde p$ is disjoint from $\widetilde L$ except for its startpoint at $\widetilde x$, the interior of $S$ being flat. Moreover $\widetilde q$ is disjoint from $\widetilde L$. For otherwise $\widetilde q$ would intersect $\widetilde L$ at a point $\widetilde z$ whose distance to $\widetilde x$ would be smaller than $\langle x \rangle_e / t$. But then the sub-segment of $\widetilde L$ between $\widetilde z$ and $\widetilde x$ would be no longer, and so would be included in $\widetilde e$. In particular $\widetilde q$ and $\widetilde e$ would intersect, a contradiction. This proves that $\widetilde a_0$ lies strictly to the left of $\widetilde L$.

To prove that $\widetilde a_1$ lies strictly to the right of $\widetilde L$, consider the lift $\widetilde y$ of $y$ in $\widetilde e$, and the lift $\widetilde p_1$ of $p$ that ends at $\widetilde y$. Then the startpoint of $\widetilde p_1$ is $\widetilde a_1$, and it lies strictly to the right of $\widetilde L$ because $\widetilde p_1$ meets $\widetilde y$ on the right of $\widetilde L$, and because $\widetilde p_1$ is otherwise disjoint from $\widetilde L$. This proves the claim, and the lemma.
\end{proof}

\begin{lemma}\label{L:lambda and hapiness local}
We have $h_S(e) = O(c_S(e) \cdot (1+\log \lceil \ell(e)/s\rceil))$.
\end{lemma}

\begin{proof}
Let $t > 1$. Assume $h_S(e) > 12 t \cdot (3 + \log \lceil \ell(e)/s\rceil )$. We will prove $c_S(e) \geq t$, and this will prove the lemma. In $S$ there is a shortest path $p$ that intersects $e$ more than $12 t \cdot (3 + \log \lceil \ell(e)/s\rceil )$ times. Cut $e$ at its middle point. One of the two resulting sub-segments of $e$, say $f$, intersects $p$ more than $6 t \cdot (3 + \log \lceil \ell(e)/s\rceil )$ times. Partition $f$ into sub-segments $f_0, f_1, \dots, f_n$ with $n \leq 2 + \log \lceil \ell(e)/s\rceil $, where the sub-segment $f_0$ contains the points $x \in f$ such that $\langle x \rangle_e \leq s/4$, and where for every $1 \leq i \leq n$ the sub-segment $f_i$ contains the points $x \in f$ such that $2^{i-3} s \leq \langle x \rangle_e \leq 2^{i-2} s$. There is $0 \leq i \leq n$ such that $p$ intersects $f_i$ more than $6 t$ times. Then the relative interior of $p$ crosses $f_i$ twice (at least) in the same direction at points $x$ and $y$, such that the sub-segment of $f_i$ between $x$ and $y$ is shorter than $2^{i-4} s / t$, because $\ell(f_i) \leq 2^{i-3} s$. Also $i \geq 1$ as no shortest path crosses $f_0$ twice, because $\ell(f_0) < s/2$, and because the interior of $S$ is flat. In particular $\langle x \rangle_e \geq 2^{i-3} s$. Then $c_S(e) \geq t$ by Lemma~\ref{shortest path encloses}.
\end{proof}

\begin{lemma}\label{L:length and enclosure}
We have $\ell(e) = O( c_S(e) \cdot n \lceil L/s\rceil L)$.
\end{lemma}

\begin{proof}
Let $t > 1$. Assume $\ell(e) \geq 600 t \cdot n\lceil L/s \rceil L$. We will prove that $c_S(e) \geq t$, and this will prove the proposition. To do so let $d = 120n \lceil L/s \rceil L$. Cut $e$ into three segments, a middle segment $e_0$ of length $d$, and two peripheral segments each longer than $2t\cdot d$. We claim that there is in $S$ a shortest path crossing the relative interior of $e_0$ twice in the same direction. This claim implies $c_S(e) \geq t$ by Lemma~\ref{shortest path encloses}, which proves the proposition. 

Let us prove the claim. Consider a triangulation $T$ of $S$ with $n$ triangles, whose edges are all smaller than $L > 0$. Cut each edge of $T$ into $2 \lceil L/s \rceil$ equal-length segments, that we shall call \emph{doors}. Each door is smaller than or equal to half the systole of $S$ so it is a shortest path. There are at most $6 n \lceil L/s \rceil$ doors because $T$ has at most $3n$ edges. Each sub-segment $e_1$ of length $4L$ of $e_0$ contains in its relative interior three points $x_0,x_1,x_2$ in this order such that $x_0 \notin p$, $x_1 \in p$, and $x_2 \notin p$ for some door $p$. The relative interior of $e_0$ intersects at least $30 n \lceil L/s \rceil$ times doors this way, so there is a door $p$ intersected at least $5$ times by the relative interior of $e_0$. Then each intersection is a single point ($p$ and $e_0$ do not overlap). Two of those intersection points may be endpoints of $p$, but otherwise the relative interior of $p$ crosses the relative interior of $e_0$ at least three times. So $p$ crosses $e_0$ twice in the same direction, which proves the claim, and the proposition.
\end{proof}

\begin{proof}[Proof of Proposition~\ref{P:lambda and hapiness global}]
We have $h_S(e) = O(c_S(e) \cdot (1+ \log \lceil \ell(e)/s\rceil))$ by Lemma~\ref{L:lambda and hapiness local}. Also $\ell(e)/s = O(c_S(e) \cdot n \cdot \lceil L/s\rceil^2)$ by Lemma~\ref{L:length and enclosure}. So $\log(\lceil \ell(e)/s\rceil) =O(1 + \log c_S(e) + \log(n) + \log \lceil L/s\rceil)$. This proves the proposition.
\end{proof}


\section{Appendix of Section~\ref{sec:geometric analysis}: proof of Lemma~\ref{lem:geom simplify tubes}}\label{sec:app geom simplify tubes}

\begin{proof}[Proof of Lemma~\ref{lem:geom simplify tubes}]
Assume that $e'$ was not already an edge of $R_A^1$ for otherwise there is nothing to do. Then there is a good biface $B$ computed by step 2 of \algoref{SimplifyTubes} such that $e$ is one of the two interior edges of $B^1$. Also $B$ is thick, for $B$ has not been marked as inactive. So by Proposition~\ref{prop:thick tube} there is a boundary edge $e$ of $B^1$ such that $c_S(e) \geq c_S(e')-5$ and $\ell(e) \geq (1-4/c_S(e')) \ell(e')$.
\end{proof}

\section{Appendix of Section~\ref{sec:proofofTcorealgo}: proof of Lemma~\ref{lem:bound systole}}\label{app:bound systole}

\begin{proof}[Proof of Lemma~\ref{lem:bound systole}]
Otherwise one of the two loops of $X^1$ that constitute the boundary of $\mathcal S(X)$ is contractible in $S$. So this loop bounds a topological disk in $S$ by a result of Epstein~\cite[Theorem~1.7]{e-c2mi-66}. The interior of the disk is flat, and its boundary is geodesic except possibly at the basepoint of the loop. This contradicts the formula of Gauss--Bonnet.
\end{proof}


\section{Appendix: proof of Proposition~\ref{thm:improving}}\label{app:extension}

In this section we prove Proposition~\ref{thm:improving}, which we restate for convenience:

\thmimproving*

We deduce Proposition~\ref{thm:improving} from Proposition~\ref{T:core algorithm}, essentially by cutting out caps around the singularities in the interior of the surface, by applying Proposition~\ref{T:core algorithm} to the truncated surface, and by putting the caps back afterward. See Figure~\ref{fig:caps}.

\begin{figure}[ht]
    \centering
    \includegraphics[width=\linewidth]{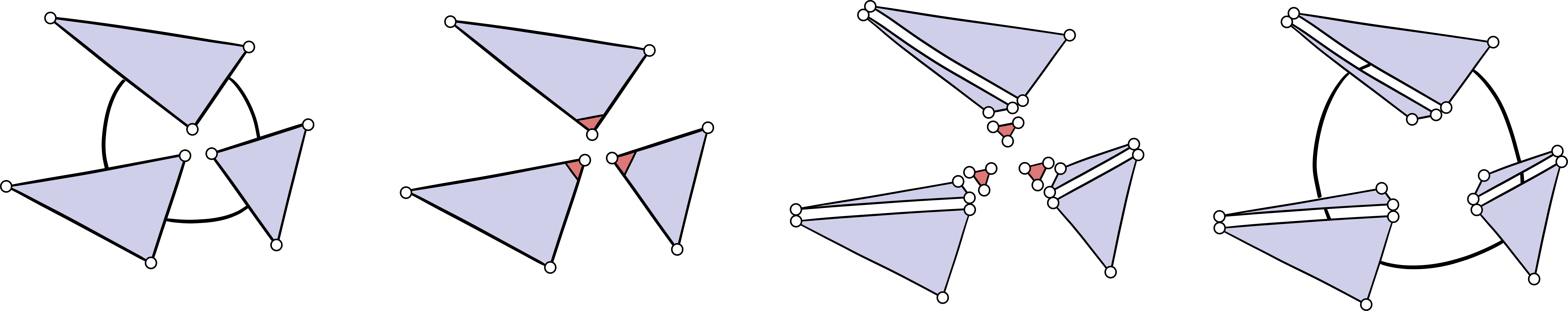}
    \caption{Cutting out a cap in the proof of Proposition~\ref{thm:improving}.}
    \label{fig:caps}
\end{figure}

\begin{proof}[Proof of Proposition~\ref{thm:improving}]
Let $S := \mathcal S(T)$ be the surface of $T$. Let $d$ be the minimum height of the triangles of $T$. Given a vertex $v$ of $T^1$ in the interior of $S$, we define a region around $v$ in $S$, as follows. On every directed edge $e$ of $T^1$ whose tail is $v$, place a point at distance $d/6$ from the tail of $e$ along $e$. Link those $k \geq 1$ points in order (clockwise say, but counter-clockwise would do to) around $v$, using geodesic segments within the faces of $T^1$ incident to $v$. In each corner of $T^1$ incident to $v$ there is a newly created triangle incident to $v$. Those $k$ triangles define a region $C$ around $v$, which we call \emph{cap} of $v$. Importantly, every point in the cap of $v$ is at distance smaller than or equal to $d/6$ from $v$ in $S$. Also every segment $p$ tracing the boundary of $C$ satisfies $\ell(p) \geq d/6r$. To see that consider the face $F$ of $T^1$ containing $p$, and the two sides $e_0$ and $e_1$ of $F$ incident to $v$. For each $i$ consider the point on $e_i$ at distance $m := \min(\ell(e_0), \ell(e_1))$ from $v$ along $e_i$. Join those two points by a geodesic segment $q$ in $F$. Then $q$ is at least as long as the minimum height of the triangle corresponding to $F$, so $\ell(q) \geq m/r$. Moreover $\ell(p)/\ell(q) = d / (6m)$ by the theorem of Thales. This proves $\ell(p) \geq d / 6r$.

For the sake of analysis, given an arbitrary vertex $v$ of $T^1$ (possibly on the boundary of $S$), we define another kind of region around $v$. Link the middle points of the edges around $v$ in order around $v$. The resulting triangles around $v$ constitute the \emph{protected region} of $v$. Importantly, every path smaller than $d/2$ starting from $v$ must lie in the protected region of $v$. Indeed the relative interior of every geodesic path $p$ smaller than $d$ starting from $v$ is included in a single face or edge of $T^1$. Then every prefix of $p$ smaller than $\ell(p)/2$ lies in the protected region of $v$.

First construct in $O(n)$ time a triangular portalgon $T_0$ whose surface is $S$, as follows. Consider every singularity in the interior of $S$ (if any). This singularity is a vertex $v$ of $T^1$. Trace the boundary of the cap around $v$ in the faces of $T^1$. Then cut those faces along the trace, as in Figure~\ref{fig:caps}. Afterward some polygons of $T_0$ may not be triangles, so cut each polygon of $T_0$ into triangles along shortcuts. Now remove the triangles corresponding to the caps from $T_0$, and let $T_1$ be the resulting triangular portalgon. The surface $\mathcal S(T_1)$ is flat.

Our first claim is that the systole of $\mathcal S(T_1)$ is greater than or equal to $d/6r$. By contradiction assume that there is a non-contractible closed curve $\gamma$ in $\mathcal S(T_1)$ smaller than $d/6r$. Without loss of generality $\gamma$ intersects a vertex $w$ of $T_1^1$; indeed if such a $\gamma$ has minimum length and does not intersect any vertex of $T_1^1$ then it can be slided along the surface, without changing its length, until it intersects a vertex of $T_1^1$. If $w$ is a vertex of $T^1$, then $\gamma$ lies in the protected region around $w$, and so $\gamma$ is contractible in $\mathcal S(T_1)$, a contradiction. If $w$ is a vertex on the boundary of some cap $C$ removed, then $\gamma$ lies in the protected region around the central vertex of $C$. In that case $\gamma$ is at least as long as any edge of the boundary of $C$, so $\ell(\gamma) \geq d/6r$. This proves the first claim.

The number of triangles and the maximum side length of a triangle of $T_1$ may be greater than those of $T$, but only by a constant factor. Using the first claim and Proposition~\ref{T:core algorithm}, replace $T_1$ by a portalgon of $O(n \log (r))$ triangles, whose surface is that of $T_1$, and whose segment-happiness is $O(\log (n) \log^2 (r))$, all in $O(n \log^2(n) \log^2(r))$ time. Place back the caps on $\mathcal S(T_1)$, and return the resulting triangular portalgon $T'$.

The segment-happiness of $T'$ and the happiness of $T'$ do not differ by more than a constant factor because the polygons of $T'$ are all triangles. Our second claim is that the segment-happiness of $T'$, and thus the happiness of $T'$, is bounded by $O(n \log(n) \log^2(r))$. To prove the second claim, we call \emph{cap path} any shortest path in $S$ that lies in the closure of some cap. We call \emph{rogue path} any shortest path in $S$ whose relative interior is disjoint from the closures of the caps. Every rogue path intersects every edge of $T'^1$ at most $O(\log (n) \log^2 (r))$ times, because the segment-happiness of $T_1$ is $O(\log (n) \log^2 (r))$. Also every cap path intersects every edge of $T'^1$ at most once. Now consider a shortest path $p$ in $S$. Then $p$ uniquely writes as a sequence $X$ of alternatively cap paths and rogue paths. Also, there cannot be two distinct cap paths $q_0$ and $q_1$ in $X$ that both lie in the same cap $C$. For otherwise any point of $q_0$ would be at distance at most $d/3$ from any point of $q_1$. Also the subpath of $p$ between $q_0$ and $q_1$ contains a rogue path that must leave the protected region around the central vertex of $C$, and is thus longer than $d/2-d/6 = d/3$. That contradicts the fact that $p$ is a shortest path. We proved that there are at most $O(n)$ paths in $X$, each intersecting at most $O(\log(n) \log^2(r))$ times any given edge of $T'^1$. This proves the second claim, and the proposition.
\end{proof}


\section{Appendix: proof of Theorem~\ref{thm:main lower bound}}\label{app:lower bound}

In this section, we detail our construction for the lower bound, and we prove Theorem~\ref{thm:main lower bound}, which we restate for convenience:

\maintheoremlowerbound*

Recall that we obtain our results within the real RAM model of computation described by Erickson, van der Hoog, and Miltzow~\cite{erickson2022smoothing}. It is an extension of the standard integer word RAM, with an additional memory array storing reals, and with additional instructions. The available instructions are described in~\cite[Table 1]{erickson2022smoothing}. The arithmetic operations that can be performed by the machine on real numbers are addition, subtraction, multiplication, division, and square root.

Recall also that, when modifying a portalgon $T$, we actually modify our representation of $T$, using elementary operations that are easily seen to achievable by a real RAM. For example, consider, in the Euclidean plane, two triangles $\Delta$ and $\Delta'$, given by the coordinates of their vertices. Assuming that $\Delta$ and $\Delta'$ have respective sides $s$ and $s'$ of the same length, we consider the operation of displacing $\Delta$ in the plane so that afterward $\Delta$ and $\Delta'$ are side by side along $s = s'$, and we compute the coordinates of the vertices of $\Delta$ after the displacement. This operation can be achieved by a real RAM.~\footnote{In fact, this can be done without even using the square root operation. To see that, think of the two initial vertices $s_0$ and $s_1$ of $s$ and the vertices $s_0'$ and $s_1'$ of $s'$ as complex numbers. The displacement of $\Delta$ is then described by the map $z \to (z - s_0) \cdot (s_1' - s_0') / (s_1 - s_0) + s_0'$.}

\begin{figure}[ht]
    \centering
    \includegraphics[width=0.8\linewidth]{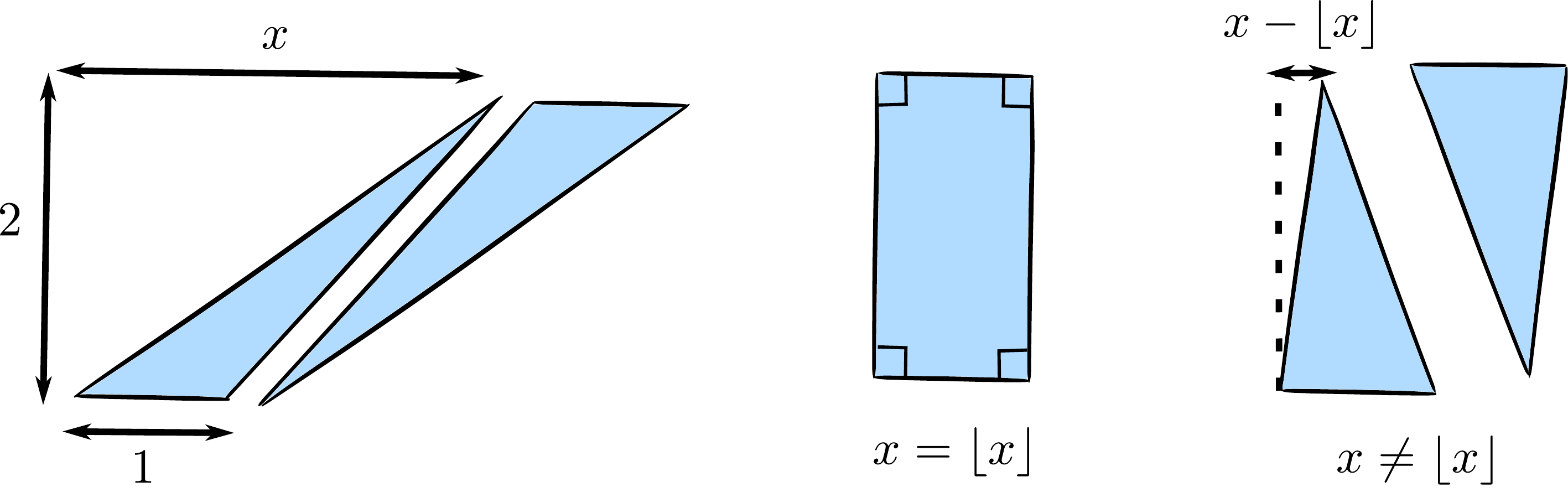}
    \caption{(Left) The polygons in the portalgon $T_x$. (Right) The polygons in the portalgon $D_x$, depending on whether $x = \lfloor x \rfloor$ or not.}
    \label{fig:lower bound}
\end{figure}

To obtain our lower bound, we consider, for every $x \in (1,\infty)$, a particular representation of a portalgon $T_x$. See Figure~\ref{fig:lower bound}. It has exactly two triangles. The first triangle has vertices $(0,0)$, $(1,0)$ and $(x, 2)$. The second triangle is isometric to the first triangle, its vertices are $(1,0)$, $(x,2)$, and $(x+1,2)$. In $T_x$ the sides of the first triangle are matched with the corresponding sides of the second triangle. The surface $\mathcal S(T_x)$ is a flat torus $S$, independent of $x$. We denote by $D_x$ the portalgon of the Delaunay tessellation of $S$. 

To obtain Theorem~\ref{thm:main lower bound}, we use the fact that, on a real RAM, the floor of a positive real number cannot be computed in (strongly) sub-logarithmic time (Lemma~\ref{lem:floor takes log} below), and we reduce the problem of computing the floor of a positive real number to the problem of computing a Delaunay tessellation, in order to transpose the lower bound. More formally, we use the following:

\begin{restatable}{lemma}{lemmafloortakeslog}\label{lem:floor takes log}
Let $c \in (0,1)$. There is no real RAM program computing $\lfloor x \rfloor$ from $x \in (1,\infty)$ in $O((\log x)^c)$ time.
\end{restatable}

The proof of Lemma~\ref{lem:floor takes log} is deferred to Appendix~\ref{app:floor takes log}. A result similar to Lemma~\ref{lem:floor takes log} was proved by Blum, Shub, and Smale~\cite[Section~4,~Proposition~3]{blum1989theory} on a different machine, excluding the square root operation. We adapt their arguments to the machine described by Erickson, van der Hoog, and Miltzow~\cite{erickson2022smoothing}, including the square root operation.

Theorem~\ref{thm:main lower bound} is almost immediate from Lemma~\ref{lem:floor takes log}:

\begin{proof}[Proof of Theorem~\ref{thm:main lower bound}]
The aspect ratio of $T_x$ is $O(x^2)$ because the triangles of $T_x$ have maximum side length $O(x)$ and minimum height $\Omega(1/(1+x))$. We now prove that there is no real RAM program computing a representation of $D_x$ from $x \in (1,\infty)$ in $O((\log x)^c)$ time for some $c \in (0,1)$. To do so we describe a real RAM program that, given $x$ and a representation of $D_x$, computes $\lfloor x \rfloor$ in constant time. This will prove the theorem by Lemma~\ref{lem:floor takes log}. There are two cases. See Figure~\ref{fig:lower bound}. Either $D_x$ has a unique polygon (in fact a rectangle), in which case $\lfloor x \rfloor = x$, so we return $x$. Otherwise $D$ has two triangles. In this case let $\Delta$ be any of these two triangles in our representation of $D_x$. Our triangle $\Delta$ has a unique side of length one. An orientation-preserving isometry of the plane displaces this unit length side to the segment between $(0,0)$ and $(1,0)$. The same isometry displaces the third vertex of $\Delta$ to a point $(u,v)$ such that $u = x - \lfloor x \rfloor$. We compute $u$ and return $x-u$. This can be done using only the basic arithmetic operations of the real RAM on the coordinates of the vertices of $\Delta$. The theorem is proved.
\end{proof}

\subsection{Proof of Lemma~\ref{lem:floor takes log}}\label{app:floor takes log}

In this section we prove Lemma~\ref{lem:floor takes log}, which we restate for convenience:

\lemmafloortakeslog*

As already mentioned, a result similar to Lemma~\ref{lem:floor takes log} was proved by Blum, Shub, and Smale~\cite[Section~4,~Proposition~3]{blum1989theory} on a different machine, excluding the square root operation. We adapt their arguments to the machine described by Erickson, van der Hoog, and Miltzow~\cite{erickson2022smoothing}, including the square root operation. This requires a few additional arguments.

As a preliminary, note that the real number operations available to the machine have domains of definition. Indeed, dividing by zero is undefined. Moreover, it is convenient for us to define the square root operation on $(0,\infty)$, thus excluding $0$. This way square root is real analytic on its domain, and more generally all real number operations available to the machine are real analytic on their domain, so their combinations too, a fact we will use in our proof. Having the square root of $0$ undefined does not change the computational power of the machine anyway.

Informally, the core of the proof consists in separating the complexity of a computation induced by the flow control instructions from the complexity induced by the operations on real numbers. So first, we consider a program in which there is no flow control instruction, in other words which instructions are executed does not depend on the input. We formalize that with a simper model of computation. A \emphdef{straight-line program} is a sequence of $n \geq 1$ instructions $I_1, \dots, I_n$ where for each $i \in [n]$ the instruction $I_i$ is either $x_i \leftarrow c$ for some constant $c  \in \mathbb R$, $x_i \leftarrow x_j \oplus x_k$ for some $j, k \in \{0, \dots, i-1\}$ and $\oplus \in \{+,-,\times,/\}$, or $x_i \leftarrow \sqrt{x_j}$. We compute in the natural way. The input is the initial value of the variable $x_0$. For each $i \in [n]$ the instruction $I_i$ computes the value of the variable $x_i$, and the output is the value of $x_n$ computed by the last instruction $I_n$. The computation fails if at some point we divide by zero, or if we take the square root of a non-positive number. This defines a function $f: D \subset \mathbb R \to \mathbb R$, where $D$ contains the input values for which the computation does not fail, and where $f$ maps each input value to the corresponding output value. We say that $f$ is the function \emphdef{associated to} to the straight-line program.

The key argument is to show that the function associated to a straight-line program is “nice”. Formally, given $y \in \mathbb R$, we say that $f$ \emphdef{flattens at} $y$ if there is an open set $O \subset \mathbb R$ such that $O \subseteq D$ and $f(O) = \{y\}$. Then:

\begin{lemma}\label{lem:flattens bound}
If $f : D \subset \mathbb R \to \mathbb R$ is associated to some straight-line program with $n$ instructions then $f$ flattens at less than $(7n)^{3n}$ values.
\end{lemma}

Note that without the square root operation, every function associated to a straight-line program would be rational, so it would flatten at at most one value. When including the square root operation however, such functions can flatten at several values: consider, for example, the function that maps each $x \in \mathbb R \setminus \{0\}$ to $\sqrt{x^2}/x$. This is why we need extra arguments compared to Blum, Shub, and Smale~\cite[Section~4,~Proposition~3]{blum1989theory}.

\begin{proof}[Proof of Lemma~\ref{lem:flattens bound}]
The function $f$ is real analytic. Therefore if $f$ is constant on some open set $O \subset \mathbb R$ included in $D$, then $f$ is constant on the entire connected component of $D$ that contains $O$. We will now prove that $D$ has less than $(7n)^{3n}$ connected components. This will prove the lemma. 

By assumption $f$ is associated to a straight-line program $I_1, \dots, I_n$. For each $i \in [n]$ we associate to the instruction $I_i$ some polynomial equation(s). For example each instruction of the form $x_i \leftarrow x_j \times x_k$ is associated to the single equation $x_i - x_j x_k = 0$. And each instruction of the form $x_i \leftarrow \sqrt{x_j}$ is associated to the three equations $x_i^2 - x_j = 0$, $x_i > 0$, and $x_j > 0$. The other instructions are handled the same way. This defines at most $3n$ polynomial equations, each of degree at most two. The corresponding semi-algebraic subset $X \subset \mathbb R^{n+1}$ contains the possible values of the variables $(x_0, \dots, x_n)$ during an execution of the straight-line program. By definition the first co-coordinate projection maps $X$ to $D$ in a one-to-one manner, and this projection is continuous, so the number of connected components of $D$ is smaller than or equal to the number of connected components of $X$. And the latter is smaller than $(7n)^{3n}$ by a result of Milnor~\cite[Theorem~3]{milnor1964betti}.
\end{proof}

\begin{proof}[Proof of Lemma~\ref{lem:floor takes log}]
We consider a program $P$ that satisfies each of the following for every $x \in (1,\infty)$. Initialize all word and real registers with $0$, except the first real register initialized with $x$, and run the machine with program $P$. Then $P$ is such that 1) the machine never tries to perform an undefined operation (dividing by zero, or taking the square root of a non-positive number), 2) at some point the machine halts, and 3) afterward the first real register contains the value $\lfloor x \rfloor$. Assume by contradiction that there are $c \in (0,1)$ and $n \in \mathbb N$ such that $n > 1$, and such that for every $x \geq n$, when executed with program $P$ on input $x$, the machine halts after at most $(\log x)^c$ instructions. We will derive a contradiction, which will prove the lemma.

Consider an integer $m \geq n+1$. By assumption, when executed with program $P$ on some input $x \in [n, m]$, the machine halts after at most $\alpha = (\log m)^c$ instructions. Which instructions are executed depends on $x$, and more precisely on the outcomes of the flow control instructions. Because at most $\alpha$ flow control instructions are executed, there are at most $2^{\alpha}$ possibilities for the outcomes of the flow control instructions. For each such possibility, the output (value of the first real register once the machine halts) is a function of the input (initial value of the first real register) that is associated to a straight-line program with at most $\alpha$ instructions. In particular, this function flattens at less than $(7\alpha)^{3\alpha}$ values by Lemma~\ref{lem:flattens bound}, we will use this fact. Let $F$ contain all these functions, over all the possibilities for the outcomes of the flow control instructions. Now consider an integer $k \in [n, m-1]$. There is a non-empty open subset $O_k$ of the open interval $(k,k+1)$ such that the outcomes of the flow control instructions are the same for all inputs in $O_k$. So there is $f_k \in F$, whose domain contains $O_k$, such that $f_k(O_k) = \{k\}$. In other words, $f_k$ flattens at $k$. Now recall that each $f_k$ flattens at less than $(7\alpha)^{3\alpha}$ values, so $(m-n) \leq 2^{\alpha} \cdot (7\alpha)^{3\alpha}$. As $m$ goes to $\infty$ this inequality becomes false, because $\alpha = (\log m)^c$ and $c < 1$, a contradiction.
\end{proof}


\section{Appendix: Voronoi diagrams and Delaunay tessellations}\label{app:delaunay and voronoi defs}

In this section we provide elementary definitions and properties of Delaunay tessellations and Voronoi diagrams, for use in Appendix~\ref{A:Delaunay}. This is all folklore, but we could not find all the exact statements in the literature, so we provide proofs for completeness.

\subsection{Delaunay tessellations}\label{sec:preliminaries delaunay}

Consider a closed polyhedral surface $S$, and a finite non-empty $V \subset S$ containing all singularities of $S$. We consider the definition of Bobenko and Springborn~\cite[Section~2]{bobenko2007discrete} of the Delaunay tessellation of $(S,V)$. First, they define an \emphdef{immersed empty disk} as a pair $(D, \varphi)$ where $D$ is an open metric disk of $\mathbb R^2$, and $\varphi : \overline D \to S$ is a map defined on the closure $\overline D$ of $D$ that satisfies the following: the restriction of $\varphi$ to $D$ is an isometric immersion, and $\varphi(D) \cap V = \emptyset$. Note that $\varphi$ is not necessarily injective. Bobenko and Springborn~\cite[Proposition~4]{bobenko2007discrete} proved:

\begin{lemma}[Bobenko and Springborn]\label{L:exists Delaunay}
There is a unique tessellation $\mathcal D$ of $S$ such that for every immersed empty disk $(D, \varphi)$, if $\varphi^{-1}(V)$ is not empty, then the convex hull of $\varphi^{-1}(V)$ projects via $\varphi$ to either a vertex, an edge, or the closure of a face of $\mathcal{D}$, and such that every vertex, edge, and face of $\mathcal{D}$ can be obtained this way.
\end{lemma}

The tessellation $\mathcal D$ given by Lemma~\ref{L:exists Delaunay} is the \emphdef{Delaunay tessellation} of $(S, V)$. It is “in general” a triangulation, but not always. 

We will also use the following definition. For every point $x \in S$ there is an immersed empty disk $(D, \varphi)$ such that $\varphi$ maps the center of $D$ to $x$, and such that $\varphi^{-1}(V) \neq \emptyset$. And $(D, \varphi)$ is unique to $x$ in the sense that if $(D', \varphi')$ is another such immersed empty disk then there is a plane isometry $\psi : \mathbb{R}^2 \to \mathbb{R}^2$ satisfying $D' = \psi(D)$ and $\varphi = \varphi' \circ \psi$. We say that $(D, \varphi)$ is the \emphdef{maxi-disk} of the point $x$.

\subsection{Voronoi diagrams}

Again, consider a closed polyhedral surface $S$, and a finite non-empty $V \subset S$ containing all singularities of $S$. The \emphdef{Voronoi diagram} of $(S, V)$ contains the points $x \in S$ such that the distance between $x$ and $V$ is realized by at least two distinct paths in $S$. Note that it is possible for the Voronoi diagram of $(S,V)$ to contain a point $x$ such that all the shortest paths between $x$ and $V$ end at the same point of $V$. This is for example the case if $S$ is a flat torus and $V$ contains exactly one point of $S$. In this section we prove the following:

\begin{lemma}\label{lem:allprops voronoi}
Let $S$ be a closed polyhedral surface. Let $V \subset S$ be finite, non-empty, and containing all singularities of $S$. The Voronoi diagram of $(S, V)$ is a graph with finitely many vertices in which every vertex has degree greater than or equal to three, every edge is geodesic, every face is homeomorphic to an open disk and contains exactly one point of $V$, and every angle at a corner of a face is smaller than or equal to $\pi$.
\end{lemma}

Note that without the assumption that $V$ contains all the singularities of $S$, it would be possible for the Voronoi diagram of $(S,V)$ to not be a graph with geodesic edges.

\begin{proof}[Proof of Lemma~\ref{lem:allprops voronoi}]
Consider the Voronoi diagram $\mathcal V$ of $(S, V)$. We have three claims that immediately imply the lemma. Our first claim is that $\mathcal V$ is a graph with finitely many vertices, in which every vertex has degree greater than or equal to three, and in which every edge is geodesic. To prove the first claim consider a point $x \in S$, and the maxi-disk $(D, \varphi)$ of $x$. Let $x^\star$ be the center of $D$. The geodesic paths between $x^\star$ and $\varphi^{-1}(V)$ in $\mathbb R^2$ correspond via $\varphi$ to the shortest paths between $x$ and $V$ in $S$. So $x$ belongs to $\mathcal V$ if and only if $\varphi^{-1}(V)$ contains several points. Assume that $x$ belongs to $\mathcal V$, and let $m \geq 2$ be the number of points in $\varphi^{-1}(V)$. Consider, in $\mathbb R^2$, the Voronoi diagram of $\varphi^{-1}(V)$, which we denote by $X$. Then $X$ consists in $m$ geodesic rays emanating from $x^ \star$. There is an open ball $O \subset D$ on which $\varphi$ is injective, containing $x^\star$, and such that $\varphi(X \cap O) = \mathcal V \cap \varphi(O)$. There are two cases. If $m = 2$ then $\mathcal V$ is locally a geodesic path around $x$. If $m \geq 3$ then $\mathcal V$ is locally a geodesic star whose central vertex is $x$. In particular $\mathcal V$ is a graph whose minimum degree is greater than or equal to three, and whose edges are geodesic segments. And $\mathcal V$ has finitely many vertices because $S$ is compact. That proves the first claim.

Now consider a face $F$ of $\mathcal V$. Our second claim is that $F$ is simply connected, and that $F$ contains exactly one point of $V$. This implies that $F$ is homeomorphic to an open disk because $F$ is not homeomorphic to a sphere. To prove the second claim first consider a point $x \in F$. There is a unique shortest path $p$ from $x$ to $V$. Then $p$ is disjoint from $\mathcal V$. So the endpoint of $p$ belongs to $F$. That proves $F \cap V \neq \emptyset$. Now consider the universal covering space $\widetilde F$ of $F$. Then $\widetilde F$ does not contain two distinct lifts of points of $V$. For otherwise let $\widetilde V$ contain the lifts of the points of $V$ in $\widetilde F$. There is a point $\widetilde x \in \widetilde F$ whose distance to $\widetilde V$ is realized by two distinct paths. And $\widetilde x$ lifts a point of $\mathcal V$, a contradiction. That proves the second claim.

Finally, consider a vertex $v$ of $\mathcal V$. Our third claim is that around $v$ the angles between consecutive edges of $\mathcal V$ are all smaller than or equal to $\pi$. To prove the claim consider the maxi-disk $(D, \varphi)$ of $v$. Let $v^\star$ be the center of $D$. Let $X$ be the Voronoi diagram of $\varphi^{-1}(V)$ in the plane. The faces of $X$ are all convex, being intersections of half-planes. So the angles between consecutive rays of $X$ around $v^\star$ are all smaller than or equal to $\pi$. There is an open disk $O$ on which $\varphi$ is injective, containing $v^\star$, such that $\varphi(X \cap O) = \mathcal V \cap \varphi(O)$. That proves the third claim, and the lemma.
\end{proof}

\subsection{Voronoi diagram and Delaunay tessellation}

A graph $G$ is \emphdef{cellularly} embedded on a surface $S$ if the faces of the embedding are all homeomorphic to open disks. In this section, given two graphs $G$ and $H$ cellularly embedded on $S$, we say that $G$ and $H$ are \emphdef{isomorphic} if there is an orientation-preserving homeomorphism of $S$ that maps $G$ to $H$, for some orientation of $S$. This does not depend on the orientation of $S$. We prove the following:

\begin{lemma}\label{lem:duality delaunay voronoi}
Let $S$ be a closed polyhedral surface. Let $V \subset S$ be finite, non-empty, and containing all singularities of $S$. The Voronoi diagram of $(S,V)$ is isomorphic to the dual of the Delaunay tessellation of $(S,V)$.
\end{lemma}

(Recall that in Lemma~\ref{lem:duality delaunay voronoi} the Voronoi diagram of $(S,V)$ is a graph cellularly embedded on $S$ by Lemma~\ref{lem:allprops voronoi}.)

\begin{proof}[Proof of Lemma~\ref{lem:duality delaunay voronoi}]
Consider the Voronoi diagram $\mathcal V$ of $(S, V)$, and the Delaunay tessellation $\mathcal D$ of $(S,V)$. Consider a point $x$ of $S$, and its maxi-disk $(D, \varphi)$. We already proved that $x$ is a vertex of $\mathcal V$ is and only if $\varphi^{-1}(V)$ contains at least three points. This is the case if and only if the convex hull of $\varphi^{-1}(V)$ projects via $\varphi$ to the closure of a face $f$ of $\mathcal D$ (Lemma~\ref{L:exists Delaunay}). Every face of $\mathcal D$ can be obtained this way (Lemma~\ref{L:exists Delaunay}), and distinct vertices of $\mathcal V$ are clearly mapped to distinct faces of $\mathcal D$. So this defines a one-to-one correspondence between the vertices of $\mathcal V$ and the faces of $\mathcal D$. When a vertex $v$ of $\mathcal V$ corresponds to a face $f$ of $\mathcal D$ this way we say that $v$ is \emphdef{dual} to $f$.

Now fix a vertex $v$ of $\mathcal V$, and its dual face $f$ of $\mathcal D$. We call \emphdef{side} of $f$ any directed edge of $\mathcal D$ that sees $f$ on its left. We now relate the directed edges based at $v$ in $\mathcal V$ to the sides of $f$. Again, let $(D, \varphi)$ be the maxi-disk of $v$. Let $v^\star$ be the center of $D$, and let $y_0, \dots, y_{m-1}$ be the $m \geq 3$ points of $\varphi^{-1}(V)$. In $\mathbb R^2$ the Voronoi diagram of $\varphi^{-1}(V)$ consists in $m$ geodesic rays $r_0, \dots, r_{m-1}$ emanating from $v^\star$, so that $r_0, y_0, \dots, r_{m-1}, y_{m-1}$ are in clockwise order around $v^\star$. There is an open ball $O \subset D$ on which $\varphi$ is injective, containing $v^\star$, such that within $O$ the rays $r_0, \dots, r_{m-1}$ correspond via $\varphi$ to the directed edges $e_0, \dots, e_{m-1}$ emanating from $v$ in $\mathcal V$. For every $i$ the geodesic path from $y_i$ to $y_{i+1}$ corresponds via $\varphi$ to a side $e_i'$ of $f$, indices are modulo $m$. We say that $e_i$ and $e_i'$ are \emphdef{dual}. This duality is a one-to-one correspondence between the directed edges based at $v$ and the sides of $f$. The former are cyclically ordered around $v$, the latter are cyclically ordered along the boundary of $f$, and the duality correspondence respects these cyclic orders.

We claim that if a directed edge $e_0$ of $\mathcal V$ is dual to a directed edge $e_0'$ of $\mathcal D$, then the reversal of $e_0$ is dual to the reversal of $e_0'$. The claim immediately implies the lemma, for then the duality correspondences define the desired graph isomorphism between $V$ and $\mathcal D$. Let us prove the claim. Let $e_1'$ be the reversal of $e_0'$, and let $e_1$ be the dual of $e_1'$. We shall prove that $e_1$ is the reversal of $e_0$. We consider the maxi-disks $(D_0, \varphi_0)$ and $(D_1, \varphi_1)$ of the base vertices of $e_0$ and $e_1$, and we realize them so that they agree on the geodesic segment $p$ that is the pre-image of the common edge of $e_0'$ and $e_1'$. Then $\varphi_0$ and $\varphi_1$ agree on $\overline D_0 \cap \overline D_1$, so they agree with a common map $\varphi_0 \cup \varphi_1 : \overline D_0 \cup \overline D_1 \to S$. Let $q$ be the geodesic segment between the centers of $D_0$ and $D_1$. Then $q$ is contained in $\overline D_0 \cup \overline D_1$, and projects via $\varphi_0 \cup \varphi_1$ to the common edge of $e_0$ and $e_1$ in $\mathcal V$. Indeed for every point $x^\star$ in the relative interior of $q$ the maxi-disk $(D, \varphi)$ of $\varphi(x^\star)$ can be realized so that $x^\star$ is the center of $D$, and so that $\varphi$ agrees with $\varphi_0 \cup \varphi_1$ on $\overline D \cap (\overline D_0 \cup \overline D_1)$. Then $\varphi^{-1}(V)$ contains exactly the two endpoints of $p$, and so $\varphi(x^\star)$ belongs to the relative interior of an edge of $\mathcal V$. This proves the claim, and the lemma.
\end{proof}


\section{Appendix: proof of Proposition~\ref{thm:isometry}}\label{A:Delaunay}

See Appendix~\ref{app:delaunay and voronoi defs} for basic definitions and properties of Delaunay tessellations and Voronoi diagrams. In this section we prove Proposition~\ref{thm:isometry}, which we restate for convenience:

\thmisometry*

Proposition~\ref{thm:isometry} slightly extends known results, and is not surprising at all, but we provide proofs for completeness. Our strategy for computing the Delaunay tessellation is, classically, to first compute the Voronoi diagram:

\begin{restatable}{proposition}{propcomputevoronoi}\label{prop:compute voronoi}
Let $T$ be a portalgon of $n$ triangles, of happiness $h$. Let $V$ be a set of vertices of $T^1$. Assume that $V$ is not empty and contains all singularities of $\mathcal S(T)$. We can compute in $O(n^2 h \log(nh))$ time a portalgon $T'$ of $O(n^2h)$ triangles, whose surface is $\mathcal S(T)$, and a subgraph $\mathcal V$ of $T'^1$, such that $\mathcal V$ is the Voronoi diagram of $(\mathcal S(T), V)$.
\end{restatable}

We will then derive the Delaunay tessellation of from the Voronoi diagram:

\begin{restatable}{proposition}{computedelaunayfromvoronoi}\label{prop:compute delaunay from voronoi}
Let $T$ be a portalgon of $n$ triangles. Let $V$ be a set of vertices of $T^1$. Let $\mathcal V$ be a subgraph of $T^1$. Assume that $V$ is not empty and contains all singularities of $\mathcal S(T)$, and that $\mathcal V$ is the Voronoi diagram of $(\mathcal S(T),V)$. We can compute the portalgon of the Delaunay tessellation of $(\mathcal S(T), V)$ in $O(n)$ time.
\end{restatable}

Proposition~\ref{prop:compute voronoi} and Proposition~\ref{prop:compute delaunay from voronoi} will immediately imply Proposition~\ref{thm:isometry}:

\begin{proof}[Proof of Proposition~\ref{thm:isometry}, assuming Propositions~\ref{prop:compute voronoi}~and~\ref{prop:compute delaunay from voronoi}]
Let $V$ contain the singularities of $\mathcal S(T)$, except if $\mathcal S(T)$ is a flat torus in which case let $V$ contain a single arbitrary vertex of $T^1$. Apply Proposition~\ref{prop:compute voronoi} to replace $T$ by a portalgon $T'$ of $O(n^2h)$ triangles, and to compute a subgraph $\mathcal V$ of $\mathcal T'^1$ that is also the Voronoi diagram of $(\mathcal S(T), V)$, all in $O(n^2h\log(nh))$ time. Apply Proposition~\ref{prop:compute delaunay from voronoi} to derive from $T'$ and $\mathcal V$ the portalgon of the Delaunay tessellation of $(\mathcal S(T), V)$ in $O(n^2h)$ time.
\end{proof}

All there remains to do is to prove Proposition~\ref{prop:compute voronoi} and Proposition~\ref{prop:compute delaunay from voronoi}. We prove Proposition~\ref{prop:compute delaunay from voronoi} in Section~\ref{sec:computedelfromvor}, and we prove Proposition~\ref{prop:compute voronoi} in Section~\ref{sec:computevor}.

\subsection{Computing the Delaunay tessellation from the Voronoi diagram}\label{sec:computedelfromvor}

In this section we prove Proposition~\ref{prop:compute delaunay from voronoi}, which we restate for convenience:

\computedelaunayfromvoronoi*

In the setting of Proposition~\ref{prop:compute delaunay from voronoi}, our goal is to compute the portalgon of the Delaunay tessellation $\mathcal D$ of $(\mathcal S(T),V)$. If we do not care about the shapes of the polygons in the portalgon, then we can easily compute this portalgon from the embedded graph $\mathcal V$, due to the duality between $\mathcal D$ and $\mathcal V$ (Lemma~\ref{lem:duality delaunay voronoi}). All there remains to do is to compute the shape of each polygon in the portalgon. First we need a definition and a lemma. Consider a walk $W$ in the dual of $T^1$. To ease the reading assume that every edge of $T^1$ is incident to two distinct faces of $T^1$; the following definition extends in a straightforward manner to general triangulations. In the plane $\mathbb R^2$ realize the $k \geq 1$ faces visited by $W$ isometrically, and respecting their orientation, by respective triangles $\Delta_1, \dots, \Delta_k$. Make sure that for every $1 \leq i < k$ the triangles $\Delta_i$ and $\Delta_{i+1}$ agree on the placement of the $i$-th edge of $T^1$ crossed by $W$. The resulting sequence $\Delta = (\Delta_1, \dots, \Delta_k)$ is an \emphdef{unfolding} of $W$. In general a vertex $w$ of $T^1$ may occur several times among the vertices of the triangles in $\Delta$, and those occurrences may be at distinct points in the plane. Nevertheless:

\begin{lemma}\label{L:Voronoi unfolding}
If the faces of $T^1$ visited by $W$ are all included in the same face of $\mathcal V$, and if $w \in V$, then all the occurrences of $w$ in $\Delta$ are at the same point of $\mathbb R^2$.
\end{lemma}

\begin{proof}
Let $F$ be the face of $\mathcal V$ containing the faces of $T^1$ visited by $W$. By Lemma~\ref{lem:allprops voronoi} the face $F$ is homeomorphic to an open disk, and $w$ is the unique point of $V \cap F$. Let $\widehat F$ be the surface homeomorphic to a closed disk obtained by cutting the closure of $F$ along the boundary of $F$. The angles at the corners of $\widehat F$ are smaller than or equal to $\pi$ by Lemma~\ref{lem:allprops voronoi}. So the shortest paths between those corners and $w$ are, together with the boundary edges of $\widehat F$, the edges of a triangulation $Y$ of $\widehat F$. The dual of $Y$ is a cycle, and $w$ is the central vertex of $Y$. If $\Delta$ is an unfolding of a walk in the dual of $Y$, then all occurrences of $w$ in $\Delta$ are at the same point in the plane. This easily extends to every other triangulation $Y'$ of $\widehat F$, by considering a triangulation of $\widehat F$ that contains both $Y$ and $Y'$.
\end{proof}

In the portalgon of $\mathcal D$, consider a polygon $P$. We describe how to compute the positions of the vertices of $P$. Note that these positions are only defined up to translating and rotating $P$ in the plane. As a preliminary, consider the vertex $v$ of the Voronoi diagram $\mathcal V$ that is dual to $P$. Embed a neighborhood of $v$ in the plane $\mathbb R^2$, isometrically and respecting the orientation, by embedding the faces of $T^1$ incident to $v$. This is possible because $v$ is flat. 

Assign to each vertex $x$ of $P$ a point $\pi_P(x) \in \mathbb R^2$ as follows. The vertex $x$ of $P$ is dual to an incidence $c$ between the vertex $v$ of $\mathcal V$ and some face $F$ of $\mathcal V$. In this incidence $c$, consider one of the faces $W_0$ of $T^1$ that we embedded in the plane, and its embedding $W_0^*$. Consider the unique point $w \in V \cap F$ (Lemma~\ref{lem:allprops voronoi}). Consider a walk $W$ in the dual of $T^1$ that starts with $W_0$, remains in $F$, and visits at least one face of $T^1$ incident to $w$. Unfold the faces visited by $W$ in the plane, starting from $W_0^*$. Let $\pi_P(x)$ be any occurrence of $w$ in the unfolding.

\begin{lemma}\label{lem:delaunay assignment}
Up to translating and rotating $P$, the assignment $\pi_P$ maps each vertex of $P$ to its position.
\end{lemma}

\begin{proof}[Proof of Lemma~\ref{lem:delaunay assignment}]
Consider the maxi-disk $(D, \varphi)$ of $v$. Without loss of generality $\varphi$ agrees with the embedding of the neighborhood of $v$ that we fixed as a preliminary. Recall from the definition of the Delaunay tessellation that $P$ is the convex hull of $\varphi^{-1}(V)$. Let $v^\star$ be the middle point of $D$ (the embedding of $v$). The points in $\varphi^{-1}(V)$ correspond to the incidences between $v$ and the faces of $\mathcal V$ around $v$. Consider such an incidence $c$, and its corresponding point $y \in \varphi^{-1}(V)$. Consider also the face $F$ of $\mathcal V$ that contains $c$. The geodesic path $p$ from $v^\star$ to $y$ projects via $\varphi$ to a shortest path $\varphi \circ p$ from $v$ to $V$. And $\varphi \circ p$ immediately enters $F$ after leaving $v$. So the relative interior of $\varphi \circ p$ is included in $F$, and thus ends at the unique point $w \in V \cap F$. By slightly perturbing $p$ without changing its endpoints we may ensure that $\varphi \circ p$ corresponds to a walk in the dual of $T^1$. Then $y = \pi_P(x)$ by Lemma~\ref{L:Voronoi unfolding}. This proves the lemma.
\end{proof}

\begin{proof}[Proof of Proposition~\ref{prop:compute delaunay from voronoi}]
We must compute the portalgon of the Delaunay tessellation $\mathcal D$ of $(S, V)$. We immediately compute the combinatorics of the portalgon from $\mathcal V$, because $\mathcal V$ is isomorphic to the dual of $\mathcal D$ by Lemma~\ref{lem:duality delaunay voronoi}.

Now, by Lemma~\ref{lem:delaunay assignment}, all there remains to do is to compute, for each polygon $P$ of the portalgon, the assignment $\pi_P$ of positions for the vertices of $P$. Achieving the claimed linear running time when doing so requires a slight technicality. Consider a face $F$ of $\mathcal V$, and the point $w \in V \cap F$. Recall that for some faces $W_0$ of $T^1$ included in $F$ we need to construct a walk $W$ from $W_0$ to $w$ in the dual of $T^1$, unfold $W$, and retain the relative positions of some occurrences of $W_0$ and $w$ in the unfolding. Doing so independently for every face $W_0$ may take too long as we would visit faces of $T^1$ several times. Instead we consider a single spanning tree $Y$ in the dual of $T^1$ within $F$, we unfold the faces of $T^1$ that are included in $F$ along $Y$, and we retrieve all the required information from this unfolding. Note that the choice of $Y$ does not matter, and that the unfolding may overlap. Doing so for all faces $F$ of $\mathcal V$ takes linear time.
\end{proof}

\subsection{Computing the Voronoi diagram}\label{sec:computevor}

In this section we prove Proposition~\ref{prop:compute voronoi}, which we restate for convenience:

\propcomputevoronoi*

To prove Proposition~\ref{prop:compute voronoi} we revisit the single source shortest path algorithm described by Löffler, Ophelders, Silveira, and Staals~\cite{portalgons}. In particular we extend their algorithm to multiple sources (we let the sources be the points in $V$). The authors consider a triangulated surface, and compute the shortest paths emanating from a point $x_0$ on the surface by decomposing the surface according to how those paths visit the faces of the triangulation. They describe a discrete process that simulates the propagation of some waves on the surface. Their waves all start from the point $x_0$. In the setting of Proposition~\ref{prop:compute voronoi}, we adapt this strategy to simulate waves on $\mathcal S(T)$ that start from all the points in $V$, so that the waves meet along the Voronoi diagram $\mathcal V$ of $(\mathcal S(T), V)$. That simplifies the algorithm because waves now meet along a graph with geodesic edges (Lemma~\ref{lem:allprops voronoi}) and do not go through singularities. We now provide a formalization of this \emph{wave algorithm}. The continuous propagation of waves is discretized by a propagation of \emph{events}. A crucial feature of our formalization of the algorithm is that it operates on triangles, point sets, and Voronoi diagrams \emph{in the plane} $\mathbb R^2$, never in the surface $\mathcal S(T)$. Only the proofs of correctness will argue on the surface $\mathcal S(T)$. We will insist on that.

Recall that the triangles of the input portalgon $T$ lie in the Euclidean plane $\mathbb R^2$, and they are disjoint. The reader can think of them as being far away from each other if this helps the reading. The data structure maintains, for every triangle $\Delta$ of $T$, a set $X_\Delta$ of points in $\mathbb R^2$. We insist, again, that all these objects lie \emph{in the plane} $\mathbb R^2$, not in the surface $\mathcal S(T)$.

We need a definition. Given $X \subset \mathbb R^2$ finite and $x \in X$ we denote by $\text{Vor}(x, X)$ the closed cell of $x$ in the Voronoi diagram of $X$ in $\mathbb R^2$. Formally, $\text{Vor}(x, X)$ contains the points $y \in \mathbb R^2$ such that the distance between $x$ and $y$ is smaller than or equal to the distance between $x'$ and $y$ for every $x' \in X$.

Central to the wave algorithm is the notion of candidate event that we now define. Consider a triangle $\Delta$ of $T$, a side $s$ of $\Delta$, a point $x \in \mathbb R^2$, and some $t > 0$. The surface $\mathcal S(T)$ being closed, there are a triangle $\Delta'$ of $T$ and a side $s'$ of $\Delta'$ such that $s$ is matched to $s'$. Consider the orientation-preserving isometry of $\mathbb R^2$ that maps $s$ to $s'$ and puts $\Delta$ side-by-side with $\Delta'$, apply this isometry to $x$, and consider the resulting point $x' \in \mathbb R^2$. The tuple $(t, \Delta, s, x)$ is a \emphdef{candidate event} if it satisfies each of the following. First, $x \notin X_\Delta$ and $x' \in X_{\Delta'}$. Second, the intersection between $\text{Vor}(x', X_{\Delta'})$ and $s'$ is not empty, and its distance to $x'$ is equal to $t$. In other words, $t$ is equal to the smallest distance between $x'$ and a point of $\text{Vor}(x', X_{\Delta'}) \cap s'$. We say that $t$ is the \emphdef{date} of the candidate event $(t, \Delta, s, x)$. 

The data structure additionally maintains a list of the candidate events sorted by date.

\algo{Wave algorithm}{Initialize for each triangle $\Delta$ of $T$ the set $X_\Delta$ with the vertices of $\Delta$ that correspond to points in $V$, if any. Then, as long as possible, consider any candidate event $(t, \Delta, s, x)$ of smallest date $t$, add $x$ to $X_\Delta$, and repeat. In the end return the sets $(X_\Delta)_\Delta$.}

Again, we insist that the wave algorithm operates \emph{in the plane} $\mathbb R^2$. In particular the sets $X_\Delta$ are subsets of $\mathbb R^2$. Nevertheless, their Voronoi diagrams are related to the Voronoi diagram of $V$ on the surface $\mathcal S(T)$, and more strongly:

\begin{proposition}\label{prop:wave algorithm terminates}
The wave algorithm terminates after $O(n^2h)$ iterations. In the end, for every triangle $\Delta$ of $T$, the intersection with $\Delta$ of the Voronoi diagram of $X_\Delta$ in $\mathbb R^2$ is the pre-image in $\Delta$ of the Voronoi diagram of $V$ in $\mathcal S(T)$.
\end{proposition}

It is easy to compute the list of the candidate events from the sets $(X_\Delta)_\Delta$ in polynomial time. More strongly:

\begin{proposition}\label{prop:maintain candidate events}
We can maintain the list of candidate events sorted by date through $k$ insertions of points in the sets $(X_\Delta)_\Delta$ in $O(k \log k)$ total time.
\end{proposition}

Proposition~\ref{prop:compute voronoi} is immediate from Proposition~\ref{prop:wave algorithm terminates} and Proposition~\ref{prop:maintain candidate events}:

\begin{proof}[Proof of Proposition~\ref{prop:compute voronoi}]
Proposition~\ref{prop:wave algorithm terminates} and Proposition~\ref{prop:maintain candidate events} imply 
that the wave algorithm can be performed in $O(n^2 \cdot h \cdot \log(nh))$ time. Consider the returned sets $(X_\Delta)_\Delta$. The sum of the cardinalities of the sets $X_\Delta$, summed over all the triangles $\Delta$ of $T$, is $O(n^2h)$ by Proposition~\ref{prop:wave algorithm terminates}. Also, for every triangle $\Delta$, the intersection with $\Delta$ of the Voronoi diagram of $X_\Delta$ in $\mathbb R^2$ is the pre-image in $\Delta$ of the Voronoi diagram of $V$ in $\mathcal S(T)$, by Proposition~\ref{prop:wave algorithm terminates}. Cutting each triangle $\Delta$ along the Voronoi diagram of $X_\Delta$, and cutting the resulting polygons into triangles along vertex-to-vertex arcs, provides the desired triangular portalgon $T'$, along with $\mathcal V$.
\end{proof}

The rest of this section is dedicated to the proofs of Proposition~\ref{prop:wave algorithm terminates} and Proposition~\ref{prop:maintain candidate events}. 

\subsubsection{Proof of Proposition~\ref{prop:wave algorithm terminates}}

In this section we prove Proposition~\ref{prop:wave algorithm terminates}. Recall that the triangles of the portalgon $T$ are realized dis-jointly in the Euclidean plane $\mathbb R^2$, and that we think of these triangles as being far away from each other, this will help the reading. It is now convenient to introduce a notation for the projection of this disjoint union of triangles onto the surface $\mathcal S(T)$, so we let $\rho$ be this projection. 

Given a triangle $\Delta$ of $T$, we consider the immersed empty disks $(D, \varphi)$ such that the center of $D$ belongs to $\Delta$, and such that $\varphi$ agrees with $\rho$ on $\overline D \cap \Delta$. We say that $(D, \varphi)$ is an immersed empty disk \emphdef{attached to} $\Delta$. We further consider the union of the sets $\varphi^{-1}(V)$ over the immersed empty disks $(D, \varphi)$ attached to $\Delta$. We call this union the \emphdef{constellation} of $\Delta$, and we denote it by $V_\Delta$. We will show that the sets $X_\Delta$ computed by the wave algorithm are exactly the constellations $V_\Delta$. Before that, we have two preliminary lemmas on constellations. 

First, the constellations, while lying \emph{in the plane} $\mathbb R^2$, are related to the Voronoi diagram of $V$ \emph{in the surface} $\mathcal S(T)$:

\begin{lemma}\label{lem:vp is what we want}
For every triangle $\Delta$ of $T$ the intersection with $\Delta$ of the Voronoi diagram of the constellation $V_\Delta$ is the pre-image in $\Delta$ of the Voronoi diagram of $V$ in $\mathcal S(T)$.
\end{lemma}

The proof of Lemma~\ref{lem:vp is what we want} relies on the following, which will be used again:

\begin{lemma}\label{lem: coherence disk}
Let $(D, \varphi)$ be an immersed empty disk attached to a triangle $\Delta$ of $T$. Then $V_\Delta \cap \overline D = \varphi^{-1}(V)$. In particular $V_\Delta \cap D = \emptyset$.
\end{lemma}

\begin{proof}
We have $V_\Delta \cap \overline D \supseteq \varphi^{-1}(V)$ by definition of the constellation $V_\Delta$. The other inclusion is immediate from the fact that if two immersed empty disks $(D, \varphi)$ and $(D', \varphi')$ are attached to $\Delta$ then $\varphi$ and $\varphi'$ agree on $\overline D \cap \overline D'$. Finally $\varphi^{-1}(V) \cap D = \emptyset$ by definition of an immersed empty disk (recall that $D$ is open).
\end{proof}

\begin{proof}[Proof of Lemma~\ref{lem:vp is what we want}]
Consider a point $x \in \Delta$. There is a unique immersed empty disk $(D, \varphi)$ attached to $\Delta$ such that the center or $D$ is $x$, and such that the radius of $D$ is maximum. Then $\varphi^{-1}(V) \neq \emptyset$ because the radius of $D$ is maximum. And $\varphi^{-1}(V) = V_\Delta \cap \overline D$ by Lemma~\ref{lem: coherence disk}. The geodesic path(s) between $x$ and the point(s) in $\varphi^{-1}(V)$ corresponds via $\varphi$ to the shortest path(s) between $\rho(x)$ and $V$. So $\rho(x)$ belongs to the Voronoi diagram of $V$ in $\mathcal S(T)$ if and only if $\varphi^{-1}(V)$ contains several points, equivalently $V_\Delta \cap \overline D$, which is the case if and only if $x$ belongs to the Voronoi diagram of $V_\Delta$ in $\mathbb R^2$.
\end{proof}

Second, the cardinalities of the constellations are bounded by the number $n$ of triangles and the happiness $h$ of the portalgon $T$:

\begin{lemma}\label{L:VF bound}
For every triangle $\Delta$ of $T$ the constellation $V_\Delta$ has cardinality $O(nh)$.
\end{lemma}

\begin{proof}
Given a triangle $\Delta$ of $T$, and a point $x$ in the constellation $V_\Delta$, there is by definition an immersed empty disk $(D, \varphi)$ attached to $\Delta$ such that $x \in \varphi^{-1}(V)$. And the center $y$ of $D$ belongs to $\Delta$. Then the geodesic segment between $y$ and $x$ projects via $\rho$ to a path between $\rho(y)$ and $\rho(x)$, and the length of this path is the smallest possible among all the paths between $\rho(y)$ and a point of $V$ (possibly another point of $V$ than $\rho(x)$). We will argue on such shortest paths between a point of $\mathcal S(T)$ and the set $V$.

We call regions the following subsets of $\mathcal S(T)$: a vertex of $T^1$, the relative interior of an edge of $T^1$, and a face of $T^1$. The regions partition $\mathcal S(T)$. For every shortest path $p$ between a point $x \in \mathcal S(T)$ and the set $V$, record the sequence of regions intersected by $p$ when directed from $V$ to $x$. If two such paths $p$ and $p'$ end in $\rho(\Delta)$ and have the same sequence then they correspond to the same point in the constellation $V_\Delta$. We claim that for every region $R$ there are $O(nh)$ sequences ending with $R$. This claim implies the lemma. Let us prove the claim. We say that a sequence is maximal if it is not a strict prefix of another sequence. And we say that a sequence is critical if it is the maximal common prefix of two distinct maximal sequences. Every critical sequence ends with a face of $T^1$. For every face $R'$ of $T^1$ there is at most one critical sequence ending with $R'$. Indeed every critical sequence is realized by two distinct paths. If two distinct critical sequences were to end with $R'$, then at least two of the four associated paths would cross, and thus could be shortened, a contradiction. We proved that there are $O(n)$ critical sequences. So there are $O(n)$ maximal sequences. And every sequence contains $O(h)$ occurrences of $R$ because the happiness of $T$ is equal to $h$. This proves the claim, and the lemma.
\end{proof}

We will now show that the wave algorithm computes the constellations. To do so, we introduce an invariant. We need a definition. Fix a point $x \in V_\Delta$, and consider all the immersed empty disks $(D, \varphi)$ attached to $\Delta$ such that $x \in \varphi^{-1}(V)$. Among all these immersed empty disks $(D, \varphi)$, the smallest radius of $D$ is the \emphdef{depth} of $x$ in $V_\Delta$.

\algo{invariant}{
There is $\tau > 0$ such that both of the following hold for every triangle $\Delta$ of $T$. Every point of $X_\Delta$ belongs to the constellation $V_\Delta$. And every point of $V_\Delta \setminus X_\Delta$ has depth greater than or equal to $\tau$ in $V_\Delta$.
}

It is not clear a priori that the invariant is maintained by the wave algorithm, and this will be proved only at the end, when proving Proposition~\ref{prop:wave algorithm terminates}. Before that we need some lemmas.

\begin{lemma}\label{lem:min date event}
Assume that the invariant holds for some $\tau > 0$, and that there is a candidate event $(t, \Delta, s, x)$ such that $t \leq \tau$. Then $t = \tau$, $x$ belongs to $V_\Delta$, and the depth of $x$ in $V_\Delta$ is equal to $\tau$.
\end{lemma}

\begin{proof}
We claim that $x \in V_\Delta$, and that the depth of $x$ in $V_\Delta$ is smaller than or equal to $t$. The claim implies the lemma. Indeed we assumed $t \leq \tau$. And, if $x \in V_\Delta$, then the depth of $x$ in $V_\Delta$ cannot be smaller than $\tau$, for otherwise the invariant would imply $x \in X_\Delta$, contradicting the fact that $(t, \Delta, s, x)$ is a candidate event. All there remains to do is to prove the claim.

To do so consider the triangle $\Delta'$ of $T$ and the side $s'$ of $\Delta'$ such that $s$ is matched to $s'$. Consider the orientation-preserving isometry of $\mathbb R^2$ that maps $s$ to $s'$ and puts $\Delta$ side-by-side with $\Delta'$, apply this isometry to $x$, and consider the resulting point $x' \in \mathbb R^2$. Using the assumption that $(t, \Delta, s, x)$ is a candidate event, the point $x'$ belongs to $X_{\Delta'}$, while $x$ does not belong to $X_\Delta$. Moreover there is a point $z'$ along $s'$ such that $x'$ is at distance $t$ from $z'$, and such that no point of $X_{\Delta'}$ is closer to $z'$ than $x'$. Consider the immersed empty disk $(D', \varphi')$ attached to $\Delta'$ such that the center of $D'$ is $z'$, and such that the radius of $D'$ is maximum.

By contradiction, assume that the radius of $D'$ is smaller than $t$. There is a point $v \in \varphi'^{-1}(V)$ because the radius of $D'$ is maximum. We have $v \in V_{\Delta'}$, and the depth of $v$ in $V_{\Delta'}$ is smaller than or equal to the radius of $D'$, which is smaller than $\tau$. So $v$ belongs to $X_{\Delta'}$ by the invariant. But then $v$ is a point of $X_{\Delta'}$ closer to $z'$ than $x'$, a contradiction.

We proved that the radius of $D'$ is greater than or equal to $t$. Then $x'$ belongs to $\overline D$. Moreover $x'$ belongs to $X_{\Delta'}$, and thus to $V_{\Delta'}$ by the invariant. Therefore $x'$ belongs to $\varphi'^{-1}(V)$ by Lemma~\ref{lem: coherence disk}. In particular the radius of $D'$ is \emph{equal} to $t$.

It is now convenient to name the orientation-preserving isometry of $\mathbb R^2$ that maps $s$ to $s'$ and puts $\Delta$ side-by-side with $\Delta'$, so let $\lambda : \mathbb R^2 \to \mathbb R^2$ be this isometry. Consider the point $z = \lambda^{-1}(z')$, and the immersed empty disk $(D, \varphi)$ attached to $\Delta$ such that the center of $D$ is $z$, and such that the radius of $D$ is maximum. Observe that $\lambda(D) = D'$, and that $\varphi' \circ \lambda = \varphi$. In particular the radius of $D$ is also $t$. And, crucially, $x \in \varphi^{-1}(V)$, because we already proved $x' \in \varphi'^{-1}(V)$. This proves that $x \in V_\Delta$. And the depth of $x$ in $V_\Delta$ is smaller than or equal to the radius of $D$, which is $t$. The claim is proved, along with the lemma.
\end{proof}

\begin{lemma}\label{lem:exists candidate event}
Assume that the invariant holds for some $\tau > 0$. Further assume that there is a triangle $\Delta$ of $T$ such that $V_\Delta \setminus X_\Delta$ contains a point whose depth in $V_\Delta$ is $\tau$. Then there is a candidate event whose date is smaller than or equal to $\tau$.
\end{lemma}

The proof of Lemma~\ref{lem:exists candidate event} relies on the following:

\begin{lemma}\label{lem:depth smaller}
Let $(D, \varphi)$ be an immersed empty disk attached to a triangle $\Delta$ of $T$. Assume that there is $x \in \varphi^{-1}(V)$, and let $y$ be the center of $D$. If the geodesic segment between $x$ and $y$ intersect $\Delta$ in any other point than $y$ then the depth of $x$ in $V_\Delta$ is smaller than the radius of $D$.
\end{lemma}

\begin{proof}
Assuming that the geodesic segment between $x$ and $y$ intersects $\Delta$ in a point $y' \neq y$ (at least), consider the open disk $D'$ whose center is $y'$ and whose boundary circle contains $x'$. Then $D' \subset D$. Let $\varphi'$ be the restriction of $\varphi$ to $\overline D'$. Then $(D', \varphi')$ is an immersed empty disk, $\varphi'$ agrees with $\rho$ on $\Delta \cap \overline D'$, and $x \in \varphi'^{-1}(V)$. So the depth of $x$ in $V_\Delta$ is smaller than or equal to the radius of $D'$, which is smaller than the radius of $D$.
\end{proof}

\begin{proof}[Proof of Lemma~\ref{lem:exists candidate event}]
Consider a point $x \in V_\Delta \setminus X_\Delta$ that has depth $\tau$ in $V_\Delta$. There is an immersed empty disk $(D, \varphi)$ that satisfies each of the following. Let $y$ be the center of $D$. Then $y$ belongs to $\Delta$, the radius of $D$ is $\tau$, $\varphi$ agrees with $\rho$ on $\overline D \cap \Delta$, and $x \in \varphi^{-1}(V)$. In top of that we can add that $y$ belongs to the boundary of $\Delta$, for otherwise the depth of $x$ in $V_\Delta$ would be smaller than $\tau$ by Lemma~\ref{lem:depth smaller}, a contradiction. There are two cases: either $y$ lies in the relative interior of a side of $\Delta$, or $y$ is a vertex of $\Delta$. 

\paragraph{First case}
 
First consider the case where $y$ lies in the relative interior of a side $s$ of $\Delta$. In this case we shall prove that there is $t \leq \tau$ such that $(t, \Delta, s, x)$ is a candidate event. The surface $\mathcal S(T)$ being closed, there are a triangle $\Delta'$ of $T$, and a side $s'$ of $\Delta'$, such that $s$ is matched to $s'$. Consider the orientation-preserving isometry $\lambda : \mathbb R^2 \to \mathbb R^2$ that maps $s$ to $s'$ and puts $\Delta$ side-by-side with $\Delta'$. We consider the points $x' = \lambda(x)$ and $y' = \lambda(y)$, the open disk $D' = \lambda(D)$, and the map $\varphi' = \varphi \circ \lambda^{-1}$. Observe that $(D', \varphi')$ is an immersed empty disk attached to $\Delta'$, that the center of $D'$ is $y'$, and that $x'$ belongs to the boundary circle of $D'$. Informally, $x'$, $y'$, and $(D', \varphi')$ correspond to $x$, $y$, and $(D, \varphi)$, but in the reference frame of $\Delta'$.

We claim that $x'$ belongs to $X_{\Delta'}$. To prove the claim consider the geodesic line $L$ supported by $s$, and direct $L$ so that $\Delta$ is on the \emph{right} of $L$. Similarly, consider the geodesic line $L' = \lambda(L)$. Then $L'$ is supported by $s'$, and $\Delta'$ is on the \emph{left} of $L'$. We have that $x$ lies strictly on the left of $L$, for otherwise the depth of $x$ in $V_\Delta$ would be smaller than $\tau$ by Lemma~\ref{lem:depth smaller}, a contradiction. So $x'$ lies (strictly) on the left of $L'$. And so the depth of $x'$ in $V_{\Delta'}$ is smaller than $\tau$ by Lemma~\ref{lem:depth smaller}. Therefore $x' \in X_{\Delta'}$ by the invariant. The claim is proved.

We use the claim immediately, $x'$ belongs to $X_{\Delta'}$. No point of $X_{\Delta'}$ is closer to $y'$ than $x'$, because $X_{\Delta'} \subseteq V_{\Delta'}$ by the invariant, and because $D' \cap V_{\Delta'} = \emptyset$ by Lemma~\ref{lem: coherence disk}. So $\text{Vor}(x', X_{\Delta'})$ intersects $s'$ (at least in $y'$), and its intersection with $s'$ is at distance a distance $t \leq \tau$ from $x'$ (because the distance between $y'$ and $x'$ is $\tau$). The tuple $(t, \Delta, s, x)$ is a candidate event. We are done in this case. 

\paragraph*{Second case.}

Now consider the case where $y$ is a vertex of $\Delta$. Then $\rho(y)$ is a vertex of the graph $T^1$ embedded on $\mathcal S(T)$. Note also that $\rho(y)$ lies in the interior of $\mathcal S(T)$ because $\mathcal S(T)$ has no boundary. And $\rho(y)$ is flat as it does not belong to $V$. We assume that no face of $T^1$ appears twice around $y$, for this eases the reading, and the proof trivially extends to the general case. Consider the $k \geq 3$ faces of $T^1$ incident to $\rho(y)$, in order around $\rho(y)$ (clockwise say, but counter-clockwise would do too), and the corresponding triangles $\Delta_0, \dots, \Delta_{k-1}$ of $T$, with $\Delta_0 = \Delta$. We fix $\Delta_0$, and we place copies of the triangles $\Delta_1, \dots, \Delta_{k-1}$ around $y$, in order. This is possibly because $\rho(y)$ is flat. For each $i$ we record the orientation-preserving isometry $\lambda_i : \mathbb R^2 \to \mathbb R^2$ that maps the copy of $\Delta_i$ around $y$ to the original triangle $\Delta_i$. We consider the points $x_i = \lambda_i(x)$ and $y_i = \lambda_i(y)$. Also we consider the open disk $D_i = \lambda_i(D)$ and the map $\varphi_i = \varphi \circ \lambda_i^{-1}$. Observe that $(D_i, \varphi_i)$ is an immersed empty disk attached to $\Delta_i$, that the center of $D_i$ is $y_i$, and that $x_i$ belongs to the boundary circle of $D_i$. Informally, $x_i$, $y_i$, and $(D_i, \varphi_i)$ correspond to $x$, $y$, and $(D, \varphi)$, but in the reference frame of $\Delta_i$.
 
We claim that there is $i$ such that $x_i \in X_{\Delta_i}$. Indeed there is $i$ such that the geodesic segment between $y$ and $x$ intersects the copied triangle $\lambda_i^{-1}(\Delta_i)$ in another point than $y$. Then the geodesic segment between $y_i$ and $x_i$ intersects $\Delta_i$ in another point than $y$. So $x_i$ belongs to $V_{\Delta_i}$ and has depth smaller than $\tau$ in $V_{\Delta_i}$, by Lemma~\ref{lem:depth smaller} applied to $\Delta_i$, $(D_i, \varphi_i)$, $x_i$, and $y_i$. And so $x_i \in X_{\Delta_i}$ by the invariant. The claim is proved.

Using the claim immediately, and because $x_0 \notin X_{\Delta_0}$, there is $i$ such that $x_i \in X_{\Delta_i}$ and $x_{i+1} \notin X_{\Delta_{i+1}}$, indices are modulo $k$. Consider the side $s_i$ of $\Delta_i$ that is matched to a side of $\Delta_{i+1}$. We shall prove that there is $t \leq \tau$ such that $(\tau, \Delta_i, s_i, x_i)$ is a candidate event. To do so first observe that no point of $X_{\Delta_i}$ is closer to $y_i$ than $x_i$ because $X_{\Delta_i} \subseteq V_{\Delta_i}$ by the invariant, and because $D_i \cap V_{\Delta_i} = \emptyset$ by Lemma~\ref{lem: coherence disk}. So $\text{Vor}(x_i, X_{\Delta_i})$ intersects $s_i$ (at least in $y_i$), and its intersection with $s_i$ is at a distance $t \leq \tau$ from $x_i$ (because the distance between $y_i$ and $x_i$ is $\tau$). The tuple $(t, \Delta_i, s_i, x)$ is a candidate event. We are done in this case. The lemma is proved.
\end{proof}

\begin{proof}[Proposition~\ref{prop:wave algorithm terminates}]
First we prove that the invariant holds throughout the execution of the wave algorithm. To prove the claim first observe that the invariant holds after the initialization phase. Now assume that it holds at the beginning of an iteration of the loop, for some $\tau > 0$, and that there is a candidate event $(t, \Delta, s, x)$, of smallest date $t$. If every triangle $\Delta$ of $T$ satisfies $X_\Delta = V_\Delta$ then the invariant holds for every $\tau > 0$ anyway. Otherwise there are without loss of generality a triangle $\Delta$ and a point in $V_\Delta \setminus X_\Delta$ whose depth in $V_\Delta$ is $\tau$, so there is a candidate event whose date is smaller than or equal to $\tau$ by Lemma~\ref{lem:exists candidate event}. In any case $t \leq \tau$ holds without loss of generality. Then $t = \tau$, $x$ belongs to $V_\Delta$, and the depth of $x$ in $V_\Delta$ is equal to $\tau$ by Lemma~\ref{lem:min date event}. So, after adding $x$ to $X_\Delta$, the invariant still holds. This proves that the invariant holds throughout the execution of the wave algorithm.

The wave algorithm never adds twice the same point in a set $X_\Delta$ of a triangle $\Delta$ of $T$. Moreover $X_\Delta \subseteq V_\Delta$ by the invariant. And the cardinality of $V_\Delta$ is $O(nh)$ by Lemma~\ref{L:VF bound}. So the wave algorithm terminates after $O(n^2h)$ iterations. The algorithm does not stop until $X_\Delta = V_\Delta$ for every triangle $\Delta$ of $T$, by Lemma~\ref{lem:exists candidate event}. And the sets $(V_\Delta)_\Delta$ are as desired by Lemma~\ref{lem:vp is what we want}. The lemma is proved.
\end{proof}

\subsubsection{Proof of Proposition~\ref{prop:maintain candidate events}}

In this section we prove Proposition~\ref{prop:maintain candidate events}, that during the wave algorithm the list of the candidate events sorted by date can be maintained in amortized $O(\log(nh))$ time per insertion of a point in the set $X_\Delta$ of a triangle $\Delta$.

The crux of the matter is to maintain the intersection of a Voronoi diagram in $\mathbb R^2$ with a closed segment of $\mathbb R^2$ in a dynamic manner while adding the sources one-by-one to the Voronoi diagram. To do that we consider a game that we play with Alice. Informally, Alice sends us the sources of the Voronoi diagram one-by-one, and we tell her what is changed after each insertion of a source. Formally, Alice initially sends us a closed segment $I$ of $\mathbb R^2$. Then Alice sends us $k \geq 1$ pairwise distinct points $z_1, \dots, z_k \in \mathbb R^2$ in this order. We do not know the points before they are sent to us by Alice, nor the number of points to be sent. For each $i \in [k]$, after the $i$-th point $z_i$ is sent to us by Alice, and before $i+1$-th point $z_{i+1}$ is sent to us, we must send two things to Alice. First, we must send the set $U_i \subseteq [i]$ containing the index $i$ together with the indices $j \in [i-1]$ such that $\text{Vor}(z_j,Z_i) \cap I \neq \text{Vor}(z_j, Z_{i-1}) \cap I$. Second, for each index $j \in U_i$, we must send the  (possibly empty) set $\text{Vor}(z_j, Z_i) \cap I$. Note that each set $\text{Vor}(z_j, Z_i) \cap I$ is a closed segment of $\mathbb R^2$, so it is either empty, a single point, or has two distinct endpoints by which it is uniquely determined. We have two lemmas:

\begin{lemma}\label{lem:cardinality voronoi update set}
The sum over $1 \leq i \leq k$ of the cardinality of the set $U_i$ is smaller than or equal to $5k$.
\end{lemma}

\begin{proof}
Consider $i \in [k]$. We claim that at most four indices $j \in U_i$ are such that $\text{Vor}(z_j, Z_i)$ intersects $I$. The claim immediately implies the lemma.

To prove the claim we consider the subset $Y$ of $I$ that contains the points that are strictly closer to $z_i$ than to any point of $Z_{i-1}$. And we consider the closure $X$ of $Y$. Then $X$ is a closed segment of $\mathbb R^2$ and, assuming that $U_i$ is not empty, we have that $Y$ is not empty, so $X$ is not empty, and $X$ is not a single point either. Informally, we now consider the two “ends” of $X$. Formally, we consider the two endpoints $x_0$ and $x_1$ of $X$, and for each $\varepsilon \in \{0,1\}$, we consider an arbitrarily short closed segment $X_\varepsilon \subset X$, not a single point, that contains $x_\varepsilon$. Provided $X_\varepsilon$ is short enough, there are no more than two indices $j \in U_i$ such that the relative interior of $X_\varepsilon$ is included in $\text{Vor}(z_j, Z_{i-1})$ . 

On the other hand if $j \in U_i$ is such that $\text{Vor}(z_j, Z_i)$ intersects $I$, then not only $\text{Vor}(z_j, Z_{i-1})$ also intersects $I$, but $\text{Vor}(z_j, Z_{i-1}) \cap I$ contains a point in $Y$ and a point not in $Y$, so it contains the relative interior of $X_{\varepsilon}$ for some $\varepsilon \in \{0,1\}$. This proves the claim, and the lemma.
\end{proof}

\begin{lemma}\label{L:amortized insertion}
There is an algorithm that receives $I$ and $z_1, \dots, z_k$ in this order, and that, after receiving $z_i$, $i \in [k]$, returns $U_i$ together with the closed segments $\text{Vor}(z_j, Z_i) \cap I$ for all $j \in U_i$, and runs in $O(k \log k)$ total time. 
\end{lemma}

\begin{proof}
Consider $i \in [k]$, and assume that the point $z_i$ has just been sent by Alice. We must return to Alice. The crux of the matter is to have maintained at this point the list of tuples $(j, \text{Vor}(z_j, Z_{i-1}) \cap I)$ over $j \in [i-1]$, ordered by the position of $\text{Vor}(z_j, Z_{i-1}) \cap I$ along $I$ (for some direction of $I$, and resolving any ambiguity arbitrarily). Now we can use the list to answer Alice, and update the list, as follows. Given a tuple $(j, \text{Vor}(z_j, Z_{i-1}) \cap I)$ we can determine in constant time whether $j \in U_i$ by checking whether there is a point of $\text{Vor}(j, Z_{i-1}) \cap I$ that is strictly closer to $z_i$ than to $z_j$. If $j \notin U_i$, then either all the tuples $(j', \text{Vor}(z_{j'}, Z_{i-1}) \cap I)$ before $(j, \text{Vor}(z_j, Z_{i-1}) \cap I)$ in the list are such that $j' \notin U_i$, or all the tuples after $(j, \text{Vor}(z_j, Z_{i-1}) \cap I)$ are like that, and we can find out which case it is in constant time. So we can list by dichotomy the $k' \geq 0$ tuples $(j, \text{Vor}(z_j, Z_{i-1}) \cap I)$ such that $j \in U_i$ in $O(k' + \log k)$ time. For each such tuple $(j, \text{Vor}(z_j, Z_{i-1}) \cap I)$, we derive $\text{Vor}(z_j, Z_i) \cap I$ from $\text{Vor}(z_j, Z_{i-1})$, $z_j$, and $z_i$ in constant time. In the end we compute $\text{Vor}(z_i, Z_i) \cap I$ in $O(\log k)$ time by finding by dichotomy the first and last tuples $(j, \text{Vor}(z_j, Z_{i-1})$ such that $\text{Vor}(z_j, Z_{i-1}) \cap I$ contains a point that is at least as close to $z_i$ than to $z_j$, if any. This way we can return to Alice, and update the list of tuples, in $O(k' + \log k)$ total time. Lemma~\ref{lem:cardinality voronoi update set} concludes.
\end{proof}

In the following, when maintaining the list of candidate events, we also maintain appropriate search trees in which we store the candidate events, so that the candidate events can be accessed by date or position in logarithmic time.

\begin{proof}[Proof of Proposition~\ref{prop:maintain candidate events}]
When inserting a point $x$ in the set $X_\Delta$ of a triangle $\Delta$, we maintain the list of candidate events sorted by date as follows. First, we find the candidate events of the form $(\cdot, \Delta, \cdot, x)$, and we remove these candidate events from the list. All but $O(\log k)$ of the time spent here is amortized by the fact that every event deleted here was created earlier in the execution of the algorithm. 

Second, for every side $s$ of $\Delta$, we consider the triangle $\Delta'$ and the side $s'$ of $\Delta'$ such that $s$ is matched to $s$, along with the orientation-preserving isometry $\lambda : \mathbb R^2 \to \mathbb R^2$ that maps $s$ to $s'$ and puts $\lambda(\Delta)$ and $\Delta'$ side by side. Among the candidate events of the form $(\cdot, \Delta', s', \lambda(y))$, $y \in X_\Delta$, those for which $\text{Vor}(y, X_\Delta \cup \{x\}) \cap s \neq \text{Vor}(y, X_\Delta) \cap s$ may have to updated. If $\text{Vor}(y, X_\Delta \cup \{x\}) \cap s = \emptyset $, then the event must be deleted. Otherwise, only the date of the event may change. This is done in amortized $O(\log k)$ time using Lemma~\ref{L:amortized insertion}. 

Finally, Lemma~\ref{L:amortized insertion} also provides us with the set $\text{Vor}(x, X_\Delta \cup \{x\}) \cap s$. If this set is not empty, and if $\lambda(x) \notin X_{\Delta'}$, then we consider the distance $t$ between $x$ and $\text{Vor}(x, X_\Delta \cup \{x\}) \cap s$, and we create the event $(t, s', \Delta', \lambda(x))$, in $O(\log k)$ time.
\end{proof}


\section{Appendix: applications of Theorem~\ref{thm:main result}}\label{app:applications}

Theorem~\ref{thm:main result} has two interesting applications. First, recall that Delaunay triangulations have bounded happiness~\cite[Section~4]{portalgons}. Combined with Theorem~\ref{thm:main result} we obtain:

\begin{corollary}\label{cor:happy}
Let $T$ be a portalgon of $n$ triangles, of aspect ratio $r$, whose surface $\mathcal S(T)$ is closed. One can compute in $O(n^3 \log^2(n) \cdot \log^4 (r))$ time a portalgon $T'$ of $O(n)$ triangles, whose surface is $\mathcal S(T)$, and whose happiness is $O(1)$.
\end{corollary}

\begin{proof}[Proof of Corollary~\ref{cor:happy}]
Apply Theorem~\ref{thm:main result} to compute the portalgon $T'$ of the Delaunay tessellation of $\mathcal S(T)$ in $O(n^3 \log^2(n) \cdot \log^4 (r))$ time. Some polygons of $T'$ may not be triangles. Cut the polygons of $T'$ that are not triangles (if any) along vertex-to-vertex arcs to obtain a triangular portalgon $T''$. Then $T''$ is the portalgon of a Delaunay triangulation of $\mathcal S(T)$, so $T''$ has bounded happiness by the result of L{\"o}ffler, Ophelders, Silveira, and Staals~\cite[Section~4]{portalgons}. Moreover the vertex set of its 1-skeleton $T''^1$ is exactly the set of singularities of $\mathcal S(T)$, except if $\mathcal S(T)$ is a flat torus in which case $T''^1$ has exactly one vertex, so in any case $T''$ has $O(n)$ triangles.
\end{proof}

On the portalgon $T'$ returned by Corollary~\ref{cor:happy} the single-source shortest path algorithm of L{\"o}ffler, Ophelders, Silveira, and Staals~\cite[Section~3]{portalgons} would run in time $O(n^2 \log^{O(1)}(n))$, so that:

\begin{observation}
On the portalgon $T'$, one can compute a shortest path between two given points in time $O(n^2 \log^{O(1)}(n))$.
\end{observation}

Second, Theorem~\ref{thm:main result} enables us to test whether the surfaces of two given portalgons are isometric, simply by computing and comparing the portalgons of the associated Delaunay tessellations:

\begin{corollary}\label{cor:isometry}
Let $T$ and $T'$ be portalgons of less than $n$ triangles, whose aspect ratios are smaller than $r$, and whose surfaces $\mathcal S(T)$ and $\mathcal S(T')$ are closed. One can determine whether $\mathcal S(T)$ and $\mathcal S(T')$ are isometric in $O(n^3 \log^2(n) \cdot \log^4 (r))$ time.
\end{corollary}

\begin{proof}
Theorem~\ref{thm:main result} computes the portalgons $\mathcal T$ and $\mathcal T'$ of the Delaunay tessellations of respectively $\mathcal S(T)$ and $\mathcal S(T')$  in $O(n^3 \log^2(n) \cdot \log^4 (r)$ time. We claim that we can determine whether $\mathcal T$ and $\mathcal T'$ are equal in $O(n^2)$ time. The claim immediately implies the corollary. 

Let us prove the claim. We consider the sides of the polygons of $\mathcal T$ and $\mathcal T'$. There are $O(n)$ such sides. Fix a side $s$ of a polygon of $\mathcal T$. For every side $s'$ of a polygon of $\mathcal T'$, determine in $O(n)$ time whether there exists a one-to-one correspondence $\varphi$ from the sides of the polygons of $\mathcal T$ to the sides of the polygons of $\mathcal T'$ that maps $s$ to $s'$, the boundary closed walks of the polygons of $\mathcal T$ to the boundary closed walks of the polygons of $\mathcal T'$, and the matching of $\mathcal T$ to the matching of $\mathcal T'$. If $\varphi$ exists then $\varphi$ is unique since $\mathcal S(\mathcal T)$ and $\mathcal S(\mathcal T')$ are connected: construct $\varphi$ in $O(n)$ time. Then determine in $O(n)$ time if for every polygon $P$ of $\mathcal T$ there is an orientation-preserving isometry $\tau_P : \mathbb R^2 \to \mathbb R^2$ such that $\varphi(s) = \tau_P(s)$ for every side $s$ of $P$. In which case return correctly that $\mathcal T$ and $\mathcal T'$ are equal. In the end, if every polygon side $s'$ of $\mathcal T'$ has been looped upon, and if no equality has been found, return correctly that $\mathcal T$ and $\mathcal T'$ are distinct. This proves the claim, and the corollary.
\end{proof}

\end{document}